\documentclass[DIV12]{scrartcl}

\usepackage[utf8]{inputenc}
\usepackage[T1]{fontenc}
\usepackage{verbatim}
\usepackage{microtype}
\usepackage{caption}
\usepackage{tikz}
\usetikzlibrary{matrix}
\usepackage{stmaryrd}
\usepackage{xspace}
\usepackage[sort,compress,noadjust]{cite}

\usepackage{color}
\usepackage{url}

\usepackage{myproofs}

\title{Provenance Circuits\\for Trees and Treelike Instances\\(Extended Version)}

\author{
\begin{tabular}[t]{c}
Antoine Amarilli \\
{\normalfont T\'el\'ecom ParisTech; Institut Mines--T\'el\'ecom; CNRS LTCI} \\
{\normalfont antoine.amarilli@telecom-paristech.fr} \\[0.5em]
Pierre Bourhis \\
{\normalfont CNRS CRIStAL; Université Lille 1; INRIA Lille} \\
{\normalfont pierre.bourhis@univ-lille1.fr} \\[0.5em]
Pierre Senellart \\
{\normalfont Institut Mines-Télécom; Télécom ParisTech; CNRS LTCI} \\
{\normalfont National University of Singapore; CNRS IPAL} \\
{\normalfont pierre.senellart@telecom-paristech.fr} \\[0.5em]
\end{tabular}
}
\date{}

\usepackage{amssymb}
\usepackage{amsmath}

\newtheorem{theorem}{Theorem}[section]
\newtheorem{lemma}{Lemma}[section]
\newtheorem{proposition}{Proposition}[section]
\newtheorem{corollary}{Corollary}[section]

\newtheorem{example}{Example}[section]

\newtheorem{definition}{Definition}[section]

\usepackage{paralist}
\usepackage{graphicx}
\usepackage{tabularx}
\usepackage{booktabs}

\newcommand{\deft}[1]{\emph{#1}}

\makeatletter
\newcommand*{\defeq}{\mathrel{\rlap{%
  \raisebox{0.3ex}{$\m@th\cdot$}}%
  \raisebox{-0.3ex}{$\m@th\cdot$}}%
  =}
\makeatother

\renewcommand{\leq}{\leqslant}
\renewcommand{\geq}{\geqslant}
\renewcommand{\phi}{\varphi}

\newcommand{\prxmlclass}[1]{\mathsf{#1}}
\newcommand{\fie}{\prxmlclass{fie}}
\newcommand{\muxind}{\prxmlclass{mux, ind}}
\newcommand{\mux}{\prxmlclass{mux}}
\newcommand{\ind}{\prxmlclass{ind}}

\newcommand{\prxml}{\mathsf{PrXML}}

\DeclareMathOperator{\supp}{supp}
\DeclareMathOperator{\dom}{dom}

\DeclareMathOperator{\node}{node}

\newcommand{\neuter}[1]{\underline{#1}}
\newcommand{\lbl}[1]{\lblf(#1)}
\newcommand{\lblf}{\lambda}
\newcommand{\annotf}{\alpha}

\newcommand{\decode}[1]{\left\langle #1 \right\rangle}
\newcommand{\sems}[1]{\llbracket #1 \rrbracket}

\newcommand{\true}{1}
\newcommand{\false}{0}

\newcommand{\calB}{\mathcal{B}}

\newcommand{\calD}{\mathcal{D}}
\newcommand{\calE}{\mathcal{E}}

\newcommand{\calT}{\mathcal{T}}

\newcommand{\inp}{\mathsf{inp}}

\newcommand\restr[2]{{% we make the whole thing an ordinary symbol
  \kern-\nulldelimiterspace % automatically resize the bar with \right
  #1 % the function
  _{|#2} % this is the delimiter
  }}

\newcommand{\LCf}{\text{L}}
\newcommand{\RCf}{\text{R}}

\newcommand{\LC}{L}
\newcommand{\RC}{R}
\newcommand{\FC}{\mathit{FC}}
\newcommand{\NS}{\mathit{NS}}

\newcommand{\card}[1]{\left|#1\right|}

\let\oldPr\Pr
\renewcommand\Pr{\oldPr\nolimits}

\newcommand{\posbool}[1]{\mathrm{PosBool}[#1]}
\newcommand{\vars}[1]{\mathrm{Vars}(#1)}
\newcommand{\trunc}[2]{#1^{\leq #2}}

\newcommand{\cqneq}{$\mathrm{CQ}^{\neq}$\xspace}
\newcommand{\ucqneq}{$\mathrm{UCQ}^{\neq}$\xspace}

\newcommand{\is}[1]{\makebox[0pt]{\(\scriptstyle#1\)}}

\newcommand{\quot}[2]{#1/{#2}}

\renewcommand{\int}{\mathrm{int}}

\newcommand{\arity}[1]{\mathrm{arity}(#1)}
\newcommand{\aritysg}[1]{a_{#1}}

\newcommand{\pfun}{\pi}
\newcommand{\acpt}[2]{#2 \models #1}
\newcommand{\accepts}[2]{\acpt{#1}{#2}}
\newcommand{\naccepts}[2]{\card{\mathrm{aruns}(#1, #2)}}

\newcommand{\p}{\mathrm{p}}
\renewcommand{\o}{\mathrm{o}}
\newcommand{\n}{\mathrm{n}}
\renewcommand{\i}{\mathrm{i}}

\renewcommand{\r}{\mathrm{r}}

\newcommand{\h}{\mathrm{h}}
\newcommand{\s}{\mathrm{s}}
\renewcommand{\t}{\mathrm{t}}

\newcommand{\boolp}[1]{\overline{#1}}
\newcommand{\natp}[2]{\overline{#1}^{#2}}

\newcommand{\goal}{\mathrm{Goal}}
\newcommand{\NN}{\mathbb{N}}

\newcommand{\prov}[2]{\mathrm{Prov}(#1, #2)}
\newcommand{\provn}[2]{\mathrm{Prov}_{\Nx}(#1, #2)}
\newcommand{\provnr}[3]{\mathrm{Prov}_{\Nx}(#1, #2, #3)}
\newcommand{\alphab}[2]{\Gamma^{#1}_{#2}}

\newcommand{\Nx}{\mathbb{N}[X]}

\newcommand{\pvalr}[3]{\mathrm{Val}^{#2}_{#3}(#1)}
\newcommand{\all}{\mathrm{all}}
\newcommand{\pval}[2]{\mathrm{Val}^{#2}(#1)}
\newcommand{\val}[1]{\mathrm{Val}(#1)}

\newcommand{\circuit}{\mathrm{Circuit}}

\newcommand{\nsum}{\bigoplus}
\newcommand{\nprod}{\bigotimes}
\renewcommand{\nplus}{\oplus}
\newcommand{\ntimes}{\otimes}

\newcommand{\ours}{\mathrm{ours}}
\newcommand{\theirs}{\mathrm{theirs}}

\newcommand{\tval}[2]{#1_{[#2]}}
\newcommand{\teval}[1]{\tevalf(#1)}
\newcommand{\tevalf}{\epsilon}

\newcommand{\pierre}[1]{}
\newcommand{\antoine}[1]{}
\newcommand{\todo}[1]{}
\newcommand{\discuss}[1]{}

\newcommand{\width}{\mathrm{w}}
\newcommand{\weight}{\mathrm{w}}

\begin{document}

\maketitle

\defaultbibliographystyle{abbrv}

\begin{bibunit}
\begin{abstract}
  Query evaluation in monadic second-order logic (MSO) is 
tractable on trees and treelike instances, even though it is hard for
arbitrary instances. This tractability result has been extended to several tasks
related to query evaluation, such as counting query
results~\cite{arnborg1991easy} or performing query evaluation on probabilistic
trees~\cite{cohen2009running}. These are two
examples of the more general problem of computing augmented query
output, that is referred to as \emph{provenance}. This article presents a
provenance framework for trees and treelike instances, by describing a
linear-time construction of a circuit provenance representation for MSO
queries. We show how this provenance can be connected to the usual
definitions of semiring provenance on relational
instances~\cite{green2007provenance}, even though we compute it in an
unusual way, using tree automata; we do so via intrinsic definitions of
provenance for general semirings, independent of the
operational details of query evaluation. We show applications of
this provenance to capture existing counting and probabilistic results on
trees and treelike instances, and give novel consequences for probability
evaluation.

\end{abstract}

\section{Introduction}\label{sec:intro}
A celebrated result by Courcelle~\cite{Courcelle90} has shown that
evaluating a fixed monadic second-order (MSO) query on relational
instances, while generally hard in the input instance
for any level of the polynomial hierarchy~\cite{AjtaiFS98},
can be performed in linear time on
input instances of \emph{bounded treewidth} (or \emph{treelike}
instances), by encoding the query to an automaton on
tree encodings of instances.
This idea has been extended more recently
to monadic Datalog~\cite{gottlob2010monadic}.
In addition to query evaluation,
it is also possible to \emph{count} in linear time the number of query
answers over treelike instances~\cite{arnborg1991easy,pichler2010counting}.

However, query evaluation and counting are
special cases of the more general problem of capturing \emph{provenance
information} \cite{CheneyCT09,green2007provenance} 
of query results, which describes the link between input and
output tuples.
Provenance information can be expressed through various formalisms, such as
\emph{provenance semirings}~\cite{green2007provenance}
or Boolean formulae~\cite{suciu2011probabilistic}.
Besides counting, provenance can be exploited 
for practically important tasks such as answering queries in
incomplete databases~\cite{ImielinskiL84}, 
maintaining access rights~\cite{ParkNS12},
or computing query probability~\cite{suciu2011probabilistic}.
To our knowledge,
no previous work has looked at the general question of efficient
evaluation of
expressive queries on treelike instances while keeping track of
provenance.

Indeed, no proper definition of provenance
for queries evaluated via tree automata has been put forward.
The first contribution of this work
(Section~\ref{sec:provenance})
is thus to introduce a
general notion of \emph{provenance circuit}~\cite{deutch2014circuits} for tree automata, which
provides an efficiently computable representation of all possible results of an
automaton over a tree with uncertain annotations. Of course, we are interested
in the provenance of \emph{queries} rather than automata;
however, 
in this 
setting, the provenance that we compute has an intrinsic
definition, so it does not depend on which automaton we use to compute the query.

We then extend these results in Section~\ref{sec:encodings} to the provenance of
queries on
treelike relational instances. We propose again an intrinsic definition
of provenance capturing the subinstances that satisfy the query. We then show
that, in the same way that queries can be evaluated by compiling them to an
automaton on tree encodings, we can compute a provenance circuit for the query
by compiling it to an automaton, computing a tree decomposition of the instance, and performing the previous construction, in
linear time overall in the input instance.
Our intrinsic definition of provenance ensures the provenance only
depends on the logical query, not on the choice of query 
plan, of automaton, or of tree decomposition. 

Our next contribution in Section~\ref{sec:semirings} is to extend
such definitions of provenance from Boolean formulae to
$\mathbb{N}[X]$, the \emph{universal provenance
semiring}~\cite{green2007provenance}. This poses several challenges. First, as
semirings cannot deal satisfactorily with
negation~\cite{geerts2010database,amsterdamer2011limitations}, we must restrict to \emph{monotone}
queries, to obtain monotone provenance circuits. Second, we must keep track of
the multiplicity of facts, as well as the multiplicity of matches. For this reason, we restrict to unions of conjunctive queries (UCQ) in
that section, as richer languages do not directly provide notions
of multiplicity for matched facts. We generalize our notion of provenance
circuits for automata to instances with unknown multiplicity annotations,
using arithmetic circuits. We show that, for UCQs, the
standard provenance for the universal semiring~\cite{green2007provenance}
matches the one defined via the automaton, and that a provenance circuit for it
can be computed in linear time for
treelike instances.

Returning to the non-monotone Boolean provenance, we show in Section~\ref{sec:applications} how the tractability of
provenance computation on treelike instances implies that of two important
problems: determining the probability of a query, and counting
query matches. We show that probability evaluation of
fixed MSO queries is tractable on probabilistic XML models with local uncertainty, a result
already known in~\cite{cohen2009running}, and extend it to trees with event
annotations that satisfy a condition of having bounded \emph{scopes}. We
also show that MSO query evaluation is tractable on
treelike block-independent-disjoint
(BID) relational instances~\cite{suciu2011probabilistic}.
These tractability results for
provenance are achieved by applying message
passing~\cite{lauritzen1988local} on our
provenance circuits.
Last, we show the tractability of
counting query matches, using a reduction to the probabilistic
setting, capturing a result of~\cite{arnborg1991easy}.

\section{Preliminaries}\label{sec:prelim}
We introduce basic notions related to trees, tree automata, and Boolean
circuits.

Given a fixed \deft{alphabet} $\Gamma$, we define a \deft{$\Gamma$-tree} $T
= (V, \LC, \RC, \lblf)$
as a set of \deft{nodes}~$V$, two partial mappings $\LC,\RC:V\to V$
that associate an internal node with its left and right child, and a
\deft{labeling function} $\lblf:V \to \Gamma$.
Unless stated otherwise, the trees that we consider are rooted, directed, ordered, binary, and full (each node
has either zero or two children). We write $n \in T$ to mean $n \in V$.
We say that two trees $T_1$ and $T_2$ are \deft{isomorphic} if there is a
bijection between their node sets preserving children and labels (we
simply write it $T_1=T_2$);
they have \deft{same skeleton} if
they are isomorphic except for labels.

A \deft{bottom-up nondeterministic tree automaton}
on $\Gamma$-trees,
or \deft{$\Gamma$-bNTA},
is a tuple
$A = (Q, F, \iota, \delta)$
of a set $Q$ of \deft{states},
a subset $F \subseteq Q$ of \deft{accepting states},
an \deft{initial relation} $\iota : \Gamma \rightarrow 2^Q$
giving possible states for leaves from their label,
and a \deft{transition relation} $\delta : Q^2 \times \Gamma \rightarrow 2^Q$
determining possible states for internal nodes
from their label and the states of their children.
A~\deft{run} of~$A$ on a $\Gamma$-tree $T=(V,\LC,\RC,\lblf)$ is a
function $\rho: V \to Q$
such that for each leaf~$n$ we have
$\rho(n) \in \iota(\lbl{n})$,
and for every internal node~$n$
we have $\rho(n) \in \delta(\rho(\LC(n)), \rho(\RC(n)), \lbl{n})$.
A run is \deft{accepting} if, for the root $n_{\mathrm{r}}$ of $T$,
$\rho(n_{\mathrm{r}}) \in F$; and $A$ \deft{accepts} $T$ (written $\acpt{A}{T}$)
if there is some accepting run of $A$ on $T$.
Tree automata capture usual query languages on trees, such as
MSO~\cite{thatcher1968generalized}
and tree-pattern queries~\cite{Neven02}.

A \deft{Boolean circuit} is a directed acyclic graph $C = (G, W,
g_0,\mu)$ where $G$ is a set of \deft{gates}, $W \subseteq G \times G$ is a set
of \deft{wires} (edges), $g_0 \in G$ is a
distinguished \deft{output gate}, and
$\mu$ associates each \deft{gate} $g \in G$ with a \deft{type} $\mu(g)$ that
can be $\inp$ (\deft{input gate}, with no incoming wire in~$W$), $\neg$
(NOT-gate, with exactly one incoming wire in~$W$),
$\land$ (AND-gate) or $\lor$ (OR-gate).
A \deft{valuation} of the \deft{input gates} $C_\inp$ of~$C$ is a function 
$\nu : C_\inp \to \{\false, \true\}$; it defines inductively a unique
\deft{evaluation} $\nu':C \to \{\false, \true\}$ as follows:
$\nu'(g)$ is $\nu(g)$ if $g \in C_\inp$ (i.e., $\mu(g) = \inp$); it is $\neg
\nu'(g')$
if $\mu(g) = \neg$ (with $(g', g) \in W$);
otherwise it is $\bigodot_{(g', g) \in W} \nu'(g')$
where $\odot$ is $\mu(g)$ (hence, $\land$ or $\lor$). Note that this
implies that AND- and OR-gates with no inputs always evaluate to $\true$ and
$\false$ respectively.
We will abuse notation and use valuations and evaluations
interchangeably, and we write $\nu(C)$ to mean $\nu(g_0)$.
The \deft{function} captured by~$C$ is 
the one that maps any valuation $\nu$ of $C_\inp$ to $\nu(C)$.

\mysec{Provenance Circuits for Tree Automata}{sec:provenance}
We start by studying a notion of provenance on trees, defined in an
uncertain tree framework. 
Fixing a finite alphabet $\Gamma$ throughout this section, we view a
$\Gamma$-tree~$T$ as an \deft{uncertain tree}, where each node carries an unknown
Boolean annotation in $\{\false, \true\}$, and consider all possible
\deft{valuations} that choose an annotation for each node of $T$, calling
$\boolp{\Gamma}$ the alphabet of annotated trees:

\begin{definition}
  We write $\boolp{\Gamma} \defeq \Gamma \times \{0,1\}$.
  For any $\Gamma$-tree $T=(V,\LC,\RC,\lblf)$ and valuation $\nu: V \to \{0,1\}$, 
  $\nu(T)$ is the $\boolp{\Gamma}$-tree with same skeleton where each
  node~$n$
  is given the label $(\lblf(n), \nu(n))$.
\end{definition}

We consider automata on \emph{annotated} trees, namely,
$\boolp{\Gamma}$-bNTAs, and define their \deft{provenance} on a $\Gamma$-tree
$T$ as a Boolean function that describes which valuations of~$T$ are accepted by
the automaton. Intuitively, provenance keeps track of the
dependence between Boolean annotations and acceptance or rejection of
the tree.

\begin{definition}
  \label{def:prov}
  The \deft{provenance} 
  of a $\boolp{\Gamma}$-bNTA~$A$ on a
  $\Gamma$-tree $T=(V,\LC,\RC,\lblf)$ is the function $\prov{A}{T}$ mapping any valuation $\nu: V \to
  \{\false, \true\}$ to $\true$ or $\false$ depending on whether
  $\accepts{A}{\nu(T)}$.
\end{definition}

We now define a \deft{provenance circuit} of $A$ on a $\Gamma$-tree $T$ as a circuit that
captures the provenance of~$A$ on~$T$, $\prov{A}{T}$. Formally:

\begin{definition}
  \label{def:provcirc}
  Let $A$ be a 
  $\boolp{\Gamma}$-bNTA and $T = (V, \LC, \RC, \lblf)$ be
  a $\Gamma$-tree. A \deft{provenance circuit} of $A$ on $T$ is a
  Boolean
  circuit~$C$ with $C_\inp = V$ that captures the function $\prov{A}{T}$.
\end{definition}

An important result is that provenance circuits can be
tractably constructed:

\begin{propositionrep}{prp:provenance-circuits}
  A provenance circuit of a 
  $\boolp{\Gamma}$-bNTA~$A$ on a
  $\Gamma$-tree~$T$ can be constructed in time 
$O(\card{A}\cdot\card{T})$.
\end{propositionrep}

\begin{toappendix}
  Throughout the appendix, we will call \deft{0-gates} (resp.\ \deft{1-gates})
  OR-gates (resp.\ AND-gates) with no inputs; they always evaluate to $0$
  (resp.\ to $1$).

  We first prove this result in the specific case of a \emph{monotone}
  $\boolp{\Gamma}$-bNTA~$A$ (see Definition~\ref{def:monotone}), showing the same result but for provenance circuits
  which are taken to be \emph{monotone} Boolean circuits (i.e., they do not feature
  NOT-gates).

  \begin{propositionrep}{prp:provenance-circuits-mono}
    A monotone provenance circuit~$C$ of a monotone $\boolp{\Gamma}$-bNTA~$A$ on a
    $\Gamma$-tree~$T$ can be constructed in time 
  $O(\card{A}\cdot\card{T})$.
  \end{propositionrep}

\begin{proof}
  Fix $T = (V, \LC, \RC, \lblf)$, $A = (Q, F, \iota, \delta)$, and
  construct the provenance circuit $C = (G, W, g_0, \mu)$.
  For each node $n$ of $T$, create one input gate $g^{\i}_n$ in $C$ (which we
  identify
  to $n$, so that we have $C_\inp = V$), and create one gate $g^q_n$ for every $q \in
  Q$. If $n$ is a leaf node, for $q \in Q$, set $g^q_n$ to be:
  \begin{compactitem}
  \item if $q \in \iota(\lbl{n}, \false)$, a $\true$-gate;
  \item if $q \in \iota(\lbl{n}, \true)$ but $q \notin \iota(\lbl{n}, \false)$,
    an OR-gate with sole input $g^\i_n$;
  \item if $q \notin \iota(\lbl{n}, \true)$, a $\false$-gate.
  \end{compactitem}

  If $n$ is an internal node, create gates $g^{q_{\LCf},q_{\RCf}}_n$ and
  $g^{q_{\LCf},q_{\RCf},\i}_n$ for every pair $q_{\LCf},q_{\RCf} \in Q$
  (that appears as input states of a transition of $\delta$), the
  first one being an AND-gate of $g^{q_{\LCf}}_{\LC(n)}$ and
  $g^{q_{\RCf}}_{\RC(n)}$, the second one being an AND-gate of
  $g^{q_{\LCf},q_{\RCf}}_n$ and of $g^\i_n$. Now, for $q \in Q$, set $g^q_n$ to
  be an OR-gate of all the $g^{q_{\LCf},q_{\RCf}}_n$ such that $q \in
  \delta(q_{\LCf},q_{\RCf}, (\lbl{n}, \false))$ and of all the
  $g^{q_{\LCf},q_{\RCf},\i}_n$ such that $q \in
  \delta(q_{\LCf},q_{\RCf}, (\lbl{n}, \true))$.

  Add gate $g_0$ to be an OR-gate of all the $g^q_r$ such that $q \in F$, where
  $r$ is the root of~$T$.

  \medskip

  This construction is in time $O(\card{A}\cdot\card{T})$: more precisely, for
  every node of the tree $T$, we create a number of states that is linear
  in the number of states in $Q$ and in the number of transitions of
  $\delta$.
  
  Now we show that $C$ is indeed a
  provenance circuit of $A$ on $T$. Let $\nu: V \to \{0,1\}$ be a valuation that we
  extend to an evaluation of $C$. We
  show by induction on $n \in T$ that for any $q \in Q$, we have $\nu(g^q_n) = \true$
  iff, letting $T_n$ be the subtree of~$T$ rooted at $n$, there is a run $\rho$
  of $A$ on $T_n$ such that $\rho(n) = q$.

  For a leaf node $n$, choosing $q \in Q$, if $\nu(n) = \false$ then
  $\nu(g^q_n)=1$ iff $q \in
  \iota(\lbl{n}, \false)$, and if $\nu(n) = \true$ then $\nu(g^q_n)=\true$ iff $q
  \in \iota(\lbl{n}, \true)$, so in both cases we can define a run $\rho$ as
  $\rho(n) \defeq q$. Conversely, the existence of a run clearly ensures that
  $\nu(g^q_n) = \true$.

  For an internal node $n$, choosing $q \in Q$, if $\nu(n) = \false$ then
  $\nu(g^q_n)=\true$ iff there are some $q_{\LCf}, q_{\RCf} \in Q$ such
  that $q \in \delta(q_{\LCf}, q_{\RCf}, (\lbl{n}, \false))$,
  $\nu(g^{q_{\LCf}}_{\LCf(n)})=\true$, and
  $\nu(g^{q_{\RCf}}_{\RCf(n)})=\true$. By induction hypothesis this
  implies the existence of a run $\rho_{\LCf}$ of $A$ on~$T_{\LC(n)}$ such
  that $\rho_{\LCf}(\LCf(n)) = q_{\LCf}$ and a run $\rho_{\RCf}$ of $A$ on~$T_{\RC(n)}$ such that $\rho_{\RCf}(\RCf(n)) = q_{\RCf}$,
  from which we construct a run $\rho$ of $A$ on $T_n$ such that
  $\rho(n) = q$, by setting $\rho(n) \defeq q$ and setting $\rho(n')$ either to~$\rho_{\LCf}(n')$ or to~$\rho_{\RCf}(n')$ depending on whether $n' \in
  T_{\LC(n)}$ or $n' \in T_{\RC(n)}$. Conversely, the existence of such a run
  $\rho$ implies the existence of two such runs $\rho_{\LCf}$ and~$\rho_{\RCf}$,
  from which we deduce that $\nu(g^q_n)=\true$.

  If $\nu(n) = \true$ then $\nu(g^q_n)=\true$ iff there are some $q_{\LCf},
  q_{\RCf} \in Q$ such that
  $\nu(g^{q_{\LCf}}_{\LCf(n)})=\true$,
  $\nu(g^{q_{\RCf}}_{\RCf(n)})=\true$, and
  either $q \in \delta(q_{\LCf}, q_{\RCf}, (\lbl{n}, \false))$ or $q \in
  \delta(q_{\LCf}, q_{\RCf}, (\lbl{n}, \true))$. By monotonicity of $A$, this
  is equivalent to $q \in \delta(q_{\LCf}, q_{\RCf}, (\lbl{n}, \true))$. The rest is
  analogous to the previous case.

  The claim proven by induction clearly justifies that $C$ is a
  provenance circuit, as,
  applying it to the root of~$T$, we deduce that, for any valuation $\nu$, we
  have $\nu(C) = \true$ iff there is an accepting run of $A$ on $\nu(T)$.
\end{proof}

We now generalize this result to automata and provenance circuits which are not
necessarily monotone.
\medskip
\end{toappendix}

\begin{proof}
  We adapt the construction of Proposition~\ref{prp:provenance-circuits-mono}. The only
  difference is that we add, for every node $n \in T$, a gate $g^{\neg \i}_n$
  which is a NOT-gate of $g^\i$, and we modify the definition of the following
  nodes:
  \begin{itemize}
  \item for leaf nodes $n$, for any state $q$, we set $g^q_n$ to be an OR-gate
    of $g^\i$ if $q \in \iota(\lbl{n}, \true)$ (and a 0-gate otherwise), and
    $g^{\neg \i}$ if $q \in \iota(\lbl{n}, \false)$ (and a 0-gate otherwise).
  \item for internal nodes $n$, for every pair $q_{\LCf},q_{\RCf} \in Q$ that
    appears as input states of a transition of $\delta$, create a gate
    $g^{q_{\LCf},q_{\RCf},\neg \i}_n$ which is an AND-gate of
    $g^{q_{\LCf},q_{\RCf}}_n$ and of $g^{\neg \i}_n$. Now,
    for any state $q$, we set $g^q_n$ as before except that we use 
    $g^{q_{\LCf},q_{\RCf}}_n$ instead of $g^{q_{\LCf},q_{\RCf},\neg \i}_n$.
  \end{itemize}

  \medskip

  We show correctness as before, showing by induction on $n \in T$ that for any
  $q \in Q$, $\nu(g^q_n) = \true$ iff $A_q$ accepts $T_n$, where $A_q$ is
  obtained from $A$ by letting $q$ be the only final state. The property is
  clearly true on leaf nodes, and at internal nodes, if $\nu(n) = \false$ we
  have $\nu(g^q_n) = \true$ iff there exist $q_{\LCf}, q_{\RCf} \in Q$ such that
  $q \in \delta(q_{\LCf}, q_{\RCf}, (\lbl{n}, \false))$,
  $\nu(g^{q_{\LCf}}_{\LCf(n)})=\true$, which by induction hypothesis implies the
  existence of sub-runs on $T_{\LC(n)}$ and $T_{\RC(n)}$ that we combine as
  before. If $\nu(n) = \false$ we have $\nu(g^q_n) = \true$ iff there exist
  $q_{\LCf}, q_{\RCf} \in Q$ such that $q \in \delta(q_{\LCf}, q_{\RCf},
  (\lbl{n}, \true))$ (this time we cannot have $q \in \delta(q_{\LCf}, q_{\RCf},
  (\lbl{n}, \false))$) so we conclude in the same way. We conclude by justifying
  that $g_0$ is correctly defined, as before.
\end{proof}

The proof is by creating one gate in~$C$ per state of~$A$ per node of~$T$, and writing
out in~$C$ all possible transitions of~$A$ at each node~$n$ of~$T$,
depending on the input gate that indicates the annotation of~$n$.
In fact, we can show that~$C$ is treelike for fixed~$A$; we use
this 
in Section~\ref{sec:applications} to show the
tractability of
tree automaton evaluation on probabilistic XML trees from
$\prxml^\muxind$~\cite{kimelfeld2013probabilistic}.

It is not hard to see that this construction gives us a way to capture the
provenance of any \deft{query} on trees that can be expressed as an
automaton,
no matter the choice of automaton. A \deft{query} $q$ is any logical sentence on
$\boolp{\Gamma}$-trees which a $\boolp{\Gamma}$-tree $T$ can \deft{satisfy} (written $T \models
q$) or \deft{violate} ($T \not\models q$). 
An automaton $A_q$ \deft{tests} query $q$ if for any
$\boolp{\Gamma}$-tree~$T$, we have $\accepts{A_q}{T}$ iff $T \models q$. We define
$\prov{q}{T}$ for a $\Gamma$-tree~$T$
as in Definition~\ref{def:prov}, and run circuits for queries as in
Definition~\ref{def:provcirc}. It is immediate that
Proposition~\ref{prp:provenance-circuits} implies:

\begin{proposition}
  For any fixed query $q$ on $\boolp{\Gamma}$-trees for which we can
  compute an
  automaton
  $A_q$ that tests it,
  a provenance circuit of $q$ on a $\Gamma$-tree $T$ can be constructed in time
  $O(\card{T})$.
\end{proposition}

\noindent Note that provenance does not depend on
the automaton used to test the query.

\mysec{Provenance on Tree Encodings}{sec:encodings}
We lift the previous results to the setting of \deft{relational instances}.

A \deft{signature} $\sigma$ is a finite set of \deft{relation names} (e.g., $R$) with
associated \deft{arity} $\arity R \geq 1$.
Fixing a countable domain $\calD=\{a_k\mid k\geq 0\}$,
a \deft{relational instance}~$I$
over $\sigma$
(or $\sigma$-instance)
is a finite set~$I$ of \deft{ground facts} of the
form $R(\mathbf{a})$ with $R \in \sigma$,
where
$\mathbf{a}$ is a tuple of~$\arity{R}$ elements of~$\calD$. The \deft{active domain} $\dom(I)\subseteq\calD$
of~$I$ is
the finite set of elements of $\calD$ used in $I$.
Two instances $I$ and $I'$ are \deft{isomorphic} if
there is a bijection $\phi$ from $\dom(I)$ to $\dom(I')$ such that
$\phi(I)=I'$.
We say that an instance~$I'$ is a
\deft{subinstance}\footnote{Subinstances are not necessarily ``induced'' by a subset of the
domain, they can be arbitrary subsets of facts.}
of $I$, written $I' \subseteq
I$, if it is a
subset of the facts of~$I$, which implies $\dom(I') \subseteq
\dom(I)$.

A \deft{query} $q$ is a logical formula in (function-free) first- or second-order logic on
$\sigma$, without free second-order variables; a $\sigma$-instance $I$ can
\deft{satisfy} it ($I \models q$) or
\deft{violate} it ($I \not\models q$). For
simplicity, unless stated otherwise, we restrict to \deft{Boolean queries}, that is, queries with no
free variables, that are \deft{constant-free}.
This limitation is inessential for \deft{data complexity}, namely complexity for a fixed query:
we can handle non-Boolean queries by
building a provenance circuit for each possible output result (there are
polynomially many), and we encode constants by extending the signature with
fresh unary predicates for them.

As before, we consider unknown Boolean annotations
on the facts of an instance. However, rather than annotating the facts, it is
more natural to say that a fact annotated by~$\true$ is kept, and a fact
annotated by~$\false$ is deleted. Formally, given an instance $\sigma$, a
\deft{valuation} $\nu$ is a function from the facts of~$I$ to $\{0, 1\}$, and we
define $\nu(I)$ as the subinstance $\{F \in I \mid \nu(F) = 1\}$ of~$I$. We then
define:

\begin{definition}
  \label{def:provenance-query}
  The \deft{provenance} of a
  query $q$ on a $\sigma$-instance $I$ is
  the function $\prov{q}{I}$ mapping any valuation $\nu: I \to \{0, 1\}$ to $1$
  or $0$ depending on whether $\nu(I) \models q$. A \deft{provenance circuit} of
  $q$ on $I$ is a 
  Boolean circuit $C$ with $C_\inp = I$ that
  captures $\prov{q}{I}$.
\end{definition}

We study provenance for treelike instances (i.e., bounded-treewidth
instances), encoding queries to automata on tree encodings. Let us first define this.
The \deft{treewidth} $\width(I)$ of an instance~$I$ is a standard
measure~\cite{robertson1986graph} of how close $I$ is to a tree: the
treewidth of a tree is~$1$, that of a cycle is~$2$, and that of a
$k$-clique or $k$-grid is $k-1$; further, we have $\width(I') \leq \width(I)$ for any $I' \subseteq I$.
It is
known~\cite{Courcelle90,flum2002query}
that for any fixed $k \in \NN$, there is a finite alphabet $\alphab{k}{\sigma}$ such
that any $\sigma$-instance $I$ of treewidth~$\leq k$ can be encoded in linear
time~\cite{bodlaender1996linear} to a $\alphab{k}{\sigma}$-tree $T_I$, called
the \deft{tree encoding}, which can be decoded back to $I$ up to isomorphism
(i.e., up to the identity of constants). Each fact in $I$ is encoded in a
node for this fact in the tree encoding, where the node
label describes the fact.

The point of tree encodings is that
queries in \deft{monadic second-order logic}, the extension of first-order logic
with second-order quantification on sets, can be encoded to automata
which are then evaluated on tree encodings. Formally:

\begin{definition}
  \label{def:tests}
  For $k \in \NN$, we say that a $\alphab{k}{\sigma}$-bNTA $A^k_q$ \deft{tests}
  a query $q$ for treewidth~$k$ if,
  for any $\alphab{k}{\sigma}$-tree~$T$, we have $\accepts{A^k_q}{T}$ iff $T$ decodes to
  an instance $I$ such that $I \models q$.
\end{definition}

\begin{toappendix}
  \subsection{Formal preliminaries}

  We start by giving the omitted formal definitions:

\begin{definition}
A \deft{tree decomposition} of an
instance $I$ is a $\calT$-tree $T = (B, \LC, \RC, \dom)$
where $\calT$ is the set of subsets of $\dom(I)$.
The nodes of~$T$ are called \deft{bags} and their label is written $\dom(b)$.
We require:
\begin{enumerate}
  \item for every $a \in \dom(I)$, letting $B_a \defeq \{b \in
    B \mid a \in \dom(b)\}$, for every two bags $b_1, b_2\in B_a$,
all bags on the (unique) undirected
    path from $b_1$ to $b_2$ are also in~$B_a$;
  \item for every fact $R(\mathbf{a})$ of $I$, there exists a bag
    $b_{\mathbf{a}} \in B$ such that
    $\mathbf{a} \subseteq \dom(b_{\mathbf{a}})$.
\end{enumerate}
The \deft{width} of~$T$ is $\width(T) \defeq k-1$ where $k \defeq \max_{b \in T}
\card{\dom(b)}$.
The \deft{treewidth} (or \deft{width}) of an instance $I$, written $\width(I)$, is
the minimal width $\width(T)$ of a tree decomposition $T$ of $I$.
\end{definition}

It is NP-hard, given an instance~$I$, to determine $\width(I)$. However,
given a fixed width~$k$, one can compute in linear time in~$I$ a tree
decomposition of width $\leq k$ of $I$
if one exists~\cite{bodlaender1996linear}.

To represent bounded-treewidth instances as trees on a
finite alphabet, we introduce the notion of \deft{tree encodings}. The
representation is up to isomorphism, i.e., it loses the identity of constants.
Our finite alphabet $\alphab{k}{\sigma}$ is the set of possible facts on
an instance of width fixed to~$k$;
following the definition of proof trees in~\cite{chaudhuri1992equivalence}
we use element co-occurrences between one node and its parent in the tree as a way to encode
element reuse. Formally, we take $\alphab{k}{\sigma}$ to be the set
defined as follows:

\begin{definition}
  \label{def:kfact}
  The set of \deft{$k$-facts} of the signature $\sigma$, written
  $\alphab{k}{\sigma}$, is the set of pairs $\tau = (d,s)$ where:
  \begin{compactitem}
  \item the \deft{domain} $d$ is a subset of
  size at most~$k+1$ of the first $2k+2$ elements of the countable domain~$\calD$,
  $a_1, \ldots,
  a_{2k+2}$;
\item the \deft{structure} $s$ is a zero- or single-fact structure over~$\sigma$ such that
  $\dom(s)\subseteq d$.
  \end{compactitem}
\end{definition}

A tree encoding is just a $\alphab{k}{\sigma}$-tree.
We first explain how such a tree encoding $E$ can be decoded to a structure $I =
\decode{E}$ (defined up to isomorphism) and
a tree decomposition $T$ of width $k$ of $I$.
Process $E$ top-down. At each $(d,s)$-labeled node of $E$ that is child
of a $(d',s')$-labeled node,
pick fresh elements in $\calD$ for the elements of $d \backslash d'$
(at the root, pick all fresh elements), add the fact of $s$
to~$I$
(replacing the elements in $d$ by the fresh elements, and by the old elements of
$\dom(I)$ for $d \cap d'$), and add a bag to $T$ with the elements
of $I$ matching those in $d$. If we ever attempt to create a fact that already
exists, we abort and set $\decode{E} \defeq \bot$ (we say that $E$ is
\deft{invalid}).

We can now define tree encodings in terms of decoding:
\begin{definition}
  \label{def:encoding}
  We call a $\alphab{k}{\sigma}$-tree $T$ a \deft{tree encoding} of width~$k$ of a
  $\sigma$-structure $I$ if $\decode{T}$ is isomorphic to $I$.
\end{definition}

We note that clearly if a structure~$I$ has a tree encoding of width $k$, then
$\width(I) \leq k$. Further, observe that there is a clear injective function
from $\decode{T}$ to $T$, which maps each fact of~$\decode{T}$ to the node
of~$T$ which encoded this fact. This function is not total, because some nodes
in~$T$ contain no fact (their $s$ is a zero-fact structure).

We now justify that one can efficiently compute a tree encoding of width~$k$
of~$I$ from a tree decomposition of width~$k$ of $I$ (this
result is implicit in~\cite{chaudhuri1992equivalence}).

\begin{lemmarep}{lem:encode}
  From a tree decomposition $T$ of width $k$ of a $\sigma$-structure~$I$, one
  can compute in linear time a tree encoding~$E$ of width~$k$ of~$I$ with a bijection from
  the facts of~$I$ to the non-empty nodes of $E$.
\end{lemmarep}

\begin{proof}
The intuition is that we assign each fact $R(\mathbf{a})$ of $I$ to a bag $b \in T$ such
that $\mathbf{a} \subseteq \dom(b)$, which can be done in linear
time~\cite{flum2002query}.
We then encode each node of~$T$ as a chain of nodes in
$E$, one for each fact assigned to $T$.

Fix the $\sigma$-structure $I$ and its tree decomposition $T$ of width $k$.
Informally, we build $E$ by walking through the decomposition~$T$ and
copying it by enumerating the new facts in the domain of each bag of~$T$ as a
chain of nodes in $E$, picking the labels in $\alphab{k}{\sigma}$ so that the elements
shared between a bag and its parent in $T$ are retained, and the new elements
are chosen so as not to overlap with the parent node. Overlaps between one node
and a non-parent or non-child node are irrelevant.

Formally, we proceed as follows.
We start by precomputing a mapping that indicates, for every
tuple $\mathbf{a}$ of~$I$ such that some fact $R(\mathbf{a})$ holds in $I$,
the topmost bag $\node(\mathbf{a})$ of~$T$ such that $\mathbf{a} \subseteq
\dom(\node(\mathbf{a}))$. This can be performed in linear time by Lemma~3.1
of~\cite{flum2002query}. Then, we label the tree decomposition $T$ with the
facts of~$I$ as follows:  for each fact $F = R(\mathbf{a})$ of~$I$, we add $F$
to the label of~$\node(\mathbf{a})$.

Now, to encode a bag $b$ of $T$, consider
$b_\p$ the
parent of~$b$ in~$T$ and partition $\dom(b) = d_\o \sqcup
d_\n$ where $d_\o$ are
the \emph{old elements} already present in $\dom(b_\p)$, and
$d_\n$ are the
\emph{new elements} that did not appear in $\dom(b_\p)$. (If $b$ is the root,
then $d_\o = \emptyset$ and $d_\n = \dom(b)$.)
Under a specific node $(d_\p,n_\p)$ in $E$, with a
bijection $f_\p$ from $\dom(b_\p)$ to
$d_\p$, choose a domain $d$ of size $\card{\dom(b)}$ over the
fixed $a_1, \ldots, a_{2k+2}$ whose intersection with
$f_\p(\dom(b_\p))$ is exactly
$f_\p(d_\o)$ (this is possible, as there are $2k+2$ elements to choose from and
$\card{\dom(b_\p)} \leq k+1$) and extend the bijection $f_p$ to $f$ so that it maps
$\dom(b)$ to $d$. At the root, choose an arbitrary bijection.
Now, encode $b$ as a chain of nodes in $E$ labeled with $(d,s_i)$ where
each $s_i$ encodes one of the facts in the label of $b$ (thus defining the
bijection from $I$ to the non-empty nodes of~$E$). If there are zero such facts,
create a $(d,\emptyset)$ zero-fact node instead, rather than creating no
node. Recursively encode the children of~$b$ (if any) in~$T$, under this
chain of nodes in~$E$. Add zero-fact $(\emptyset,\emptyset)$ child nodes so that each non-leaf
node has exactly two children. We assume that all arbitrary choices are
done in a consistent manner so that the process is deterministic.
\end{proof}

Hence, when restricting to instances whose width is bounded by a constant~$k$,
one can equivalently work with $\alphab{k}{\sigma}$-trees which are encoding of
these instances, instead of working with the instances themselves.

We redefine properly the notion of an automaton testing a query:

\begin{definition}
  \label{def:ftar}
  A $\alphab{k}{\sigma}$-bNTA $A$ \deft{tests} a Boolean query~$q$ for
  treewidth~$k$ if for any $\alphab{k}{\sigma}$-tree $E$, $\accepts{A}{E}$ iff
  $\decode{E} \models q$. (In particular, if $\decode{E} = \bot$ then $A$
  rejects $E$.)
\end{definition}

\subsection{Proof of Theorem~\ref{thm:courcelle}}

We can now state and prove the theorem that says that MSO sentences can be
tested by automata for any treewidth.
The main problem is to justify that MSO sentences can be rewritten to MSO
sentences on our tree encodings, as we can then use~\cite{thatcher1968generalized}
to compile them to a bNTA. We rely on~\cite{flum2002query} for that
result, but we must translate between our tree encodings and theirs. We do so
by a general technique of justifying that certain product trees annotated with
both encodings can be recognized by a bNTA. Let us state and prove the result:
\medskip
\end{toappendix}

\begin{theoremrep}[\cite{Courcelle90}]{thm:courcelle}
  For any $k \in \NN$, for any MSO query $q$, one can
  compute a $\alphab{k}{\sigma}$-bNTA $A^k_q$ that tests $q$ for
  treewidth $\leq k$.
\end{theoremrep}

\begin{proof}
  In the context of this proof, we define a \deft{bDTA} (bottom-up deterministic
tree automaton) as a bNTA but where $\iota$ and $\delta$, rather than returning
sets of reachable states, return a single state. This implies that the automaton
has a unique run on any tree.

Let us fix $k \in \mathbb{N}^*$.
We denote by $[m]$ the set $\{1,\ldots,m\}$
and by $n_{\sigma}$ the number of relations in~$\sigma$.
Lemma 4.10 of \cite{flum2002query} shows that for a certain finite alphabet
$\Gamma(\sigma,k)$,
for any MSO formula $\phi$ over the signature $\sigma$, there exists an  MSO formula $\phi^*$ such that
for any $\Gamma(\sigma,k)$-tree $t$ representing an instance $I$,
$t$ satisfies $\phi^*$ iff $I$ satisfies $\phi$. More precisely, one can define
a partial $\decode{\cdot}'$
function on
$\Gamma(\sigma,k)$-trees such that
for every instance $I$ of treewidth $\leq k$ there is a
$\Gamma(\sigma,k)$-tree $t$ such that $\decode{T}'$ is well-defined and isomorphic to $I$
and for every $\Gamma(\sigma,k)$-tree $T$, we have $T
\models \phi^*$ iff $\decode{T}'$ is well-defined and $\decode{T}' \models
\phi$.

We first describe the alphabet $\Gamma(\sigma,k)$.
A letter of $\Gamma(\sigma,k)$ is of the form $(\gamma_1, \allowbreak \gamma_2, \ldots,
\allowbreak \gamma_{n_{\sigma}+2})$. The element $\gamma_1$ belongs to
$2^{[k]^2}$ and
describes the equalities between the elements inside a bag; the element
$\gamma_2$ belongs to
$2^{[k]^2}$ and
describes the equalities between the
elements from this bag and its parents;
For $i \geq 3$, $\gamma_i$ belongs to
$2^{k^{\arity{R_i}}}$
and describes the tuples belonging to the relation $R_i$.
In the $\Gamma(\sigma,k)$-trees, the encoding of equalities between values of
the bags and its parents are described explicitly (by $\gamma_2$) rather than
implicitly (by element reuse between parent and child, as in our encoding).

We next describe for which $\Gamma(\sigma,k)$ trees $T$ is their encoding operation
$\decode{T}'$ well-defined, and how it is then computed. We say that $T$ is
\deft{well-formed} if $\decode{T}'$ is well-defined. We accordingly say that a
$\alphab{k}{\sigma}$-tree $T$ is \deft{well-formed} if $\decode{T}$ is
\deft{well-defined}, namely, different from $\bot$.

For a $\Gamma(\sigma,k)$-tree~$T$, $\decode{T}'$ is well-defined iff for each
node $n$ with $\mathbf{\gamma} \defeq \lbl{n}$ and each child $n' \in \{\LC(n), \RC(n)\}$
with $\mathbf{\gamma'} \defeq \lbl{n'}$:
\begin{itemize}
\item $\gamma_1$ is closed by transitive closure, i.e., if $(i,j)$ and $(j,e)$ belong to $\gamma_1$ then $(i,e)$ belongs to $\gamma_1$
\item $\gamma_2$ is closed by transitive closure (to check this, we need to
  consider paths, rather than the mere pair $n$ and $n'$)
    $\gamma_2$ is closed by transitive closure. 
\item for each $(i,j)$ in $\gamma_1$ and $(j,e)$ in $\gamma'_2$ then $(i,e)$
  belongs to $\gamma_2'$ (and symmetrically, reversing the roles of $n$ and
  $n'$)
\item for each pair $(j_1,\ldots,j_l)$ in\ $\gamma'_m$ and if for each $b$ $(i_b,j_b)$ in $\gamma_2'$, then $(i_1,\ldots,i_l)$ is in $\gamma_m$
(and vice-versa, reversing the roles of $n$ and $n'$); a similar condition holds
with $\gamma_1$ 
\end{itemize}
Note that these conditions are clearly expressible in MSO.
While~\cite{flum2002query} does not precisely describe the behavior of $\phi^*$
on $\Gamma(\sigma,k)$-trees which are not well-formed, the above justifies our
assumption that $\phi^*$ tests well-formedness and rejects the trees which are
not well-formed.

We now define $\decode{T}'$ as follows, if~$T$ is well-formed.
Process $E$ top-down. At each node $n \in E$ with $\mathbf{\gamma} \defeq
\lbl{n}$ with parent node $n' \in E$ with $\mathbf{\gamma'} \defeq \lbl{n'}$,
pick fresh elements in $\calD$ for the positions $j$ such that there is no pair
$(i,j)$ in $\gamma'_2$
(at the root, pick all fresh elements) and if $(j_1,j_2)$ belongs to $\gamma'_1$
then the same fresh element is assigned for the elements at both positions;
if there is such a pair, pick the existing elements used when decoding $n'$.
These choices define a mapping $\nu$ from the positions to the fresh elements
and to existing elements. Now,
for each $(j_1,\ldots,j_m)$ in $\gamma_i$, then the fact
$R_i(\nu(j_1),\ldots,\nu(j_m))$ is added to $I$.
 If we ever attempt to create a fact that already exists, we ignore it.

 We have reviewed the alphabet $\Gamma(\sigma,k)$ of~\cite{flum2002query}, the conditions for the
well-definedness of $\decode{T}'$ and the semantics of this operation. Now,
following 
\cite{thatcher1968generalized,flum2002query}, with an additional step to determinize the resulting
automaton to a bDTA, the formula
$\phi^*$ from~\cite{flum2002query} can be translated into a bDTA $A_{\theirs}$ on
$\Gamma(\sigma,k)$-trees
such that for any $\Gamma(\sigma,k)$-tree $T$, 
$A_{\theirs}$ accepts $T$
iff $T \models \phi^*$, that is, iff $\decode{T}'$ is well-defined and satisfies
$\phi$. Note that this implies that $A_{\theirs}$ is
encoding-invariant.
We now explain how to translate $A_{\theirs}$ to our desired bDTA $A_{\ours}$ over
$\alphab{k}{\sigma}$ such that for every $\alphab{k}{\sigma}$-tree $T$,
$A_{\ours}$
accepts $T$ iff $\decode{T} \models \phi$.

We consider the alphabet $\Sigma = \alphab{k}{\sigma} \times \Gamma(\sigma, k)$, and
call $\pi_1$ and $\pi_2$ the operations on $\Sigma$-trees that map them
respectively to $\alphab{k}{\sigma}$ and $\Gamma(\sigma,k)$ trees with same skeleton
by keeping the first or second component of the labels. Given a
$\alphab{k}{\sigma}$-tree $T_1$ and a $\Gamma(\sigma,k)$-tree $T_2$ with same
skeleton, we will write $T_1 \times T_2$ the $\Sigma$-tree obtained from them.

We will do this by building a $\Sigma$-bDTA $A_\t$ with the following properties:
\begin{enumerate}
  \item If $A_\t$ accepts $T$ then $\decode{\pi_1(T)}$ and $\decode{\pi_2(T)}'$ are
    well-defined and isomorphic.
  \item For every $\alphab{k}{\sigma}$-tree $T_1$ such that $\decode{T_1}$ is
    well-defined, there exists a $\Gamma(\sigma, k)$-tree $T_2$ such that $A_\t$
    accepts $T_1 \times T_2$.
\end{enumerate}
Then, we can notice that from $A_{\theirs}$, we can build an $\Sigma$-bDTA
$A'_{\theirs}$
such that $T$ is recognized by $A'_{\theirs}$ iff $\pi_2(T)$ is accepted by
$A'_{\theirs}$, and build $A_{\ours}$ as the conjunction of $A_\t$ and
$A'_{\theirs}$,
projected to the first component (accept a $\alphab{k}{\sigma}$-tree $T_1$ iff there
is some $\Gamma(\sigma,k)$-tree $T_2$ such that $T_1 \times T_2$ is accepted, which is
possible using non-determinism, and then determinizing). It is now clear that
$A_{\ours}$ thus defined accepts a $\alphab{k}{\sigma}$-tree $T_1$ iff the
$\Sigma$-tree $T_1 \times T_2$ is accepted, for some $\Gamma(\sigma,k)$-tree
$T_2$, by $A'_{\theirs}$ and $A_\t$: if this happens then $T_2$ is accepted by
$A_{\theirs}$ and $\decode{T_1}$ is isomorphic to $\decode{T_2}'$ so
$\decode{T_2} \models \phi$; and conversely, if $\decode{T_1}$ models $\phi$,
there is some $\Gamma(\sigma, k)$-tree $T_2$ such that $A_\t$ accepts $T_1 \times
T_2$, and this implies that $\decode{T_2}'$ is isomorphic to $\decode{T_1}$ so
(as $\phi$, being a constant-free MSO query, is invariant under isomorphisms)
$\decode{T_2}$ satisfies $\phi^*$ and $A'_{\theirs}$ accepts $T_1 \times T_2$.
So it suffices to build the $\Sigma$-bDTA $A_\t$ with the desired properties.

We now define a simple encoding from $\alphab{k}{\sigma}$ to $\Gamma(\sigma, k)$
describing what is the tree $T_2$, given a well-formed tree $T_1$, such that
$A_\t$
accepts $T_1 \times T_2$. It will then suffice to see that it is possible, with
a MSO formula $\psi$, to check on a $\Sigma$-tree $T$ whether $\pi_2(T)$ is the encoding of
$\pi_1(T)$ in this sense. Indeed, we can then compile $\psi$ to a
$\Gamma$-bDTA using~\cite{thatcher1968generalized}.

Consider a node $n_1 \in T_1$ and let $(d, s) \defeq \dom(n)$. We define the label
of the corresponding node $n_2 \in T_2$. We define $\gamma_1$ so that the $k+1 -
\dom(d)$ last elements are all equal to the $\dom(d)$-th element (i.e., we
complete $\dom(d)$ to always have $k+1$ elements, by ``repeating'' the last
element, where ``last'' is according to an arbitrary order on domain elements).
We define $\gamma_2$ to indicate which elements of $n_1$ were shared with its
parent node, completing it to be consistent with respect to $\gamma_1$. Last, we
define $\gamma_3, \ldots, \gamma_{n_{\sigma} + 2}$ to be the tuples of elements
in the various relations of $\sigma$ in $\decode{T_1}$, with repetitions to be
consistent according to $\gamma_1$.

It is clear that this encoding maps every tree $T_1$ such that $\decode{T_1}$ is
well-defined to a tree $T_2$ such that $\decode{T_2}$ is well-defined and
isomorphic to $\decode{T_1}$. Now, to justify the existence of $\psi$, observe
that the only non-local condition to check on $T$ is the definition of the
$\gamma_3, \ldots, \gamma_{n_{\sigma} + 2}$; but we can clearly define by an MSO
formula, for a node $n \in T$ with $(d, s) \defeq \lbl{n}$, the exact set of
facts stated by $T$ for the elements represented by $d$ (there are only a finite
number of such ``types''): they are defined to check, for all facts of the
putative type, whether a node with the right fact is reachable following an
undirected path where the same elements are kept along the path. So we can
define an MSO formula checking for each node $n \in T$ whether the type of
$\pi_1(n)$ in $\pi_1(T)$ in this sense matches the graph stated in
$\pi_2(n)$.

\end{proof}

Our results apply to any query language that can be rewritten to tree automata
under a bound on instance treewidth. Beyond MSO, this is also the case of
\deft{guarded second-order logic} (GSO). GSO
extends first-order logic with second-order quantification on arbitrary-arity relations, with a
semantic restriction to \emph{guarded tuples} (already co-occurring in some instance
fact); it captures MSO (it has the same expressive power on treelike
instances~\cite{gradel2002back}) and many common database query languages, e.g.,
\deft{guarded Datalog}~\cite{Gradel:2000:EEM:1765236.1765272} or
\deft{frontier-guarded Datalog}~\cite{baget2010walking}. We use GSO in the
sequel as our choice of query language that can be rewritten to automata.
Combining the result above with the results of the previous section, we claim that
provenance for GSO queries on treelike instances can be tractably computed, and
that the resulting provenance circuit has treewidth independent of the instance.

\begin{toappendix}
  \subsection{Proof of Theorem~\ref{thm:provenance-encodings}}
  We first give the formal definition of the
  treewidth of a circuit, which we omitted. To do so, we must first give a
  normal form for circuits:

\begin{definition}
  \label{def:aritytwo}
  Let $C = (G, W, g_0, \mu)$ be a Boolean circuit. The \deft{fan-in} of a gate
  $g \in G$ is the number of gates $g' \in G$, such that $(g', g) \in W$. Note
  that our definitions of circuits impose that the fan-in of input gates is
  always~$0$ and the fan-in of NOT-gates is always $1$.
  We say $C$ is \deft{arity-two} if the fan-in of AND- and OR-gates is
  always~$2$, where we allow constant $0$- and $1$-gates as gate types of their
  own (but require that such gates have fan-in of~$0$).
\end{definition}

  Clearly this restriction is inessential as circuits can be rewritten in
  linear-time to an arity-two circuit by merging AND- and OR-gates with fan-in
  of~$1$ with their one input, replacing those with fan-in of~$0$ by a 0- or
  1-gate, and rewriting those with fan-in $>2$ to a chain of gates of the same
  type with fan-in~$2$.

  We now define tree decompositions of circuits, and their relational encoding:

\begin{definition}
  \label{def:circuittw}
  The relational signature $\sigma_{\circuit}$ features one unary relation
  $R^\i$ which applies to input gates, two unary relations $R^0$ and $R^1$ which
  apply to constant 0- and 1-gates, one binary relation $R^{\neg}(g_\o, g_\i)$
  which applies to NOT-gates (the first element is the output and the second is
  the input), and two ternary relations $R^\wedge(g_\o, g_\i, g_\i')$ and
  $R^\vee(g_\o, g_\i, g_\i')$ which apply respectively to AND- and OR-gates,
  with the first element being the input and the second and third being the
  inputs. The \deft{relational encoding} of an (arity-two non-monotone) Boolean
  circuit $C$ is the $\sigma_{\circuit}$-instance $I_C$ obtained in the expected
  way; we can clearly construct $I_C$ from $C$ in linear time. The \emph{treewidth} $\width(C)$
  of~$C$ is $\width(I_C)$ (but we talk of tree decompositions of~$C$ as shorthand).
\end{definition}

  Now, to prove the theorem, we first take care of a technical issue. Given an instance $I$ and
  its tree encoding $T_I$, we want to construct provenance circuits on $T_I$, meaning
  that we wish to consider Boolean-annotated versions of $T_I$. However, the
  tree encodings of subinstances of $I$ are not annotated versions
  of $T_I$. We need to justify that they can be taken to have the same structure
  as $T_I$, with some facts having been removed in nodes containing no fact.

\begin{definition}
  \label{def:neuter}
  For any $k$-fact $\tau = (d, s) \in \alphab{k}{\sigma}$, we define the
  \deft{neutered $k$-fact} $\neuter{\tau}$ as $(d, \emptyset)$. In particular, if
  $s = \emptyset$ then $\neuter{\tau} = \tau$.
  For $\tau \in \alphab{k}{\sigma}$ and $b \in \{\false, \true\}$, we write
  $\tval{\tau}{b}$
  to be $\tau$ if $b$ is $\true$ and $\neuter{\tau}$ if $b$ is $\false$.

  Given a $\boolp{\alphab{k}{\sigma}}$-tree $E$, we define its
  \deft{evaluation} $\teval{E}$ as the $\alphab{k}{\sigma}$-tree that has
  same skeleton, where for every node $n \in E$ with corresponding node
  $n'$ in $\teval{E}$, letting $\lbl{n} = (\tau, b) \in \alphab{k}{\sigma}
  \times \{\false, \true\}$, we have $\lbl{n'} = \tval{\tau}{b}$.
\end{definition}

This definition allows us to lift $\alphab{k}{\sigma}$-bNTAs, intuitively
testing a query, to $\boolp{\alphab{k}{\sigma}}$-bNTAs that test the same query
on Boolean-annotated tree encodings, seen via~$\tevalf$ as the tree encoding of a
subinstance. Formally:

\begin{lemma}
  \label{lem:relabel}
  For any $\alphab{k}{\sigma}$-bNTA $A$, one can compute in linear time a
  $\boolp{\alphab{k}{\sigma}}$-bNTA $A'$ on such that $\accepts{A'}{E}$
  iff $\accepts{A}{\teval{E}}$.
\end{lemma}

\begin{proof}
  Let $A=(Q,F,\iota,\delta)$. We construct the bNTA $A'=(Q,F,\iota',\delta')$
  according to the following definition:
  $\iota'((\tau, b))\defeq\iota(\tval{\tau}{b})$ and $\delta'((\tau, b), q_1,
  q_2)\defeq\delta(\tval{\tau}{b}, q_1, q_2)$ for all $b \in \{\false, \true\}$,
  $\tau \in
  \alphab{k}{\sigma}$, and $q_1, q_2 \in Q$. The process is clearly in linear time
  in~$\card{A}$. Now, it is immediate that $\accepts{A'}{E}$ iff
  $\accepts{A}{\teval{E}}$, because a run of $A'$ on~$E$ is a run of $A$
  on~$\teval{E}$, and vice-versa.
\end{proof}

  We are now ready to state and prove the result:
  \medskip
\end{toappendix}

\begin{theoremrep}{thm:provenance-encodings}
  For any fixed $k \in \NN$ and 
  GSO query~$q$, for any
  $\sigma$-instance~$I$ such that $\width(I) \leq k$, one can construct a provenance
  circuit~$C$ of $q$ on $I$
  in time $O(\card{I})$.
  The treewidth of~$C$ only depends on~$k$ and $q$ (not on $I$).
\end{theoremrep}

\begin{proof}
  Fix $k \in \NN$ and the GSO query~$q$.
  Using Theorem~\ref{thm:courcelle}, let $A$ be a $\alphab{k}{\sigma}$-bNTA
  $A^k_q$ that tests $q$ for treewidth~$k$ (remember that
  Theorem~\ref{thm:courcelle} extends from MSO to GSO because both collapse on
  treelike instances~\cite{gradel2002back}). We lift $A$ to a bNTA $A'$ on
  $\boolp{\alphab{k}{\sigma}}$ using Lemma~\ref{lem:relabel}. This is performed
  in constant time in the instance.

  Now, given the input instance~$I$ such that $\width(I) \leq k$,
  compute in linear time~\cite{bodlaender1996linear} a tree
  decomposition of~$I$, and, from this, compute in linear time a tree encoding
  $E_I$ of $I$ using Lemma~\ref{lem:encode}. We now use
  Proposition~\ref{prp:provenance-circuits} to construct a provenance circuit $C$
  of $A'$ on $E_I$. Consider now the injective function $f$ that maps the facts of~$I$ to
  the nodes of $E_I$ where those facts are encoded. We modify $C$ to replace the
  input gate $g^{\i}_n$ for any $n \in E_I$ not in the image of $f$, setting it to be
  a $1$-gate; and renaming the input gates $g^{\i}_n$ for any $n \in E_I$ to be $F$,
  for $F$ the fact such that $f(F) = n$. Let $C'$ be the result of this process.
  $C'$ is thus a
  Boolean circuit such that $C'_\inp = I$, and it was
  computed in linear time from $I$. We claim that it captures $\prov{q}{I}$.

  \medskip

  To check that it does, let $\nu: I \to \{\false, \true\}$ be a valuation of
  $I$. We show that $\nu(C') = \true$ iff $\nu(I) \models q$. To do so, the key
  point is to
  observe that, letting $\nu'$ be the valuation of $E_I$ defined by $\nu'(n) =
  \nu(F)$ if there is $F \in I$ such that $f(F) = n$, and $\nu'(n) = \true$
  otherwise, we have that $\teval{\nu'(E_I)}$ is a tree encoding of $\nu(I)$.
  Indeed, $\teval{\nu'(E_I)}$ and $E_I$ have same skeleton, the elements
  that constitute the domains of node labels are the same, and a fact $F \in I$ is
  encoded in $\teval{\nu'(E_I)}$ iff $f(F)$ is annotated by $\true$ in
  $\nu'(E_I)$ iff we have $F \in \nu(I)$. (Note that our choice to extend $\nu'$
  by setting it to be $\true$ on nodes that encode no facts makes no difference,
  as the annotation of such nodes is projected to the same label by $\tevalf$,
  i.e., for such labels $\tau$ in $E_I$, we have $\tval{\tau}{\false} =
  \tval{\tau}{\true}$.)

  Having observed this, we know that, because $A$ tests $q$, $\nu(I) \models q$
  iff $\accepts{A}{\nu'(E_I)}$. Now, by definition of Lemma~\ref{lem:relabel}, we
  have $\accepts{A}{\teval{\nu'(E_I)}}$ iff $\accepts{A'}{\nu'(E_I)}$, which by
  definition of the provenance circuit $C$ is the case iff $\nu'(C) = \true$,
  which by definition of $C'$ is the case iff $\nu(C') = \true$. Hence, $C'$ is
  indeed a provenance circuit of $q$ on $I$.

  \medskip

  The only
  point left to justify is that the treewidth of the circuit $C$ is indeed bounded.
  Indeed, the number of gates that we create in $C$ for
  each node $n$ of $E$ only depends on the automaton $A$ that tests $q$,
  and wires in $C$ only go from gates for one node $n$ to gates for nodes
  $\LC(n)$ and $\RC(n)$, so that the tree decomposition $T$ for $C$
  is obtained by putting, in each bag of the tree decomposition
  corresponding to node $n$ of $E$, the gates for node $n$, and $\LC(n)$ and
  $\RC(n)$ if they exist. The additional distinguished gate $g_0$ is added to
  the bag of the root node of $E$. This construction is described on the circuit
  before translating to arity-two, but as the fan-in of the gates of the
  original circuit is bounded by a constant, clearly rewriting to arity-two
  preserves the property that the number of gates per bag of the decomposition
  is bounded by a constant. This proves the claim.
\end{proof}

The proof is by encoding the instance~$I$ to its tree encoding $T_I$ in linear
time, and compiling the query $q$ to an automaton $A_q$ that tests it, in constant time in
the instance. Now, Section~\ref{sec:provenance} worked with
$\boolp{\alphab{k}{\sigma}}$-bNTAs rather than $\alphab{k}{\sigma}$-bNTAs, but the
difference is inessential: we can easily map any
$\boolp{\alphab{k}{\sigma}}$-tree~$T$ to a $\alphab{k}{\sigma}$-tree
$\teval{T}$ where any node label $(\tau, 1)$ is replaced by $\tau$, and any
label $(\tau,
0)$ is replaced by a dummy label indicating the absence of a fact; and we
straightforwardly translate $A$ to a $\boolp{\alphab{k}{\sigma}}$-bNTA~$A'$ such
that $\accepts{A'}{T}$ iff $\accepts{A}{\teval{T}}$ for any
$\boolp{\alphab{k}{\sigma}}$-tree~$T$.
The key point is then that, for any valuation $\nu: T \to \{0, 1\}$,
$\teval{\nu(T)}$ is a tree encoding of~$\nu(I)$ (defined in the expected way),
so we
conclude by applying
Proposition~\ref{prp:provenance-circuits} to~$A'$ and $T$.
As in Section~\ref{sec:provenance},
our definition of provenance is intrinsic to the query and does not depend on
its formulation, on the choice of tree decomposition, or on the choice of
automaton to evaluate the query on tree encodings.

Note that tractability holds only in data complexity. For combined complexity, we incur the cost of compiling the query to an
automaton, which is nonelementary in general~\cite{meyer1975weak}. However, for some
restricted query classes, such as \deft{unions of conjunctive queries} (UCQs), the compilation phase has lower cost.

\begin{toappendix}
  \subsection{Proof of EXPTIME rewriting of UCQs}

  \begin{proposition}[\cite{chaudhuri1992equivalence}]
    \label{prp:ucq}
    For any UCQ $q$ and $k \in \NN$, a $\alphab{k}{\sigma}$-bNTA that tests $q$
    for treewidth $\leq k$ can be computed in EXPTIME in~$q$ and $k$.
  \end{proposition}

  \begin{proof}
    
Let $q$ be a UCQ and $k$ be an integer.

This proof relies on the notion of proof trees introduced in \cite{chaudhuri1992equivalence}.
The proof trees are intuitively tree encodings of an \deft{unfolding}, or
\deft{expansion tree}, of a Datalog
query $P$ (refer to Definition~\ref{def:datalog} for the definition of Datalog).
An \deft{expansion tree} of $P$ 
is a ranked tree (not binary in general) defined as follows:
the node labels are pairs of a fact $F$ from an intentional predicate of
$\sigma_\int$ and an instantiation of the body of a rule $r \in P$ (i.e., the variables are
mapped to elements of the instance in a way that satisfies the body of $r$)
such that the corresponding instantiation of the head of $r$ is $F$.

Such a tree is \deft{well-formed} if
for any node $n$ labeled by $(F,x)$
there is a bijection $f$ between the children of $n$ and
the intensional facts of the instantiation $x$ such that for any node $n$,
$f(n)$ is exactly the head fact of $n$. (In particular, if the same intensional
fact is used multiple times in the rule, then there are as many children as
there are occurrences of this fact).
We will require that in the rules of the query $P$, every body contains either
$0$ or $2$ intensional facts, so that expansion trees are full binary trees.

From an expansion tree, it is possible to derive a \deft{proof tree}, which is a
$\Sigma(P)$-tree for some finite set $\Sigma(P)$ (for fixed $P$), as follows:
the alphabet $\Sigma(P)$ is the set of pairs of tuples over some fixed set of $2
\card{P}$ values and of a rule of $P$, and the intuition of a $\Sigma(P)$-tree,
just like for our notion of $k$-facts, is that
sharing an element between one node and its parent encodes that it is the same
element, but elements shared between, e.g., siblings, are not necessarily the
same element. Note that proof trees, as expansion trees, are full binary trees.
In this proof we use proof trees to mean this,
as~\cite{chaudhuri1992equivalence}, and we do not mean the notion of proof tree
used for Datalog provenance in Definition~\ref{def:datalog}.

Having described how to encode an expansion tree to a proof tree, we describe
the \deft{decoding} $\decode{T}'$ of a $\Sigma(P)$-tree $T$: first, apply a
process analogous to our own notion of decoding, to obtain an expansion tree
$T'$; second, consider the extensional facts that appear in the instantiation of
the bodies in the labels of $T'$, and define $\decode{T}'$ to be the instance
formed of those facts. Of course, if any of these processes fails, or if the
intermediate expansion tree is not well-formed, we abort and set $\decode{T}' =
\bot$.

Our goal is now to define a Datalog query $P$ such that there is a surjective
homomorphism from
$\Sigma(P)$ to $\alphab{k}{\sigma}$. Fix $\sigma_\int$ to have intensional relations
$P_0, \ldots, P_{k+1}$ of arity $0, \ldots, k+1$.
(We technically disallowed
  predicates of arity 0 in our definition of instances, but there is clearly no
problem in this context.)
For every tuples of variables
$\mathbf{x}$, $\mathbf{y}$, $\mathbf{z_1}$, $\mathbf{z_2}$ taken from a set of
$3k + \aritysg{\sigma}$ variables denoted by $S_X$ (where $\aritysg{\sigma}$ is the arity of
$\sigma$), with the condition $\mathbf{y} \subseteq \mathbf{x}$,
for every relation $R$ of $\sigma$, and $0 \leq i,j_1,j_2 \leq k+1$, create the rules in $P$:
\[
          P_i(\mathbf{x}) \leftarrow R(\mathbf{y}) P_{j_1}(\mathbf{z_1}) P_{j_2}(\mathbf{z_2})
\]
and
\[
          P(\mathbf{x}) \leftarrow R(\mathbf{y}) 
\]
Finally, we create the rules
\[
          P_i(\mathbf{x}) \leftarrow  P_{j_1}(\mathbf{z_1}) P_{j_2}(\mathbf{z_2})
\]

In terms of size, each rule of the query $P$ contains a number of variables polynomial in $k$
and $\sigma$, and the overall size of the query is exponential in a polynomial
of $k$ and $\sigma$. Last, the size of $\Sigma(P)$ is in
$O(\card{P}\cdot\card{P}^{a(P)})$, where $a(P)$ is the maximal arity of the intentional relations.
Let $\beta$ be a set of values of cardinality equal to $2k+2$.
Then, $\Sigma(P)$ is equal to the pairs $P(\mathbf{a}),r$ where $r$ is a rule.
We define the following homomorphism $h$ from  $\Sigma(P)$ to
$\alphab{k}{\sigma}$.
Let $(P_i(\mathbf{a}), r)$ be a element of $\Sigma(P)$. If $r$ does not have an
extensional fact then $h(P_i(\mathbf{a}),r)$ is equal to $(\mathbf{a},\emptyset)$.
Otherwise, the atom $R(\mathbf{y})$ occurs in the body of $r$, let $\nu$ be the
valuation from the variables of $r$ defined according to the head atom
$P_i(\mathbf{a})$ (as we imposed $\mathbf{y} \subseteq \mathbf{x}$ above) such
that $\nu(\mathbf{x})$ is equal to $\mathbf{a}$: $h(P(\mathbf{a}),r)$ is equal to $(\mathbf{a},      R(\nu(\mathbf{y}))$.
$h$ is thus defined from $\Sigma(P)$-trees to the $\alphab{k}{\sigma}$, and it is
clearly surjective. Furthermore, it is clear that this application extends to a
surjective mapping $h'$ from $\Sigma(P)$-trees to $\alphab{k}{\sigma}$-trees, with the
property that whenever $\decode{h'(T)}$ is defined then $T$ is well-formed and
$\decode{h'(T)}$ and
$\decode{T}'$ are isomorphic.

We now explain how we construct our automaton for the query $q$. Let us first
assume that $q$ is a conjunctive query (CQ). We consider the Datalog query $P$ that we constructed
above. From the proof of Proposition~5.10 of
\cite{chaudhuri1992equivalence}, we deduce that we can construct, in time
polynomial in its size, a bNTA $A_P$ on
$\Sigma(P)$  whose number of states is in
is in $O(\card{\Sigma(P)} \cdot 2^{\card{q}+V_q* V_P})$,
where $V_P$ (resp., $V_q$) is the maximal number of variables in a rule of $P$
(resp., in $q$)
such that $A_P$ recognizes the language of the well-formed $\Sigma(P)$-trees $T$
such that $\decode{T}'$ satisfies $q$. For our query $P$, the size of $A_P$ is
therefore exponential in a polynomial of $k$, $\sigma$ and $\card{q}$.

Because $h'$ is an surjective homomorphism from $\Sigma(P)$-trees to
$\alphab{k}{\sigma}$-trees and $A_P$ is on $\Sigma(P)$
 with the Property 1.4.3 of \cite{tata} that shows that bNTA are closed by
 homomorphism, we compute in polynomial time in $A_P$
a bNTA $A_P'$ on $\alphab{k}{\sigma}$ that has
size exponential in a polynomial of $\sigma$, $\card{q}$ and $k$. We intersect
it with a bNTA (clearly constructible in EXPTIME) that checks whether a
$\alphab{k}{\sigma}$-tree is a valid encoding, and rejects otherwise. This yields the final automaton    $A$.

We now check that $A$ tests the query $q$. Let $T$ be a
$\alphab{k}{\sigma}$-tree. If $\decode{T}$ satisfies $q$, then it is well-defined,
Let $T'$ be a preimage of $T$ by $h'$. By our condition on $h'$, $\decode{T'}'$
is well-defined and isomorphic to $\decode{T}$, so (as $q$ features no constants
and is thus preserved by isomorphisms) it satisfies $q$, and
therefore $T'$ was accepted by $A_P$, so $T$ is accepted by $A$. Conversely, if
$A$ accepts $T$, then let $T'$ be a preimage of $T$ by $h$ such that $A'$
accepts $T'$.
As $\decode{T}$ is well-defined,
$T'$ is well-defined and $\decode{T'}'$ and
$\decode{T}$ are isomorphic; but as $T'$ is accepted by $A_P$, we must have
$\decode{T'}' \models q$, so $\decode{T} \models q$.

The result can be extended to an UCQ $q$ by applying the result to every CQ and
taking the union of the resulting automata (whose size is the sum of the input
automata).

  \end{proof}
\end{toappendix}

\mysec{General Semirings}{sec:semirings}
In this section we connect our previous results
to the existing definitions of \deft{semiring provenance} on arbitrary
relational instances~\cite{green2007provenance}:

\begin{definition}
  A \deft{commutative semiring} $(K, \oplus, \otimes, 0_K, 1_K)$ is a set $K$ with
binary operations $\oplus$ and~$\otimes$ and distinguished elements
$0_K$ and~$1_K$, such that $(K, \oplus)$ and $(K,\otimes)$ are
commutative monoids with identity
element $0_K$ and $1_K$,
$\otimes$ distributes over $\oplus$, and $0_K \otimes a = 0_K$ for all $a
\in K$.
\end{definition}

Provenance for semiring $K$ is defined on instances where each fact
is annotated with an element of~$K$. The provenance of a query on such an
instance is an element of $K$ obtained by combining fact annotations
following the semantics of the query, intuitively describing how the
query output depends on the annotations (see exact definitions in~\cite{green2007provenance}). This general setting has many specific
applications:

\begin{example}
  For any variable set~$X$, the monotone Boolean functions over~$X$ form a
  semiring
  $(\posbool{X},\lor,\land,\false,\true)$.
  We write them as
  propositional formulae, but two equivalent formulae
  (e.g., $x \vee x$ and $x$) denote the same $\posbool{X}$ object.
  On instances where each fact is annotated by its own variable in $X$,
  the $\posbool{X}$-provenance of a query~$q$ is a monotone Boolean function on~$X$
  describing which subinstances satisfy~$q$. As we will see,
  this
  is what we defined in Section~\ref{sec:encodings}, using circuits as
  compact representations.

  The \deft{natural numbers} $\mathbb{N}$ with the usual $+$ and $\times$
  form a
  semiring. On instances where facts are annotated with an element of~$\NN$
  representing a multiplicity,
  the provenance of a query describes its number of matches under the bag
  semantics.

  The \deft{security
  semiring}~\cite{foster2008annotated}
  $\mathbb{S}$ is defined on the ordered set $1_{\mathbb{S}} < C < S < T
  < 0_{\mathbb{S}}$ (respectively: \emph{always available},
  \emph{confidential},       \emph{secret}, \emph{top secret},
  \emph{never available}) as $(\{1, C, S, T, 0\}, \min, \max, 0, 1)$. The provenance of a query for $\mathbb{S}$ denotes the  minimal level of security clearance required to see that it holds.
  The \deft{fuzzy semiring} \cite{amsterdamer2011limitations}
  is $([0, 1], \max, \min, 0, 1)$. The provenance of a query for this semiring is the minimal fuzziness value that has to be tolerated
for facts so that the query is satisfied.

The \deft{tropical semiring} \cite{deutch2014circuits}
is $(\mathbb{N} \sqcup \{\infty\}, \min, +, \infty, 0)$. Fact annotations
are costs, and the tropical provenance of a
query is the minimal cost of the facts required to satisfy it, with multiple uses of a fact being charged multiple times.

  For any set of variables $X$, the \deft{polynomial semiring} $\NN[X]$ is the
  semiring of polynomials with variables in $X$ and coefficients in $\NN$, with
  the usual sum and product over polynomials, and with $0, 1 \in \NN$.
\end{example}

\begin{toappendix}
  
\begin{definition}
  \label{def:datalog}
  A \deft{Datalog query} $P$ over the signature $\sigma$ consists of a signature
  $\sigma_\int$ of \deft{intensional predicates} with a special $0$-ary relation
  $\goal$ and a finite set of \deft{rules} of the form $R(\mathbf{x}) \leftarrow
  R_1(\mathbf{y_1}),\ldots,R_k(\mathbf{y_k})$ where $R \in \sigma_\int$, $R_i
  \in \sigma \sqcup \sigma_\int$ for $1 \leq i \leq k$, and each variable in the
  tuple $\mathbf{x}$ also occurs in some tuple $\mathbf{y_i}$. The left-hand
  (resp., right-hand) side of a Datalog rule is called the \deft{head} (resp.,
  \deft{body}) of the rule.

  A \deft{proof tree} $T$ of a Datalog query $P$ over an instance $I$ is a
  (non-binary) ordered tree
  with nodes annotated by facts over $\sigma \cup \sigma_\int$ on elements of
  $\dom(I)$ and internal nodes annotated by rules, such that the fact of the root of $T$ is $\goal$,     and, for every
  internal node $n$ in $T$ with children $n_1, \ldots, n_m$, the indicated rule
  $R(\mathbf{x}) \leftarrow \Psi(\mathbf{y})$ on $n$ in $P$ is such that there is a
  homomorphism $h$ mapping $R(\mathbf{x})$ to the fact of $n$ and mapping
  $\Psi(\mathbf{y})$ to the facts of the $n_i$. (Note that this definition
  implies that internal nodes are necessarily annotated by a fact
  of~$\sigma_\int$.) We write $I \models P$ if $P$ has a proof tree on~$I$.
\end{definition}

The \deft{semiring provenance} of a Datalog query on an instance is defined as
follows:

\begin{definition}
  \label{def:dlprov}
  Given a semiring $K$ and an instance $I$ where each fact $F$ carries an
  annotation $\annotf(F) \in K$,
  the \deft{provenance} of a Datalog query $P$ on
  $I$ is the following~\cite{green2007provenance}:
  \[
    \bigoplus_{T\text{ proof tree of }P} \bigotimes_{n\text{ leaf of }T}
    \annotf(n).
  \]
  Note that this expression may not always be defined
  depending on the query, instance, and semiring. In particular, the number of
  terms in the sum may be infinite, so that the result cannot necessarily be
  represented in the semiring.
\end{definition}

We now use this to define Datalog queries associated to conjunctive queries
(CQs) and unions of CQs (UCQs), which
will be useful for provenance.

  \begin{definition}
    \label{def:ucqtodl}
    We assume that CQs and UCQs contain no equality atoms. The Datalog query
    $P_q$ \deft{associated} to a CQ $q$ has only one rule, $\goal
    \leftarrow q$. The Datalog query $P_q$ associated to a UCQ $q = \bigvee_i
    q_i$ has rules $\goal \leftarrow q_1$, ..., $\goal
    \leftarrow q_n$.
    
    Observe that in this case the provenance of $P_q$ in the
    sense of Definition~\ref{def:dlprov} is always defined, no matter the
    semiring, as the number of possible derivation trees is clearly finite.
  \end{definition}

\end{toappendix}

Semiring provenance does not support
negation well \cite{geerts2010database,
amsterdamer2011limitations} and is
therefore only defined for \emph{monotone} queries: 
a query $q$ is \deft{monotone} if, for any instances $I
\subseteq I'$, if $I \models q$ then $I' \models q$. Provenance
circuits for semiring provenance are
\deft{monotone} circuits~\cite{deutch2014circuits}:
they do not feature NOT-gates.
We can show that, adapting the
constructions of Section~\ref{sec:provenance} to work with a notion
of \emph{monotone} bNTAs, Theorem~\ref{thm:provenance-encodings} applied to
monotone queries yields a \deft{monotone} provenance circuit:

\begin{theoremrep}{thm:provenance-encodings-mono}
  For any fixed $k \in \NN$ and 
  monotone
  GSO query~$q$, for any
  $\sigma$-instance~$I$ such that $\width(I) \leq k$, one can construct
  in time $O(\card{I})$
  a monotone provenance
  circuit of~$q$ on~$I$
  whose treewidth
  only depends on~$k$ and $q$ (not on $I$).
\end{theoremrep}

\begin{toappendix}
  We first define our notion of monotonicity for $\boolp{\Gamma}$-bNTAs. Intuitively, it
implies that if a $\boolp{\Gamma}$-tree is accepted by the automaton, it will
not be rejected when changing annotations from $\false$ to $\true$.
Formally:

\begin{definition}
\label{def:monotone}
  We consider the partial order $<$ on $\boolp{\Gamma}$ defined by $(\tau, 0) <
  (\tau, 1)$ for all $\tau \in \Gamma$.
  We say that a $\boolp{\Gamma}$-bNTA $A = (Q, F,
\iota, \delta)$ is
\deft{monotone} if for every $\tau \leq \tau'$ in $\boolp{\Gamma}$, we
have $\iota(\tau) \subseteq
\iota(\tau')$ and $\delta(q_1, q_2, \tau) \subseteq \delta(q_1, q_2, \tau')$ for every
$q_1, q_2 \in Q$.
\end{definition}

It is easy to see that the provenance of a \emph{monotone}
$\boolp{\Gamma}$-bNTA~$A$ on any tree~$T$ is a
\deft{monotone} function in the following sense: for any valuations $\nu$ and $\nu'$, if $\nu(g) = \true$ implies $\nu'(g)
= \true$ for all $g \in C_\inp$ (which we write $\nu \leq \nu'$), then
$(\prov{A}{T})(\nu) = \true$ implies $(\prov{A}{T})(\nu') = \true$.
Indeed, for any monotone
$\Gamma$-tree $T$ and valuations $\nu \leq
\nu'$, $\accepts{A}{\nu(T)}$ implies $\accepts{A}{\nu'(T)}$.

We already know that the results of Section~\ref{sec:provenance} also hold for
monotone bNTAs and monotone provenance circuits (see
Proposition~\ref{prp:provenance-circuits-mono}). Hence, if we know that the
monotone GSO sentence~$q$ of interest can be tested (Definition~\ref{def:ftar})
by a \emph{monotone} bNTA, i.e., if the analogue of Theorem~\ref{thm:courcelle}
holds for monotone queries and bNTAs, then clearly the reduction from treelike
instances to trees (Theorem~\ref{thm:provenance-encodings}) generalizes. 
\todo{$E \leq E'$ doesn't seem defined?}
We first prove an auxiliary lemma:

\begin{lemma}
    \label{lem:monencode}
    For every $\boolp{\alphab{k}{\sigma}}$-trees $E$ and $E'$, if $E \leq E'$
    then $\decode{\teval{E}}
    \subseteq \decode{\teval{E'}}$.
  \end{lemma}
  
  \begin{proof}
    We follow the decoding process and notice that, as the domains of the
    $\alphab{k}{\sigma}$ node labels in $E$ and $E'$ are the same, the same fresh
    elements are used throughout, so the only difference between
    $\decode{\teval{E}}$ and
    $\decode{\teval{E'}}$ is about the annotation of the created facts; and we notice
    that whenever $E \leq E'$ then every fact created in $E$ is also created in $E'$.
  \end{proof}

We now prove the analogue of Theorem~\ref{thm:courcelle} for monotone queries
and automata. We immediately generalize the notion of an automaton testing a
query (Definition~\ref{def:ftar}) to $\boolp{\alphab{k}{\sigma}}$-bNTAs and
trees, with $\decode{T}$ for a $\boolp{\alphab{k}{\sigma}}$-tree~$T$ defined as
above.

\begin{lemma}
  \label{lem:monexists}
  Let $k \in \mathbb{N}$, $q$ be a monotone query, and $A$ be a
  $\boolp{\alphab{k}{\sigma}}$-bNTA
  that tests $q$ for treewidth $k$. One can compute in linear time from $A$ a
  $\boolp{\alphab{k}{\sigma}}$-bNTA $A'$ that tests $q$ for treewidth $k$
  and is monotone for the partial order on $\boolp{\alphab{k}{\sigma}}$.
\end{lemma}

\begin{proof}
  Fix $k$, $q$, and $A = (Q, F, \iota, \delta)$ the
  $\boolp{\alphab{k}{\sigma}}$-bNTA.
  We build the
  bNTA $A' = (Q, F, \iota', \delta')$ by setting,
  for all $(\tau, i) \in \boolp{\alphab{k}{\sigma}}$,
  $\iota'((\tau, i)) \defeq \bigcup_{0 \leq j \leq i} \iota((\tau, j))$
  and, for all $q_1, q_2 \in Q$, we pose:
  $\delta'(q_1, q_2, (\tau, i)) \defeq \bigcup_{0 \leq j \leq i} \delta(q_1, q_2, (\tau, j))$.

  Clearly $A'$ is monotone by construction for $\boolp{\alphab{k}{\sigma}}$. Besides, for
  any $\boolp{\alphab{k}{\sigma}}$-tree $T$, if $A$ accepts $T$ then $A'$ accepts $T$, so to prove
  the correctness of~$A'$ it suffices to prove the converse implication.

  Let us consider such a $T$, and consider an accepting run $\rho$ of~$A'$ on $T$.
  We build a new tree $T'$ whose skeleton is that of~$T$ and where for any
  leaf (resp.\ internal node)
  $n' \in T'$ with corresponding node $n \in T$ with
  $\lbl{n}=(\tau,j)$, we set $\lbl{n'}$ in $T'$ to be
  $(\tau, i)$ for some $i$ such that $\rho(n) \in \iota((\tau, i))$
  (resp.\ $\rho(n)
  \in \delta(\rho(\LC(n)), \rho(\RC(n)), (\tau, i))$), the existence of such an $i$
  being guaranteed by the definition of~$\iota'$ (resp.\ $\delta'$).

  We now observe that, by construction, $\rho$ is a run of~$A$ on $T'$, and it is
  still accepting, so that $T'$ is accepted by $A$. Hence, $\decode{T'}
  \models q$. But now we observe that, once again by construction, for every
  node $n'$ of~$T'$ with label $\tau'$ and with corresponding node $n$ in $T$
  with label $\tau$, it holds that $\tau' \leq \tau$. Hence we have $T'
  \leq T$, for which we can easily prove that $\decode{T'} \subseteq \decode{T}$, and
  thus, by monotonicity of~$q$, we must have $\decode{T} \models P$. Thus, $A$
  accepts $T$, proving the desired result.
\end{proof}

  From our earlier explanations, this proves
  Theorem~\ref{thm:provenance-encodings-mono}.
\end{toappendix}

Hence, for monotone GSO queries for which~\cite{green2007provenance} defines a
notion of semiring provenance (e.g., those that can be encoded to \deft{Datalog}, a
recursive query language that subsumes UCQs), our provenance
$\prov{q}{I}$ is easily seen to match 
the provenance of~\cite{green2007provenance}, specialized to the semiring
$\posbool{X}$ of monotone Boolean functions.
Indeed, both provenances obey the same intrinsic definition: they are the
function that maps to~$1$ exactly the valuations corresponding to subinstances
accepted by the query. Hence, we can understand
Theorem~\ref{thm:provenance-encodings-mono} as a tractability result for
$\posbool{X}$-provenance (represented as a circuit) on treelike instances.

Of course, the definitions of~\cite{green2007provenance}
go beyond $\posbool{X}$ and extend to arbitrary commutative semirings. We now
turn to this more general question.

\myparagraph{}{$\NN[X]$-provenance for UCQs}

First, we note that,
as shown by~\cite{green2007provenance}, the provenance of Datalog
queries for \emph{any}
semiring~$K$ can be computed in the semiring~$\NN[X]$, on instances where each
fact is annotated by its own variable in~$X$. Indeed, the provenance can then be
\emph{specialized} to~$K$, and the actual fact
annotations in~$K$, once known, can be used to replace the variables in the result,
thanks to a \emph{commutation with homomorphisms} property. Hence, \emph{we
restrict to $\NN[X]$-provenance} and to instances of this form, 
which covers all the examples above.

Second, in our setting of treelike instances, we evaluate queries using tree
automata, which are compiled from logical formulae with no prescribed execution
plan. For the semiring $\NN[X]$, this is hard to connect to the general definitions of provenance
in~\cite{green2007provenance}, which are mainly designed for
positive relational
algebra operators or Datalog queries. Hence, to generalize our constructions to
$\NN[X]$-provenance, \emph{we now restrict our query language to UCQs},
assuming without loss of generality that they contain no equality
atoms, 
We comment at the end of this section on the difficulties arising for richer
query languages.

We formally define the \deft{$\NN[X]$-provenance} of UCQs on relational instances by encoding them
straightforwardly to Datalog and using the Datalog provenance definition
of~\cite{green2007provenance}. The resulting provenance can be rephrased as
follows:

\begin{definition}
  \label{def:cqnprov}
  The \deft{$\Nx$-provenance} of a UCQ $q= \bigvee_{i=1}^n \exists
  \mathbf{x}_i \, q_i(\mathbf{x}_i)$ (where $q_i$ is a conjunction of atoms
  with free variables $\mathbf{x}_i$)
  on an instance $I$ is defined as:\\
  {\null\hfill $\provn{q}{I} \defeq \nsum_{i=1}^n\nsum_{\substack{f : \mathbf{x}_i \to
      \dom(I) \mathrm{~such~that~} I
  \models q_i(f(\mathbf{x}_i))}}
\nprod_{A(\mathbf{x}_i) \in q_i} A(f(\mathbf{x}_i))$.\hfill\null}\\
  In other words, we sum over each disjunct, and over each match of the
  disjunct; for each match, we
  take the product, over the atoms of the disjunct, of their image fact in~$I$,
  identifying each fact to the one variable in~$X$ that annotates it. 
\end{definition}

\begin{toappendix}
\begin{proposition}
  For any UCQ $q$,
  $\provn{q}{I}$, is the
  $\Nx$-provenance in the sense of~\cite{green2007provenance} of the associated
  Datalog query $P_q$ on $I$.
\end{proposition}

\begin{proof}
  The proof trees of each CQ within $q$ have a fixed structure, the only
  unspecified part being
  the assignment of variables. It is then clear that each variable assignment
  gives a proof tree, and this mapping is injective because all variables in the
  assignment occur in the proof tree. So for each CQ, we are summing on the same thing, and
  each term of the sum is the leaves of the proof tree, which is what we
  imposed. Further, the set of proof trees of $q$ is the union of the set
  of proof
  trees of each CQ, which means we can separate the sum for each CQ.
\end{proof}
\end{toappendix}

We know that $\provn{q}{I}$ enjoys all the usual properties of
provenance: it can be specialized to $\posbool{X}$, yielding back the
previous definition; it can be evaluated in the $\mathbb{N}$ semiring to
count the number of matches of a query; etc.

\begin{example}
  Consider the
  instance
  $I=\{F_1 \defeq R(a, a), \allowbreak F_2 \defeq R(b, c), \allowbreak F_3
  \defeq R(c, b)\}$
  and the CQ $q: \exists x y \, R(x, y) R(y, x)$.
  We have
  $\provn{q}{I} = F_1^2 + 2 F_2 F_3$
  and
  $\prov{q}{I} = F_1 \vee (F_2 \wedge F_3)$.
  Unlike $\posbool{X}$-provenance, $\NN[X]$-provenance 
  can describe that multiple atoms of the query map to the same fact, and that
  the same subinstance is obtained with two different query matches.
  Evaluating in the semiring $\mathbb{N}$ with facts annotated by~$1$,
  $q$ has $1^2+2\times 1\times 1=3$ matches.
\end{example}

\myparagraph{}{Provenance circuits for trees}
Guided by this definition of $\NN[X]$-provenance, we generalize the construction
of Section~\ref{sec:provenance} of provenance on trees to a more expressive provenance
construction, before we extend it to treelike instances as in
Section~\ref{sec:encodings}.

Instead of considering $\boolp{\Gamma}$-trees, we consider
$\natp{\Gamma}{p}$-trees for $p\in\mathbb{N}$, whose label set is $\Gamma \times \{0, \ldots,
p\}$ rather than $\Gamma \times \{0, 1\}$. Intuitively, rather than
uncertainty about whether facts are present or missing, we represent
uncertainty about the \emph{number of available copies} of facts, as UCQ
matches may include the same fact multiple times. We impose
on~$\boolp{\Gamma}$ the partial order $<$ defined by $(\tau, i) < (\tau, j)$
for all $\tau \in \Gamma$ and $i < j$ in $\{0, \ldots, p\}$, and call a
$\natp{\Gamma}{p}$-bNTA $A = (Q, F,
\iota, \delta)$
\deft{monotone} if for every $\tau < \tau'$ in $\natp{\Gamma}{p}$, we
have $\iota(\tau) \subseteq
\iota(\tau')$ and $\delta(q_1, q_2, \tau) \subseteq \delta(q_1, q_2, \tau')$ for every
$q_1, q_2 \in Q$. We write
$\pval{T}{p}$ for the set of all \emph{$p$-valuations} $\nu: V \to \{0,
\ldots, p\}$ of a $\Gamma$-tree~$T$. We write $\naccepts{A}{T}$ for a
$\natp{\Gamma}{p}$-tree~$T$ and $\natp{\Gamma}{p}$-bNTA~$A$ to denote the
number of accepting runs of~$A$ on~$T$.
We can now define:

\begin{definition}
  \label{def:cuprov}
  The \emph{$\Nx$-provenance} 
  of a
  $\natp{\Gamma}{p}$-bNTA $A$ on a $\Gamma$-tree~$T$
  is\\
  \hspace*{2em}$\provn{A}{T} \defeq
  \nsum_{\nu \in \pval{T}{p}}
  \naccepts{A}{\nu(T)}
  \nprod_{n \in T} n^{\nu(n)}$\\
  where each node $n \in T$ is identified with its own variable in~$X$.
  Intuitively, we sum over all valuations~$\nu$ of~$T$ to $\{0, \ldots, p\}$, and
  take the product of the tree nodes to the power of their valuation in $\nu$,
  with the number of accepting runs of~$A$ on~$\nu(T)$ as coefficient; in particular, the term
  for $\nu$ is $0$ if $A$ rejects $\nu(T)$.
\end{definition}

This definition specializes in $\posbool{X}$ to our earlier definition of
$\prov{A}{T}$, but extends it with the two features of $\NN[X]$:
multiple copies of the same nodes (represented as $n^{\nu(n)}$)
and multiple derivations (represented as $\naccepts{A}{\nu(T)}$). To construct
this general provenance, we need \emph{arithmetic
circuits}:

\begin{definition}
  A \deft{$K$-circuit} for semiring $(K, \oplus, \otimes, 0_K, 1_K)$ is a
  circuit with $\oplus$- and $\otimes$-gates instead of OR- and AND-gates (and
  no analogue of NOT-gates), whose
  input gates stand for elements of $K$. As
  before, the constants $0_K$ and $1_K$ can be written as $\oplus$- and
  $\otimes$-gates with no inputs. The element of $K$ \deft{captured} by a
  $K$-circuit is the element captured by its distinguished gate, under the recursive
  definition that $\oplus$- and $\otimes$-gates capture the sum and product of
  the elements captured by their operands, and input gates capture their own
  value.
\end{definition}

We now show an efficient construction for such provenance circuits, generalizing
the monotone analogue of Proposition~\ref{prp:provenance-circuits}. The proof
technique is to replace AND- and OR-gates by $\otimes$- and
$\oplus$-gates, and to consider possible annotations in $\{0, \ldots, p\}$
instead of $\{0, 1\}$. The correctness is proved by induction via a general
identity relating the
provenance on a tree to that of its left and right subtrees.

\begin{theorem}\label{thm:provenance-ncircuits}
  For any fixed $p\in\NN$, 
  for a $\natp{\Gamma}{p}$-bNTA~$A$ and a
  $\Gamma$-tree $T$,
  a $\Nx$-circuit capturing $\provn{A}{T}$ can be constructed
  in time $O(\card{A} \cdot \card{T})$.
\end{theorem}
  
\begin{toappendix}

  To prove Theorem~\ref{thm:provenance-nencodings}, we need
  a generalization of $\NN[X]$-provenance of automata on trees.
  Indeed, while Definition~\ref{def:cuprov} is natural for trees, using it to
  define the provenance of queries on treelike instances would lead to a subtle
  problem.
  The reason is that this provenance describes \emph{all} valuations of the tree
  for which the automaton accepts (up to the maximal multiplicity~$p$), and not
  the \emph{minimal} ones. For UCQs, this would intuitively mean that the
  resulting
  $\NN[X]$-provenance would reflect \emph{all} subinstances satisfying the
  query,
  not the minimal ones. This does not match Definition~\ref{def:cqnprov} and is
  undesirable: for instance, the specialization of such a provenance to $\NN$
  would have nothing to do with the number of query matches. The reason why this
  problem did not occur before is because both choices of definition collapse in
  the $\posbool{X}$ setting of the previous sections; so this problem is
  specific
  to the setting of general semirings such as~$\NN[X]$, which are not
  necessarily
  \emph{absorptive}~\cite{deutch2014circuits}.

However, in the case of UCQs, we can introduce a generalization of
Definition~\ref{def:cuprov}, for which we can show the analogue of
Theorem~\ref{thm:provenance-ncircuits}, and that will give us the right
provenance once lifted to treelike instances as in Section~\ref{sec:encodings}.
  For
  a $\Gamma$-tree~$T$ and $p, l \in \NN$, we
  introduce for $l\in\mathbb{N}$ the set of $p$-valuations that sum to $l$:
  $\pvalr{T}{p}{l} \defeq \{\nu \in \pval{T}{p} \mid \sum_{n \in T} \nu(n) =
  l\}$. We take $\pvalr{T}{p}{\all}$ to be $\pval{T}{p}$. We can now generalize
  Definition~\ref{def:cuprov} as follows:

  \begin{definition}
    \label{def:cuprovl}
    Let $l\in\NN\cup\{\all\}$.
    The \emph{$\Nx$-$l$-provenance}
    of a
    $\natp{\Gamma}{p}$-bNTA $A$ on a $\Gamma$-tree~$T$
    is:
    \[
    \provnr{A}{T}{l} \defeq
    \nsum_{\nu \in \pvalr{T}{p}{l}}
    \naccepts{A}{\nu(T)}
  \nprod_{n \in T} n^{\nu(n)}.\]

    Note that $\provn{A}{T} = \provnr{A}{T}{\all}$, so this definition indeed
    generalizes Definition~\ref{def:cuprov}.
  \end{definition}

We first prove the key lemma about the propagation of provenance throughout
encodings, which will be used in the inductive step of our correctness proof of
provenance circuits:

  \begin{lemmarep}{lem:proveqn}
  For any $l, p \in \NN$, $l \leq p$,
  for any non-singleton $\Gamma$-tree $T=(V,\LC,\RC,\lblf)$, letting $T_{\LCf}$ and $T_{\RCf}$ be its left
  and right subtrees and $n_\r$ be its root node, for any
  $\natp{\Gamma}{p}$-bNTA $A=(Q,F,\iota,\delta)$, writing $A_q$ for all $q \in Q$ the bNTA obtained from $A$ by
  making $q$ the only final state, we have:
  \[
    \provnr{A}{T}{l} =
    \nsum_{\substack{
    l_1 + l_2 + l' = l\\
    q_{\LCf}, q_{\RCf} \in Q\\
    q \in \delta(q_{\LCf}, q_{\RCf}, (\lblf(n_\r), l'))\\
    }}
    \provnr{A_{q_{\LCf}}}{T_{\LCf}}{l_1}
    \ntimes
    \provnr{A_{q_{\RCf}}}{T_{\RCf}}{l_2}
    \ntimes
    n_\r^{l'}
  \]
\end{lemmarep}

\begin{proof}
  We first observe the following identity, for any
  $\nu \in \pvalr{T}{p}{l}$ and any $q\in Q$, by definition of automaton runs:
  \[
    \naccepts{A_q}{\nu(T)} =
    \sum_{\substack{q_{\LCf}, q_{\RCf} \in Q\\
    q \in \delta(q_{\LCf}, q_{\RCf}, (\lblf(n_\r), \nu(n_\r)))}}
    \naccepts{A_{q_{\LCf}}}{\nu(T_{\LCf})}
    \cdot
    \naccepts{A_{q_{\RCf}}}{\nu(T_{\RCf})}
  \]

  We then observe that $\pvalr{T}{p}{l}$ can be decomposed as
  \[\bigsqcup_{l_1 + l_2 + l' = l} \pvalr{T_{\LCf}}{p}{l_1} \times
  \pvalr{T_{\RCf}}{p}{l_2} \times \{n_\r\mapsto
l'\}\] as a
  valuation of $T$ summing to $l$ can be chosen as a valuation of its left and right subtree
  and of $n_\r$ by assigning the possible weights.
  We also observe that the product over $n \in T$ can be split in a
  product on $n_\r$, on $n \in T_{\LCf}$ and on $n \in T_{\RCf}$. We can thus
  rewrite as follows:
  \[
    \provnr{A}{T}{l} =
    \nsum_{\substack{\nu_{\LCf} \in \pvalr{T_{\LCf}}{p}{l_1}\\
    \nu_{\RCf} \in \pvalr{T_{\RCf}}{p}{l_2}\\
    l_1 + l_2 + l' = l}}
    \nsum_{\substack{q_{\LCf}, q_{\RCf} \in Q\\
    q \in \Delta}}
    m_{\LCf} \cdot m_{\RCf}
    \left(
      \nprod_{n \in T_{\LCf}} n^{\nu_{\LCf}(n)}
    \right)
    \left(
      \nprod_{n \in T_{\RCf}} n^{\nu_{\RCf}(n)}
    \right)
    n_\r^{l'}
  \]
  where we abbreviated
  $m_{\LCf} \defeq \naccepts{A_{q_{\LCf}}}{\nu(T_{\LCf})}$,
  $m_{\RCf} \defeq \naccepts{A_{q_{\RCf}}}{\nu(T_{\RCf})}$, and
  $\Delta \defeq \delta(q_{\LCf}, q_{\RCf}, (\lblf(n_\r), l'))$.

  Reordering sums and performing factorizations, we obtain:
  \[
    \provnr{A}{T}{l} =
    \nsum_{\substack{q_{\LCf}, q_{\RCf} \in Q\\
    l_1 + l_2 + l' = l\\
    q \in \Delta}}
    \left(
    \nsum_{\nu_{\LCf} \in \pvalr{T_{\LCf}}{p}{l_1}}
    m_{\LCf}
      \nprod_{n \in T_{\LCf}} n^{\nu_{\LCf}(n)}
    \right)
    \left(
    \nsum_{\nu_{\RCf} \in \pvalr{T_{\RCf}}{p}{l_2}}
    m_{\RCf}
      \nprod_{n \in T_{\RCf}} n^{\nu_{\RCf}(n)}
    \right)
    n_\r^{l'}.
  \]

  Plugging back the definition of provenance yields the desired claim.
\end{proof}

We then prove the following variant of
Theorem~\ref{thm:provenance-ncircuits}:

\begin{theorem}
  \label{thm:provenance-ncircuits2}
  For any fixed $p\in\NN$ and $l \in \NN\cup\{\all\}$, a
  \emph{$\Nx$-$l$-provenance circuit}
  for a $\natp{\Gamma}{p}$-bNTA~$A$ and a
  $\Gamma$-tree $T$ (i.e., a $\Nx$-circuit capturing $\provnr{A}{T}{l}$) can be constructed in time $O(\card{A} \cdot \card{T})$.
\end{theorem}

\begin{proof}
  We modify the proof of Proposition~\ref{prp:provenance-circuits-mono}.

  We fix $l_0$ to be the $l$ provided as input. We will first assume
  $l_0\in\NN$, we explain at the end of the proof how to handle the
  (simpler) case $l_0=\all$.

  For every node~$n$ of the tree~$T$, we create one input gate $g^{\i}_n$ in $C$
  (identified to $n$), and for $j \in \{0, \ldots, p\}$, we create a gate
  $g^{\i, j}_n$ which is a $\ntimes$-gate of $j$ copies of the input gate
  $g^{\i}_n$. (By ``copies'' we mean $\ntimes$- or $\nplus$-gates whose sole
  input is $g^{\i}_n$, this being a technical necessity as $K$-circuits are
  defined as graphs and not multigraphs.) In particular, $g^{\i,0}_n$ is always
  a $1$-gate.

  We create one gate $g^{q,l}_n$ for $n \in T$, $q \in Q$, and $0 \leq l \leq
  l_0$.

  For leaf nodes $n$, for $q \in Q$, we set $g^{q,l}_n$ to be $g^{\i,l}_n$ if
  $q \in \iota(\lblf(n), l)$ and a $0$-gate otherwise.
  
  For internal nodes $n$, for every pair $q_{\LCf}, q_{\RCf} \in Q$ (that
  appears as input states of a transition of~$\delta$) and $0 \leq l_1, l_2 \leq
  l_0$ such that $l_1 + l_2 \leq l_0$, we create the gate
  $g^{q_{\LCf}, l_1, q_{\RCf}, l_2}_n$ as an $\ntimes$-gate of
  $g^{q_{\LCf},l_1}_{\LC(n)}$
  and $g^{q_{\RCf},l_2}_{\RC(n)}$, and, for $0 \leq l' \leq l_0$ such that $l_1 + l_2
  + l' \leq l_0$, we create one gate
  $g^{q_{\LCf}, l_1, q_{\RCf},l_2,l'}_n$ as the $\ntimes$-gate of $g^{q_{\LCf}, l_1,
  q_{\RCf}, l_2}_n$ and $g^{\i,l'}$.
  For $0 \leq l \leq l_0$, we set $g^{q,l}_n$ to be a~$\nplus$-gate of all the $g^{q_{\LCf},
  l_1,q_{\RCf},l_2,l'}_n$ such that $q \in \delta(q_{\LCf}, q_{\RCf}, (\lblf(n),
  l'))$ and $l_1 + l_2 + l' = l$.

  We define the distinguished gate $g_0$ as an $\ntimes$-gate of the
  $g^{q,l_0}_{n_\r}$ where $n_\r$ is the root of~$T$. The construction is again in
  $O(\card{A} \cdot \card{T})$ for fixed $l$ and~$p$.

  \medskip

  To prove correctness, we show by induction that the element captured by $g^{q,l}_n$ is
  $\provnr{A_q}{T_n}{l}$ where $A_q$ is $A$ with $q$ as the only final state,
  and $T_n$ is $T$ rooted at $n$.

  As a general property, note that for any node $n$, the value captured by
  $g^{\i,j}_n$ for $0 \leq j \leq p$ is $n^j$.

  For a leaf node $n$, $\provnr{A_q}{T_n}{l} = n^l$ if $q \in \iota(\lblf(n),
  l)$ and $0$ otherwise, which is the value captured by $g^{q,l}_n$.
  
  For an internal node $n$, the claim follows immediately by
  Lemma~\ref{lem:proveqn}, applying the induction hypothesis to
  $g^{q_\LC,l_1}_{\LC(n)}$ and $g^{q_\RC,l_2}_{\RC(n)}$.

  We conclude because clearly we have $\provnr{A}{T}{l_0} = \nsum_{q \in F}
  \provnr{A_q}{T}{l_0}$, so the value captured by $g_0$ is indeed correct.

  \medskip

  Now, if $l_0=\all$, we do the same construction, but we only need a single
  node $g_n^q$ for $n\in T$ and $q\in Q$ instead of $l_0+1$ nodes
  $g_n^{q,l}$. For leaf nodes, $g_n^q$ is the $\nplus$-node of the
  $g_n^{\i,l}$; for internal nodes,
  $g_n^{q}$ is simply the $\nplus$-gate of all
  $g_n^{q_\LC,q_\RC,l}$ gates with
  $q\in\delta(q_\LC,q_\RC,(\lambda(n),l))$,
  each of them being an $\ntimes$-gate of the
  $g_n^{q_\LC,q_\RC}$ gate and the $g^{\i,l}$ gate. Finally,
  $g_n^{q_\LC,q_\RC}$ is the $\ntimes$-gate of $g_{\LC(n)}^{q_\LC}$ and
  $g_{\RC(n)}^{q_\RC}$. Correctness is shown using a variant of
  Lemma~\ref{lem:proveqn} on $\provnr{A}{T}{\all}$ which replaces $l_1 + l_2 +
  l' = l$ in the sum subscript by $0 \leq l' \leq p$.
\end{proof}

\end{toappendix}

\myparagraph{}{Provenance circuit for instances}
Moving back to provenance for UCQs on bounded-treewidth 
instances, we
obtain a linear-time provenance construction:

\begin{theoremrep}{thm:provenance-nencodings}
  For any fixed $k\in \NN$ and UCQ $q$, for any $\sigma$-instance $I$ such that
  $\width(I) \leq k$, one can construct a $\Nx$-circuit  that captures $\provn{q}{I}$ 
  in time~$O(\card{I})$.
\end{theoremrep}

The proof technique is to construct for each disjunct $q'$ of $q$ a
$\natp{\Gamma}{p}$-bNTA~$A_{q'}$, where $\Gamma \defeq \alphab{k}{\sigma}$ is the alphabet for tree encodings
of width~$k$,
and $p$ is the maximum number
of atoms in a disjunct of~$q$. We want $A_{q'}$ to test $q'$ on tree encodings over $\Gamma$,
\emph{while preserving
multiplicities}: this is done by enumerating all possible self-homomorphisms of
$q'$, changing $\sigma$ to make the multiplicity of atoms part of the
relation name, encoding the resulting queries to automata as
usual~\cite{Courcelle90} and going back to the original~$\sigma$. We then
apply a variant of Theorem~\ref{thm:provenance-ncircuits} to
construct a $\Nx$-circuit capturing the provenance of $A_{q'}$ on a tree encoding of
$I$ but for valuations that sum to the number of atoms of~$q'$; this restricts to
bag-subinstances corresponding exactly to matches of $q'$.
We obtain a $\Nx$-circuit that
captures $\provn{q}{I}$ by combining the circuits for each disjunct,
the distinguished gate of the overall circuit being a $\nplus$-gate of that of each circuit.

\begin{toappendix}
  We first give some preliminary definitions. We need to introduce
\deft{bag-instances}, to materialize the possibility that a fact
is used multiple times in a UCQ:

\begin{definition}
  \label{def:multiset}
  A \deft{multiset} is a function $M$ from a finite \deft{support}
  $\supp(M)$ to $\NN$.
  We define the relation $M \subseteq M'$ if $\supp(M) \subseteq \supp(M')$ and
  for all $s \in \supp(M)$
  we have
  $M(s) \leq M'(s)$. We write $x \in M$ to mean
  that
  $M(x) > 0$.
  Given a set $S$ and multiset $M$, we write $M \sqsubseteq S$ to mean that
  $\supp(M) \subseteq S$, and for $p \in \NN$ we write $M \sqsubseteq^p S$ to
  mean that $M \sqsubseteq S$ and $M(a) \leq p$ for all $a \in \supp(M)$.

  A \deft{bag-instance} $J$ is a multiset of facts on $\dom(J)$. Where necessary
  to avoid confusion,
  we call the ordinary instances \deft{set-instances}.
  For two bag-instances $J$ and $J'$, we say that $J$ is a
  \deft{bag-subinstance} of $J'$ if $J \subseteq J'$ holds
  (as multisets). We say that $J$ is a \deft{bag-subinstance}
  of a set-instance $I$ if $J \sqsubseteq I$, and a \deft{$p$-bag-subinstance}
  of $I$ if $J \sqsubseteq^p I$.
  In other words, set-instances are understood as bag-instances where facts
  have an arbitrarily large multiplicity (and \emph{not} multiplicity equal to~$1$).
  The \deft{truncation to $p$} of a bag-instance $J$ is
  $\trunc{J}{p}(F) \defeq \min(J(F), p)$ for all $F \in \supp(J)$.

  A \deft{bag-homomorphism} $h$ from a bag-instance $J$ to a bag-instance $J'$ is
  a mapping from $\supp(J)$ to $\supp(J')$ with the following condition:
  for each $F \in \supp(J')$, letting $F_1, \ldots, F_n$ be the facts of
  $\supp(J)$ such that $h(F_i) = F$ for $1 \leq i \leq n$, we have
  $\sum_{i=1}^n J(F_i) \leq J'(F)$.
\end{definition}

We accordingly define \deft{bag-queries} as queries on such bag-instances.
Intuitively, bag-queries are like regular Boolean queries on instances, except
that they can ``see'' the multiplicity of facts. This is crucial to talk about
the required multiplicity of facts in matches, which we need to talk about the
$\Nx$-provenance of UCQs.

  \begin{definition}
    \label{def:bagquery}
    A \deft{bag-query}~$q$ is a query on bag-instances.
    The bag-query $q'$ \deft{associated} to a CQ $\exists \mathbf{x} \,
    q(\mathbf{x})$ is defined as follows. A \deft{match} of $q$ in a
    bag-instance $J$ is a bag-homomorphism from $q$ (seen as a bag-instance of
    facts over $\mathbf{x}$) to $J$. We say that $J \models q'$ if $q$ has a
    match in~$J$.

    The bag-query associated to a UCQ $q = \bigvee_{i=1}^n \exists
    \mathbf{x}_i \, q_i(\mathbf{x}_i)$ is the disjunction of the bag-queries for
    each CQ in the disjunction.

    Alternatively it is easily seen that $J \models q'$, for $q'$ the bag-query
    associated to a UCQ $q$, iff $J$ contains a bag of facts that can be used as
    the leaves of a derivation tree for the Datalog query $P_q$ associated to~$q$.
  \end{definition}

We notice that the bag-query associated to a UCQ~$q$ is \deft{bounded}, namely,
the fact that it holds or not cannot depend on the multiplicity of facts beyond
a certain maximal value (the maximal number of atoms in a disjunct of the UCQ):

\begin{definition}
  \label{def:bounded}
  A bag-query $q$ is \deft{bounded} by $p \in \NN$ if, for
  any bag-instance $J$, if $J \models q$, then the truncation
  $\trunc{J}{p}$ of $J$ is such that $\trunc{J}{p} \models q$. A bag-query is
  \deft{bounded} if it is bounded by some $p \in \NN$.
\end{definition}

We now extend our definitions of tree encodings and automaton compilation to
bag-queries. First, tree encodings simply generalize to bag-instances as tree
encodings annotated with the multiplicities of facts, that is,
$\natp{\alphab{k}{\sigma}}{p}$-trees:

  \begin{definition}
    Let $k, p \in \NN$ and let $J$ be a bag-instance such
    that $J(F) \leq p$ for all $F \in J$. Let $I \defeq \supp(J)$ be the
    underlying instance of $J$, and let $T_I$ be its tree encoding (a
    $\alphab{k}{\sigma}$-tree). We define the \deft{$(k, p)$-tree-encoding} $T_J$ of $J$
    as the tree with same skeleton as $T_I$ where any node $n$ encoding a fact
    $F$ of $I$ is given the label $(\lbl{n}, J(F))$ and other nodes are given
    the label $(\lbl{n}, 1)$.
    We accordingly define $\decode{\cdot}$ on
    $\natp{\alphab{k}{\sigma}}{p}$-trees to yield bag-instances in
    the expected way, again returning $\bot$ whenever two nodes code the same
    fact (rather than summing up their multiplicity).
  \end{definition}

  We then define what it means for a $\natp{\alphab{k}{\sigma}}{p}$-bNTA to
  \deft{test} a bag-query $q$. Note that the definition implies that the
  automaton cannot ``see'' multiplicities beyond $p$, so we require that the
  query be $p$-bounded so that the limitation does not matter.

  \begin{definition}
    For $q$ a bag-query and $k, p \in \NN$, a $\natp{\alphab{k}{\sigma}}{p}$-bNTA
    $A$ \deft{tests} $q$ for treewidth $k$ if $q$ is bounded by~$p$ and for every
    $\natp{\alphab{k}{\sigma}}{p}$-tree~$E$, we have $\accepts{A}{E}$ iff $\decode{E}
    \models q$.
  \end{definition}

  We then provide a general definition and prove a lemma about constructing the
  union of bNTAs:

  \begin{definition}
    Let $\Gamma$ be a finite label set and let $A_i = (Q_i, F_i, \iota_i,
    \delta_i)$ be a family of $\Gamma$-bNTAs. Assume without loss of generality
    that the $Q_i$ have been renamed so that they are pairwise disjoint.
    The \deft{union bNTA} is the $\Gamma$-bNTA $A_\sqcup = (Q_\sqcup,
  F_\sqcup, \iota_\sqcup, \delta_\sqcup)$ defined by $Q_\sqcup \defeq \bigsqcup_i
  Q_i$, $F_\sqcup \defeq \bigsqcup_i F_i$, for every $\tau \in \Gamma$
  $\iota_\sqcup(\tau) \defeq \bigsqcup_i \iota_\sqcup(\tau)$, and
  $\delta_\sqcup$ is only defined for $q_1, q_2 \in Q_i$ for some $Q_i$, in
  which case it is defined as $\delta_\sqcup(q_1, q_2, \tau) \defeq
  \delta_i(q_1, q_2, \tau)$.
  \end{definition}

  \begin{lemma}
    \label{lem:union}
    For any family of $\Gamma$-bNTAs $A_i$, letting $A_\sqcup$ be the union bNTA
    of the $A_i$, for any $\Gamma$-tree $T$, we have $\accepts{A_\sqcup}{T}$ iff
    $\accepts{A_i}{T}$ for some $A_i$, and more precisely we have
    $\naccepts{A_\sqcup}{T} = \sum_i \naccepts{A_i}{T}$.
  \end{lemma}

  \begin{proof}
    The claim about acceptance and the number of runs is straightforward by
    noticing that the runs of $A_\sqcup$ on $T$ are exactly the disjoint union of the
    runs of the $A_i$ on $T$.
  \end{proof}

  Our last preliminary result is to show that every bag-query corresponding to a
  UCQ can be encoded to an automaton that tests it.

\begin{propositionrep}{prp:ucqbagftar}
  Let $q$ be a UCQ. There is $p \in \NN$ such that, for any $k \in \NN$, we can
  compute a $\natp{\alphab{k}{\sigma}}{p}$-bNTA $A$ that tests $q$ for treewidth
  $k$.
\end{propositionrep}

\begin{proof}
  We introduce some notation. We call \cqneq the language of CQs which can
  feature atoms of the form $x \neq y$, and \ucqneq the language of UCQs except
  the disjuncts are in \cqneq.
  We write $\vars{q}$ for the variables of a query $q$ of \cqneq.
  The bag-query associated to a query in \cqneq or \ucqneq is defined as for the
  corresponding query with no inequalities, but imposing the inequalities on
  matches. Formally, a match of a \cqneq query $q'$ is a match $h$ of $q'$ such
  that $h(x) \neq h(y)$ for any two variables $x$ and $y$ such that $x \neq y$
  occurs in $q$. (Note that the multiplicity of inequality atoms is irrelevant.)

  We first note that, writing the UCQ $q$ as the disjunction of CQs $q_i$, if we
  can show the claim for each $q_i$, then the result
  clearly follows from $q$ by computing one bNTA $A_i$ for each $q_i$ that
  tests $q_i$ for treewidth $k$ and uses $p = \max_i p_i$, where $p_i$ is the
  multiplicity for which the result was shown for each $q_i$ (clearly if the
  claim holds for a value of $p$ then it must hold for larger values by ignoring
  larger multiplicities). We then construct the union bNTA $A_\sqcup$ of these bNTAs
  to obtain a bNTA that tests $q$ (Lemma~\ref{lem:union}).

  We see a CQ $q$ as an existentially quantified multiset of atoms (the same
  atom, i.e., the same relation name applied to the same variables in the same
  order, can occur multiple times; in other words we distinguish, e.g., $\exists
  x R(x)$ and $\exists x R(x) R(x)$). Let $\vars{q}$ be the set of the variables
  of $q$ (which are all existentially quantified as $q$ is Boolean). We call
  $\calE_q$ the set of all equivalence classes on $\vars{q}$ (which is of course
  finite), and for ${\sim} \in
  \calE_q$ we let $\quot{q}{\sim}$ be the query in \cqneq obtained by choosing one
  representative variable in $\vars{q}$ for each equivalence class of $\sim$ and
  mapping every $x \in \vars{q}$ to the representative variable for the class of
  $x$ (dropping in the result the useless existential quantifications on
  variables that do not occur anymore), and adding disequalities $x \neq y$
  between each pair of the remaining variables.

  We rewrite a CQ $q$ to the UCQ $q' \defeq \bigvee_{\sim \in \calE_q}
  \quot{q}{\sim}$. We claim that for every bag-instance $I$, if $I \models q$
  then $I \models q'$, which justifies that for an instance $I''$, considering the subinstances of
  $I''$, $W_K(q, I'') = W_K(q', I'')$. For the forward implication, assuming that $I \models q$,
  letting $m$ be the witnessing match, we consider the $\sim_m$ relation defined
  by $x \sim_m y$ iff $m(x) = m(y)$, and it is easily seen that
  $I \models \quot{q}{\sim_m}$. For the backward implication, if $I \models
  \quot{q}{\sim}$ for some $\sim \in \calE_q$, it is immediate that $I \models
  q$ with the straightforward match. Hence, using again Lemma~\ref{lem:union}, it suffices to show the result
  for queries in \cqneq which include inequality axioms between all their
  variables. We call those \deft{forced queries}.

  We now show that the claim holds for forced queries. To see this, considering such a
  query $q$ on signature $\sigma$, letting $p$ be the sum of the multiplicities
  of all atoms in $q$ (i.e., the number of atoms in the original CQ $q$), let $\sigma_p$ be the signature obtained from
  $\sigma$ by creating a relation $R^i$ for $1 \leq i \leq p$, with arity
  $\arity{R}$, for every relation $R$ of $\sigma$, and let $q'$ be the rewriting
  of $q$ obtained by replacing every atom $R(\mathbf{a})$ with multiplicity $m$
  by the disjunction $\bigvee_{m \leq j \leq p} R^j(\mathbf{a})$ (and keeping the
  inequalities), rewritten to a \ucqneq. We now see $q'$
  as a \ucqneq in the usual sense (without multiplicities). We now claim that for any
  bag-instance $I$ on $\sigma$ where facts have multiplicity $\leq p$, letting
  $I'$ be the set-instance obtained by replacing every fact $F =
  R(\mathbf{a})$ of $I$ with multiplicity $m = I(F)$ by the fact
  $R^m(\mathbf{a})$, $I \models q$ iff $I' \models q'$. To see
  why, observe that, as $q$ is a forced query, if $q$ has a match $m$ then every
  atom $A$ of $q$ must be mapped by $m$ to a fact of $I$ (written $m(A)$) and
  this mapping must be injective (because $m$ is), so that the necessary and
  sufficient condition is that $I(m(A)) \geq p_A$ (where $p_A$ is the
  multiplicity of $A$ in $q$) for every atom $A$ of $q$; and this is equivalent
  to $I' \models q'$.

  Now, $q'$ is a \ucqneq so it can be tested \emph{in the sense of
  Definition~\ref{def:tests}} (as it is expressible in GSO, so we can
  apply Theorem~\ref{thm:courcelle}); fix $k \in
  \mathbb{N}^*$ and let $A_{q'} = (Q, F, \iota, \delta)$ be a
  $\alphab{k}{\sigma_p}$-bNTA that tests $q'$ for width $k$. We
  build a $\natp{\alphab{k}{\sigma}}{p}$-bNTA $A_q = (Q, F, \iota', \delta')$ by
  relabeling $A_{q'}$ in the following way. Recall the definition of
  $\alphab{k}{\sigma}$ (Definition~\ref{def:kfact}).
  For every $((d, f), i) \in \natp{\alphab{k}{\sigma}}{p}$, set $f'$ to
  be either $f$ if $f = \emptyset$ and $f' = R^i(\mathbf{a})$ if $f =
  R(\mathbf{a})$, and set
  $\iota'(((d, f), i))$ to be $\iota((d, f'))$
  and $\delta'(q_{\LCf}, q_{\RCf}, ((d, f), i))$ to be $\delta(q_{\LCf},
  q_{\RCf}, (d, f'))$
  for every $q_{\LCf}, q_{\RCf} \in Q$.

  We now claim that $A_{q}$ tests $q$ for treewidth $k$. To see why, it
  suffices to observe that for any $\natp{\alphab{k}{\sigma}}{p}$-tree $T$, letting $T'$ be the
  $\alphab{k}{\sigma_p}$-tree obtained in the straightforward manner, then $A_q$
  accepts $T$ iff $A_{q'}$ accepts $T'$, which is immediate by construction.
  Now indeed, as we know that $A_{q'}$
  accepts $T'$ iff $\decode{T'} \models q'$ (as $A_{q'}$ tests $q'$), and (as
  immediately $\decode{T'}$ is the $\sigma_p$-instance corresponding to
  $\decode{T}$ as $I'$ corresponds to $I$ above) that $\decode{T'} \models q'$
  iff $\decode{T} \models q$, we have the desired equivalence.

  The only thing left is to observe that $A_q$ does not only correctly test $q$ on instances
  where each fact has multiplicity $\leq p$, but correctly tests $q$ on all
  bag-instances. But this is straightforward: 
  as $q$ matches at most $p$ fact occurrences in the instance
  $I$, we have $I \models q$ iff $\trunc{I}{p} \models q$. This
  concludes the proof.
\end{proof}

We are now ready to prove Theorem~\ref{thm:provenance-nencodings}:
\end{toappendix}

\begin{proof}
  We show the proof for CQs, and then extend to UCQs.

  Let $k \in \NN$, $q: \exists \mathbf{x}\, q'(\mathbf{x})$ be the CQ. We rewrite
  $q$ to $q'':\exists \mathbf{x}\, q'(\mathbf{x})\land \bigwedge_{x \in \mathbf{x}}
  P_x(x)$ for fresh unary predicates $P_x$. We apply
  Proposition~\ref{prp:ucqbagftar} to compile $q''$ to a
  $\natp{\alphab{k}{\sigma}}{p}$-bNTA $A$, where $p$ is the number of atoms of
  $q''$, such that $A$ tests $q''$ for treewidth $k$. We can clearly design a
  $\natp{\alphab{k}{\sigma}}{p}$-bNTA $A'$ that checks on a
  $\natp{\alphab{k}{\sigma}}{p}$-tree whether, for all $x \in \mathbf{x}$, the
  input tree contains exactly one $P_x$-fact: this can be done with state space
  $2^{\mathbf{x}}$. We intersect $A$ and $A'$ to obtain a bNTA that recognizes
  all $\natp{\alphab{k}{\sigma}}{p}$-trees that satisfy the bag-query
  associated to $q''$ and have exactly one $P_x$-fact for all $x \in
  \mathbf{x}$, and determinize this bNTA to obtain an equivalent automaton $A''$ which
  is \deft{deterministic}: if it has an accepting run then it has exactly one
  accepting run.

  Let $I$ be the input instance, and $I'$ be the instance where we added one
  fact $P_x(a)$ for all $x \in \mathbf{x}$ and $a \in \dom(I)$: we call those
  the \deft{additional facts}. We can clearly
  compute $I'$ from $I$ in linear time, and the treewidth is unchanged. Let
  $T_{I'}$ be a tree encoding of $I'$, that is, a $\alphab{k}{\sigma}$-tree.

  We claim that we can construct, from $A''$, a bNTA $A'''$ such that, for any
  valuation $\nu$ of $T_{I'}$ that gives multiplicity~$1$ to the additional facts, the
  number of accepting runs of $A'''$ on $\nu(T_{I'})$ is the number of valuations
  $\nu'$ from the additional facts to $\{0, 1\}$ such that $A''$ accepts
  $\nu''(T_{I'})$, where $\nu''$ follows $\nu'$ on nodes encoding additional
  facts and follows~$\nu$ otherwise. We proceed as follows: first, duplicate the states
  of $A''$ so that every state $q$ is replicated to two states $q$ and $q'$, $q$
  and $q'$ being treated exactly the same way in terms of transitions in
  $\delta$ and in terms of being final (this preserves determinism). Now, ensure
  that for any two states $q_1$ and $q_2$ and labels $(\tau, 0)$ and $(\tau, 1)$
  in $\natp{\alphab{k}{\sigma}}{p}$ that encode a present or absent additional
  fact, $\delta(q_1, q_2, (\tau, 0))$ and $\delta(q_1, q_2, (\tau, 1))$ are
  disjoint (as $A''$ is deterministic, those are single facts, so if they are
  the same fact, replace one of them by its equivalent copy). Now, modify the
  transitions of the automaton so that, for any states $q_1$ and $q_2$ and
  $\tau$ encoding an additional fact, $\delta(q_1, q_2, (\tau, 1))$ is
  $\delta(q_1, q_2, (\tau, 0)) \cup \delta(q_1, q_2, (\tau, 1))$. It is now
  clear that the resulting automaton $A'''$ satisfies the desired property: for any
  valuation $\nu$ as above, there is a bijection between the accepting runs of
  $A'''$ and the valuations $\nu'$ as above such that $\nu''(T_{I'})$
  is accepted by $A''$.

  We now apply Theorem~\ref{thm:provenance-ncircuits2} with $l$ the number
  of facts in the CQ~$q''$ to obtain a $\Nx$-circuit that captures the
  $\Nx$-$l$-provenance of $A'''$ on $T_{I'}$;
  and fix to $1$ all inputs except those coding a
  fact of $I$ (i.e., nodes coding additional facts are set to $1$) and
  rename the remaining inputs to match the facts of $I$.
  Let $l'$ be the number of facts in the original CQ~$q$. Then the
  circuit captures:
  \[
    \nsum_{\substack{J\sqsubseteq I\\J\models
    q\\\sum_{F \in \supp(J)} J(F)=l'}}\left|\{\:f:\mathbf{x}\to\dom(I)\mid
    J \models q'(f(\mathbf{x}))\:\}\right|
    \nprod_{F\in J} F^{J(F)}.
  \]
  
  Now, $J \models q'(f(\mathbf{x}))$ just means $q'(f(\mathbf{x}))
  \subseteq J$ (as bag-instances), and as $q'(f(\mathbf{x}))$ and $J$ have same
  total multiplicity, this actually means $J = q'(f(\mathbf{x}))$. Hence, the
  above is equal to:\[
  \nsum_{\substack{f : \mathbf{x} \to \dom(I) \\ I
  \models q'(f(\mathbf{x}))}}
  \nprod_{A(\mathbf{x}) \in q'} A(f(\mathbf{x}))
  .
  \]
  and this is exactly $\provn{q}{I}$.

  For UCQs, observe that the provenance we need to compute
  (Definition~\ref{def:cqnprov}) is simply the sum of the provenance for
  each CQ. So we can just independently build a circuit for each CQ and
  combine the circuits into one (merging the input gates), while choosing
  as distinguished gate a $\nplus$-gate of each distinguished gate.
\end{proof}

Remember that an $\Nx$-circuit can then be specialized to a circuit for an
arbitrary
semiring (in particular, if the semiring has no variable,
the circuit can be
used for evaluation); thus, this provides provenance for~$q$ on~$I$
for any semiring.

\myparagraph{}{Going beyond UCQs}
To compute $\NN[X]$-provenance beyond UCQs (e.g., for monotone GSO
queries or their
intersection with Datalog), the main issue is fact multiplicity: multiple
uses of facts are easy to describe for UCQs (Definition~\ref{def:cqnprov}), but for
more expressive languages we do not know how to define
them and connect them to automata.

In
fact, we can build a query~$P$, in guarded
Datalog~\cite{Gradel:2000:EEM:1765236.1765272}, such that the
smallest number of occurrences of a fact in a derivation tree for $P$ cannot be
bounded independently from the instance. We thus cannot rewrite $P$
to a fixed finite bNTA testing multiplicities on all input
instances. However, as guarded Datalog is monotone and
GSO-expressible, we can compute the $\posbool{X}$-provenance
of~$P$
with Theorem~\ref{thm:provenance-encodings}, hinting at a difference between
$\posbool{X}$ and $\NN[X]$-provenance computation for queries beyond UCQs.

\begin{toappendix}
Note that the proof of Theorem~\ref{thm:provenance-nencodings} implicitly relies
on the fact that UCQs are bounded in the sense of Definition~\ref{def:bounded}, and we cannot hope to rewrite a query to a
$\natp{\alphab{k}{\sigma}}{p}$-bNTA that sensibly tests it if it is not bounded.
However, we can show:

\begin{proposition}
  There is a guarded monadic Datalog query $P$ whose associated bag-query
  $q_P$ is not bounded.
\end{proposition}

\begin{proof}
  Consider the Datalog query~$P$ consisting of the rules $S(y) \leftarrow S(x),R(x, a, y),A(a)$,
  $\goal \leftarrow S(x),T(x)$. For all $n \in \NN$, consider the
  instance $I_n = \{R(a_1, a, a_2), \allowbreak R(a_2, a, a_3), \ldots,
    R(a_{n-1}, a, a_n), \allowbreak
  S(a_1), \allowbreak T(a_n), \allowbreak A(a)\}$. It is easily verified that the only proof tree of $P$
  on $I_n$ has $n-1$ leaves with the fact $A(a)$. Hence, assuming that the
  bag-query $q_P$ captured by $P$ is bounded by $p$, 
  considering the bag-instance~$J$ formed of the leaves of the sole proof tree
  of $P$ on $I_{p+2}$, it is not the case that $\trunc{J}{p} \models q_P$,
  contradicting boundedness.
\end{proof}
\end{toappendix}

\mysec{Applications}{sec:applications}
In Section~\ref{sec:semirings} we have shown a $\Nx$-provenance circuit
construction for UCQs on treelike instances. This construction can be
specialized to any provenance semiring, yielding various applications:
counting query results by evaluating in~$\NN$, computing the cost of a
query in the tropical semiring, etc. By
contrast, Section~\ref{sec:encodings} presented a provenance construction for
arbitrary GSO queries, but only for a Boolean representation of
provenance, which does not capture multiplicities of facts or
derivations. The results of both sections are thus incomparable. In this
section we show applications of our constructions to two
important problems: \deft{probability evaluation}, determining the
probability that a query holds on an uncertain instance, and
\deft{counting}, counting the number of answers to a given query. These
results are consequences of the construction of
Section~\ref{sec:encodings}.

\begin{toappendix}
  \subsection{Preliminaries}

All proofs about probability evaluation in this section will use the notion of
\deft{cc-instances}, which we now introduce.

In this appendix, all Boolean circuits are non-monotone
(i.e., they allow NOT-gates)
and arity-two (Definition~\ref{def:aritytwo}),
 unless stated otherwise. We will first define the formalism of cc-instances,
then state a result about the construction of circuits (the analogue of
provenance circuits) for them, using Theorem~\ref{thm:provenance-encodings},
and finally explain how probability evaluation is performed using that result
using message-passing. We conclude by presenting the similar formalism of
\deft{pc-instances}, and stating tractability results for them implied by the
results on pcc-instances.

\paragraph{cc-instances.}
We define the formalism of \deft{cc-instances}:

\begin{definition}
  A \emph{cc-instance} is a triple $J = (I, C,
  \phi)$ of a relational $\sigma$-instance $I$, a (non-monotone arity-two) Boolean circuit $C$,
  and a mapping~$\phi$ from the facts of~$I$ to gates of\/ $C$.
  The \deft{inputs} $J_{\inp}$ of~$J$ are $C_\inp$.
  For every valuation $\nu$ of $J_\inp$,
  the possible world~$\nu(J)$ is the subinstance of~$I$ that contains the facts $F$ of $I$ such
  that $\nu(C)(\phi(F)) = \true$, and, as for c-instances, $\sems{J}$ is the set
  of possible worlds of~$J$.

  A pcc-instance is a 4-tuple $J = (I, C, \phi, \pfun)$ such that $J' = (I, C, \phi)$ is
  a cc-instance (and $J_\inp \defeq J'_\inp$) and $\pfun: J_\inp \to [0,1]$
  gives a probability to each input. As for pc-instances, the probability
  distribution $\sems{J}$ has universe $\sems{J'}$ and probability
  measure $\Pr_{J}(I') = \sum_{\substack{\nu \mid \nu(J) = I'}} \Pr_{J}(\nu)$ with
  the product distribution:
  \[\Pr_{J}(\nu)=\prod_{\substack{g \in J_\inp\\\nu(g) = \true}} \pfun(g) \prod_{\substack{g \in
  J_\inp\\\nu(g) = \false}} (1 - \pfun(g)).\]
\end{definition}

We define relational encodings and treewidth for cc-instances:

\begin{definition}
  \label{def:relenc}
  Let $\sigma_{\circuit}$ be the signature
  of the relational encoding of Boolean circuits
  (Definition~\ref{def:circuittw}). Let $\sigma$ be a signature
  and $\sigma^+$ be the signature with one relation $R^+$ of arity
  $\arity{R} + 1$ for every relation $R$ of $\sigma$.
  The \deft{relational encoding} $I_J$
  of a cc-instance $J = (I, C, \phi)$ over signature~$\sigma$,
  is the $(\sigma_\circuit \sqcup \sigma^+)$-instance containing both
  the $\sigma_\circuit$-instance $I_C$ encoding $C$
  and one fact $R^+(\mathbf{a},\phi(F))$ for every fact $F = R(\mathbf{a})$
  in~$I$.

  A \deft{tree decomposition} of a cc-instance $J$ is a tree decomposition of
  $I_J$. Tree decompositions of pcc-instances are defined as a tree
  decomposition of the corresponding cc-instance (the probabilities are
  ignored).
\end{definition}

\paragraph{Circuits for cc-instances.}
We claim the following result about cc-instances, intuitively corresponding to
the provenance circuits of Section~\ref{sec:encodings} for them (combined with
their circuit annotation):

\begin{theorem}
  \label{thm:main}
  For any fixed integer $k$
  and GSO sentence $q$, one can
  compute in linear time, from a cc-instance $J$ with
  $\width(J) \leq k$, a Boolean circuit $C$ on $J_\inp$
  such that for every valuation~$\nu$ of $J_{\inp}$, $\nu(C) = \true$
  iff $\nu(J) \models q$, with $\width(C)$ depending only on $k$ and $q$.
\end{theorem}

We now prove this result, explaining later what it implies
in terms of probability evaluation. We first introduce the notion of
\deft{cc-encoding}. Recall the definition of tree decompositions of circuits
(Definition~\ref{def:circuittw}):

\begin{definition}
  \label{def:kccenc}
A \deft{cc-encoding} of width $k$ is a tuple $E' = (E, C,
T,\chi)$ of
a $\alphab{k}{\sigma}$-tree $E$ of width~$k$,
a Boolean circuit $C$,
a tree decomposition $T$ of\/~$C$ of width~$k$
with same skeleton as $E$,
and a mapping $\chi:T\rightarrow C$ selecting a \deft{selected gate}
such that $\chi(b)\in\dom(b)$ for all $b \in T$.
The \deft{inputs} $E'_\inp$ of $E'$ are $C_\inp$.

Given a valuation $\nu$ of\/ $C_\inp$, we extend it to an evaluation of $C$, and
see it as a Boolean valuation of
$E$ by setting $\nu(n) \defeq \nu(\chi(b))$ for the bag $b$ of
$T$ corresponding to~$n$ in~$E$, and write $\nu(E')$ the resulting
$\boolp{\alphab{k}{\sigma}}$-tree.
\end{definition}

First, we explain how we can compute a cc-encoding of our cc-instance $J = (I, C, \phi)$ by
``splitting'' its tree decomposition $T$ in a tree decomposition of\/ $C$ and a
$\alphab{k}{\sigma}$-tree $E$ of $I$ with same skeleton, with $\chi$ keeping track of the
gate of\/ $C$ to which each node $n \in E$ was mapped by $\phi$. Formally:

\begin{lemmarep}{lem:uncertainte}
  Recall the definition of $\tevalf$ (Definition~\ref{def:neuter}).
  Given a cc-instance $J = (I, C, \phi)$
  and a tree decomposition $T$ of $J$ of width $k$,
  one can compute
  a cc-encoding $E' = (E, C', T',\chi)$
  of width $k$, with $C = C'$, such that for any valuation
  $\nu$ of\/ $C_\inp$, $\teval{\nu(E')}$ is an encoding of $\nu(J)$.
  The computation is in $O(\card T + \card C)$.
\end{lemmarep}

\begin{proof}
  We process the tree decomposition $T$ of~$J$ to construct $E$
  and $T'$. We adapt the encoding construction described in
  Lemma~\ref{lem:encode}.

  Whenever we process a bag $b \in T$, the mapping precomputed with $J$ (see
  Lemma~\ref{lem:encode}) is used to obtain all facts $F$ of $I$ for
  which $b$ is the topmost node where
  domain $\dom(F) \subseteq \dom(b)$ and $\phi(F)\in\dom(b)$.

  For every such fact $F$, we create one bag $b'$ in $T'$ labeled with
  all elements of $\dom(b)$ that are gates of $G$, and one node $n$ in $E$ which is the encoding of $F$  (considering only the
  domain $\dom(b) \cap \dom(I)$) as for a normal relational instance. Set the
  selected gate $\chi(b')\defeq\phi(F)$ (which is in $\dom(b')$ by the
  condition according to which we chose to consider fact $F$).

  Because $T$ was a tree decomposition of $J$, it is immediate that the
  resulting tree $T'$ is indeed a tree decomposition of width $k$ of\/ $C$ and
  that $E$ is a tree encoding of width $k$ of $I$. By construction $T'$ and
  $E$ have same skeleton, and clearly the process is in linear time in $\card{T}
  + \card{J}$. We let $E' = (E, C', T', \chi)$.

  It remains to check the last condition. Consider a Boolean valuation $\nu$ of the
  inputs of\/ $C$. Consider the instance $\nu(J)$ and its tree
  decomposition derived from $T$. It is clear that when one computes a tree
  encoding of $\nu(J)$ following $T$, one obtains an encoding $E''$ which is exactly
  $E$ except that the facts have been removed from the nodes which used to
  encode in $E$ a fact that was removed from $\nu(J)$. Hence, $E''$ is exactly
  $\teval{\nu(E')}$. This concludes the proof.
\end{proof}

Second, we show the lemmas that will allow us to ``glue together'' the circuit
$C$ of the cc-instance, which annotates the cc-encoding, with a provenance
circuit for an automaton on the tree encoding.

\begin{definition}
Let $C = (G, W, g_0, \mu)$ and $C' = (G', W', g_0', \mu')$ be circuits such that $G \cap
G' = C'_\inp$ (we say that $C$ and $C'$ are
\deft{stitchable}). The \deft{stitching} of\/ $C$ and $C'$, denoted $C
\circ C'$, is the circuit $(G \cup G', W \cup W', g_0', \mu'')$ where $\mu''(g)$ is
defined according to $\mu$ for $g \in G$ and according to $\mu'$ otherwise. In particular,
$(C \circ C')_\inp = C_\inp$.
\end{definition}

In the following lemmas about stitching, for clarity, we distinguish valuations
$\nu$ to the inputs of a circuit and the evaluation on all circuit gates, which
we write $\nu(C)$ where $C$ is the circuit.

The fundamental property of stitching is:

\begin{lemma}
  \label{lem:stitchppty}
  For any stitchable circuits $C$ and $C'$, for any gate $g$ of\/ $C'$ and
  valuation $\nu$ of\/ $C_\inp$, letting $\nu'$ be the
  restriction of $\nu(C)$ to $C'_\inp$, we have:
  $\nu'(C')(g) = \nu(C \circ C')(g)$.
\end{lemma}

\begin{proof}
  Fix $C$, $C'$, $g$, and $\nu$. As $C$ and $C \circ C'$ share the same
  inputs, $\nu$ is a valuation for both of them. Now, first note that
  for any gate $g$ of\/ $C$, $\nu(C)(g) = \nu(C \circ C')(g)$. Hence, in
  particular, for any input gate $g$ of\/ $C'$, as it is a gate of\/ $C$ because
  $C$ and $C'$ are stitchable, we have $\nu(C \circ C')(g) = \nu(C)(g) =
  \nu'(g)$. As this equality holds for any input gate $g$ of\/ $C'$, it
  inductively holds for any gate of $C'$, which proves the result.
\end{proof}

We show that a tree decomposition for $C \circ C''$ can be
obtained from two tree decompositions $T$ and $T''$ for $C$ and $C''$ that have
same skeleton, as the \deft{sum} $T + T''$ with same skeleton where each
bag $b''$ of $T + T'$ is the union of the corresponding bags $b$ and $b'$ in $T$
and $T'$. Namely:

\begin{definition}
  Given two tree decompositions $T$ and~$T'$ with same skeleton, the \deft{sum} of
  $T$ and $T'$ (written $T + T'$) is the tree decomposition $T$ with same
  skeleton where every bag $b''$ is the union of the corresponding bags $b$ and
  $b'$ in $T$ and~$T'$.
\end{definition}

The following is immediate:

\begin{lemma}
  \label{lem:sumppty}
  Given two tree decompositions with same skeleton $T$ and $T'$ of fixed width $k$ and $k'$
  for a Boolean circuit $C$ and a Boolean circuit $C'$,
  $T + T'$ can be computed in linear time in $T$ and $T'$ and has width
  $\leq k + k' + 1$.
\end{lemma}

We now show:

\begin{lemma}
  \label{lem:canstitch}
  Let $C$ and $C'$ be stitchable circuits with tree decompositions $T$ and
  $T'$ with same skeleton (with witnessing bijection $\psi$). Assume that for
  any $g \in C'_\inp$ and bag $b$ of $T'$ with $g \in \dom(b)$, we have
  $g \in \dom(\psi^{-1}(b))$. Then $T + T'$ is a tree decomposition of\/ $C
  \circ C'$.
\end{lemma}

\begin{proof}
  We consider $I_{C \circ C'}$ and show that $T + T'$ is a tree
  decomposition of it:
  \begin{compactitem}
    \item Let $g$ be a gate of\/ $C \circ C'$.
      If $g$ is not a gate of $C \cap C'$, then its occurrences in $T +
      T'$ are only its occurrences in $T$ or in $T'$, so that they form a
      connected subtree of $T + T'$ as they did in $T$ or $T'$. If it is a
      gate of $C \cap C'$, then it is an input gate of $C'$ because $C$
      and $C'$ are stitchable, and by the hypothesis, its occurrences in $T'$
      are a subset of its occurrences in $T$, so its occurrences in $T +
      T'$ are its occurrences in $T$, and they also form a connected subtree.
    \item Let $\mathbf{g}$ be a tuple occurring in a fact of $I_{C \circ
      C'}$.
      Clearly $\mathbf{g}$ occurs either in $I_{C}$ or in $I_{C'}$, so that
      it is covered by the bag $b_{\mathbf{g}}$ that covers all elements of
      $\mathbf{g}$ in $T$ or in $T'$. \qedhere
  \end{compactitem}
\end{proof}

Last, we conclude the proof of Theorem~\ref{thm:main}:

\begin{proof}
  Let $k \in \NN$, $q$ be the GSO sentence, and let $A$ be a
  $\alphab{k}{\sigma}$-bNTA that tests $q$ for treewidth $k$ according to
  Theorem~\ref{thm:courcelle}, which we lift to a $\boolp{\alphab{k}{\sigma}}$
  in the same way as in the proof of Theorem~\ref{thm:provenance-encodings}.

  Construct in linear time in the input cc-instance $J = (I, C, \phi)$ a tree decomposition of
  $J$ of width $\leq k$, and a cc-encoding $E' = (E, C, T, \chi)$ of width $k$
  of $J$, according to Lemma~\ref{lem:uncertainte}, satisfying the conditions of
  that lemma. Now, use Theorem~\ref{thm:provenance-encodings} to compute an
  arity-two
  provenance circuit $C'$ of $A$ on $E$ and with a tree decomposition $T'$ whose
  width is constant in~$I$. We further observe from the proof of the proposition
  that $T'$ that has same skeleton as $E$ (and $T$), and that for any node $n \in
  E$, the input gate for this node is in the bag corresponding to $n$ in $T'$.
  
  We then observe that $C'$ and $C$ are stitchable circuits, and that their tree
  decompositions $T'$ and $T$ have same skeleton and satisfy the conditions of
  Lemma~\ref{lem:canstitch}. We deduce from this lemma and
  Lemma~\ref{lem:sumppty} that we can construct in linear time the stitching
  $C'' \defeq C
  \circ C'$ and a tree decomposition of it, whose width does not depend on $J$.
  We now show that $C''$ satisfies the desired property, namely, $\nu(C'')$ is
  $1$ iff $\nu(J) \models q$. For any valuation $\nu$
  of $J_\inp$, we have $\nu(C \circ C') = \nu'(C')$, by
  Lemma~\ref{lem:stitchppty}, where $\nu'$ is the valuation of $C'_\inp$
  obtained from $\nu(C)$. It is clear by definition of $\nu(J)$ that a fact $F$ is
  present in $\nu(J)$ iff $\phi(F)$ is true in $\nu(C)$. We conclude using the
  fact that $\nu'$ is a provenance circuit: $\nu'(C')$ holds iff $\{F \in I \mid
  \phi(F) \text{ true in } \nu(C)\} \models q$.
\end{proof}

\paragraph{Probability evaluation.}
We now describe the consequences of Theorem~\ref{thm:main} in terms of
probability evaluation. Here is what we want to show:

\begin{corollaryrep}{cor:prob}
  The problem of computing the probability of a fixed GSO sentence
  on bounded-treewidth pcc-instances
  can be solved in ra-linear time data complexity.
\end{corollaryrep}

To prove this corollary, we need the following definition and key result:

\begin{definition}
  Let $C = (G, W, g_0, \mu)$ be an (arity-two non-monotone) Boolean circuit and $\pfun$ be a
  \deft{probabilistic valuation} of $C$ associating each $g \in C_\inp$ to a probability
  distribution $\pfun_g$ on $\{\false, \true\}$, that is, one rational
  $v_{\false} = \pfun_g(\false)$ and one rational $v_{\true} = \pfun_g(\true)$
  such that $v_{\false} + v_{\true} = 1$.
  The \deft{probability evaluation problem} for $C$ and $\pfun$  is to compute the
  probability distribution of $g_0$ under the product distribution for the inputs
  (i.e., assuming independence), that is, $\Pr_{g_0}$ mapping $v \in \{\false,
  \true\}$ to
  \[\textstyle\sum_{\substack{\nu \in \val{C_\inp}\\
  \nu(g_0) = v}} \prod_{g \in C_\inp} \pfun_{g}(\nu(g))\]
  where $\val{C_\inp}$ denotes the set of Boolean valuations of~$C_\inp$.
\end{definition}

\begin{theoremrep}{thm:probeval}
  Given a tree decomposition $T$ of width~$k$ of an arity-two Boolean circuit $C$,
  and given a probabilistic valuation $\pfun$ of $C$,
  the probability evaluation problem for $C$ and $\pfun$ can be
  solved in time \emph{ra-linear} in
  $2^k \card{T}+\card\pfun+\card{C}$.
\end{theoremrep}

With the above theorem, we can prove Corollary~\ref{cor:prob} as follows:

\begin{proof}
Let $J=(I,C,\phi,\pfun)$ be a pcc-instance of treewidth~$k$ and $q$ a query. We use Theorem~\ref{thm:main}
to construct in linear time a Boolean circuit $C'$ of treewidth~$k'$ dependent only
on~$k$ and~$q$, with distinguished gate~$g$. We build from $C'$ a tree
decomposition of width~$k'$ in linear time.
The probability that $q$ is
true in $J$ is $\Pr_g(\true)$.
We conclude as Theorem~\ref{thm:probeval} states that this can be computed in ra-linear
time in $\card{C'}+\card{\pfun}$ for fixed $k'$.
\end{proof}

We now prove Theorem~\ref{thm:probeval}.

\begin{proof}
Fix $T=(B,\LC,\RC,\dom)$ a tree decomposition of a
Boolean circuit $C=(G,W,g_0,\mu)$ (so that for any $b \in B$, $\dom(b)$ is a set
of gates of $G$). We define $E\defeq\LC\cup\RC$ and, for $g\in
G$, $V(g)$ the value set of $g$.
For $e=(b_1,b_2)\in E$, we define
$\dom(e)\defeq\dom(b_1)\cap\dom(b_2)$, the shared elements between a bag and its
parent.
We assume an arbitrary order $<$ over $G$ and
see $\dom(b)$ as a tuple by ordering elements of $\dom(b)$
with~$<$ (this ordering taking constant time as the size of bags is bounded by a
constant). If $\dom(b)=(g_1,\dots, g_m)$, we note
$V(b)=\{\false, \true\}^m$ (and similarly, for $e \in E$, $V(e)$ is
the product over $\dom(e)$).
For every $g\in G$, let $\beta(g)\in B$ be an arbitrary bag
containing~$g$ and all gates that are inputs of $g$, that is, all gates $g'$
such that $(g', g, i) \in W$ for some $W$: such a bag exists by definition of
the tree decomposition of circuits (there is a fact in $I_C$ regrouping $g$ and
the $g'$) and we can precompute such a function in linear time by a traversal of $T$.
In particular, if $g$ is an input gate, then $\beta(g)$ is an arbitrary bag containing   just $g$.

We associate to every bag $b\in B$ (resp., every
edge $e\in E$) a \emph{potential function} $\Phi^b:V(b)\to\mathbb{Q}^+$
(resp., $\Phi^e:V(e)\to\mathbb{Q}^+$), where $\mathbb{Q}^+$ denotes the
nonnegative rational numbers, initialized to the constant~$1$
function.
We will store for each bag and each edge the full table of
values of $\Phi^e$, i.e., at most $2^k$ values, each of which
has size bounded by $\card\pfun$.

The functions $\pfun_g$ for $g\in C_\inp$ are mappings from
$V(g)$ to $\mathbb{R}^+$. For a bag $b\in B$ with $g\in\dom(b)$, we
define $\pfun_g^b$ as the function that maps every tuple $\mathbf{d} \in V(b)$
to $\pfun_g(d')$
where $d'$ is the value assigned to $g$ in $\mathbf{d}$.

For $g$ a non-input gate, let $\kappa(g)$ be the tuple formed of $g$ and all gates with a
wire to~$g$, ordered by~$<$. Let $f \defeq \mu(g)$ be the function of $g$, in
$\{\neg, \vee, \wedge, 0, 1\}$.
We see~$f$
as a subrelation $R_g$ of~$V(\kappa(t))$ (the table of values of the function,
with columns reordered by applying~$<$ on $g$), that is, a set of $(\arity{f}+1)$-tuples
which represents the graph of the function.

We update the potential function by the following steps, where the product of two        functions $f$ and $f'$
which have same domain $D$ denotes pointwise multiplication, that is,
$(f \times f')(x) = f(x) \times f'(x)$ for all $x \in D$:
\begin{enumerate}
  \item  For every $g\in C_\inp$, we set
    $\Phi^{\beta(g)}\defeq\Phi^{\beta(g)}\times\pfun_g^{\beta(g)}$.

  \item For every $g\in G\backslash C_\inp$, we set
    $\Phi^{\beta(g)}(\mathbf{d})\defeq0$ if the projection of $\mathbf{d}$ onto          $\kappa(g)$ is not
    in~$R_g$, we leave $\Phi^{\beta(g)}(t)$ unchanged otherwise.
\end{enumerate}

Note that we have now initialized the potential functions in a way which exactly
corresponds to that of~\cite{huang1996inference}, for a straightforward
interpretation of our circuit with probabilistic inputs as a special case of a
belief network where all non-root nodes are deterministic (i.e., have a
conditional distribution with values in $\{0, 1\}$).

We now apply as is the \textsc{Global Propagation} steps described in Section~5.3
of~\cite{huang1996inference}: if we choose the root of the tree decomposition as
the root cluster $\mathbf{X}$, this consists in propagating potentials from the
leaves of the tree decomposition up to the root, then from the root down to the
leaves of the tree. This process is linear in $\card{T}$ and, at every bag
of~$T$, requires a number of arithmetic operations linear in~$2^k$.

As shown in~\cite{lauritzen1988local,huang1996inference}, at the end of
the process, the desired probability distribution~$\Pr_g$ for gate~$g$ can be obtained by
marginalizing $\Phi^{\beta(g)}$:
\[
  \Pr_g(d')=\sum_{\is{\substack{\mathbf{d}\in V(\beta(g))\\d_k=d'}}}\Phi^{\beta(g)}(\mathbf{d})
\] where $k$ is the position of $g$ in $\dom(\beta(g))$.

The whole process is linear in $\card{T}\times
2^k+\card{C}+\card\pfun$
under fixed-cost arithmetic; under real-cost arithmetic, belief
propagation requires multiplying and summing linearly many times
$O(\card{T}\times
2^k)$ probability values, each of size bounded by
$\card\pfun$, which is
polynomial-time in $\card{T}$, $2^k$, $\card\pfun$.
\end{proof}

\paragraph{Consequences for pc-instances.}
We define existing the formalism of
\deft{(p)c-instances}~\cite{suciu2011probabilistic}, which is analogous to
(p)cc-instances, but annotates facts with propositional formulae rather than
circuits:

\begin{definition}[\cite{huang2009maybms,green2006models}]
  \label{def:ctable}
  A \deft{c-instance} $J$ is a relational instance where each
  tuple is labeled with a propositional formula of variables (or
  \deft{events}) from a fixed set $X$.
  For a valuation $\nu$ of $X$ mapping each variable to $\{\false, \true\}$,
  the possible world $\nu(J)$ is obtained by retaining exactly the tuples whose
  annotation evaluates to $\true$ under $\nu$; $\sems{J}$ is the
  set
  of all these possible worlds.
  Observe that different valuations may yield the same possible world.

  A \deft{pc-instance} $J = (J', \pfun)$ is defined as a c-instance
  $J'$ and a \deft{probabilistic valuation}
  $\pfun: X \to [0, 1]$ for the variables used in~$J'$.
  Like
    all probabilities in this paper, the values of $\pfun$ are
  rationals.
  The probability distribution $\sems{J}$ defined by~$J$ has universe
  $\sems{J'}$ and probability measure
  $\Pr_J(I) \defeq \sum_{\substack{\nu \mid \nu(J') = I}} \Pr_J(\nu)$ with
  the product distribution on valuations:
  \[\textstyle\Pr_J(\nu)\defeq\prod_{\substack{x \in X\\\nu(x) = \true}}
    \pfun(x) \:\prod_{\substack{x \in
  X\\\nu(x) = \false}} (1 - \pfun(x)).\]
\end{definition}

We define a notion of treewidth for them:

\begin{definition}\label{def:tw-pc-instance}
  Let $\sigma^{\o} = \sigma \cup \{\mathrm{Occ}, \mathrm{Cooc}\}$, where
  $\mathrm{Occ}$ and $\mathrm{Cooc}$ have arity two.
  From a pc-instance $J$, we define the \deft{relational encoding} $I_J$
  of $J$ as the $\sigma^{\o}$-instance where each event $e$ of $J$ is encoded to
  a fresh $a_e \in \dom(J)$, and where we add a fact
  $\mathrm{Occ}(a, a_e)$ in $I_J$ whenever $a \in \dom(J)$ is used in a fact
  annotated by a formula involving $e$, and $\mathrm{Cooc}(a_e, a_f)$
  whenever events $e$ and $f$ co-occur in the formula of some fact.

  The \deft{treewidth} $\width(J)$ of a (p)c-instance $J$ is $\width(I_J)$.
\end{definition}

This notion of treewidth, through event
\mbox{(co-)occurrences}, can be connected to
treewidth for (p)cc-instances, to ensure tractability of query
evaluation on (p)c-instances of bounded treewidth in that sense.
A technicality is that we must first rewrite annotations of the
bounded-treewidth (p)c-instance to bound their size by a constant; but we can
show:

\begin{propositionrep}{prp:c2cc2}
  For any fixed $k$, given a (p)c-instance $J$ of width $\leq k$,
  we can compute
  in linear time a (p)cc-instance~$J$ which is equivalent (has the same possible
  worlds with the same probabilities) and has treewidth depending only on $k$.
\end{propositionrep}

\begin{proof}
  We first justify that we can compute in linear time from $J$
  (p)c-instance $J'$ with the same events
  such that for any valuation $\nu$, we have $\nu(J) = \nu(J')$ (and $\Pr_J(\nu)
  = \Pr_{J'}(\nu)$), and the annotations of $J'$ have size depending only on
  $k$.

  Indeed, we observe that by our assumption that $\width(J) \leq k$, for any
  formula $F$ in an annotation, the number of distinct events occurring in $F$
  is at most $k$. Indeed, there is a $\mathrm{Cooc}$ clique between these
  events in $I_J$, so that as $\width(I_J) \leq k$ (by Lemma~1
  of~\cite{gavril1974intersection}) there must be less than $k$ of them.

  Now, we observe that any formula in $J$ can be rewritten, in linear time in
  this formula for fixed $k$, to an equivalent formula whose size depends only
  on $k$. Indeed, for every valuation of the input events, which means at most
  $2^k$ valuations by the above, we can evaluate the formula in linear time;
  then we can rewrite the formula to the disjunction of all valuations that
  satisfy it, each valuation being tested as the conjunction of the right
  events and negation of events. So this overall process produces in linear time
  an equivalent (p)c-instance $J'$ where the annotation size depends only on
  $k$. So we can assume without loss of generality that the size of the
  annotations of $J$ is bounded by a constant.

  Consider now the (p)c-instance $J$, its relational encoding $I_{J}$, and a tree
  decomposition $T$ of $I_{J}$. We build a tree decomposition~$T'$ of a
  relational encoding $I_{J}$ of a cc-instance $J' = (I, C, \phi)$ designed to be
  equivalent to $J$. Start by adding to $C$ the input gates, which correspond to
  the events of $J$.

  Now, consider each fact $F = R(\mathbf{a})$ of $J$. Let $\mathbf{e}$ be the
  set of events used in the annotation $A_F$ of $F$. Note that every pair of
  $S = \mathbf{a} \sqcup \mathbf{e}$ co-occurs in some fact of $I_J$: the elements of
  $\mathbf{a}$ co-occur within $F$, the elements of $\mathbf{e}$ co-occur in a
  $\mathrm{Cooc}$ fact, and any pair of elements from $\mathbf{a}$ and
  $\mathbf{e}$ co-occur in some $\mathrm{Occ}$ fact. Hence, by Lemma~1
  of~\cite{gavril1974intersection},
  % quoted more legibly in Bodlaender and Koster, "Treewidth computations I.
  % Upper bounds", Lemma 2
  there is a bag $b_F \in T$ such that $S \subseteq \dom(b)$.

  Let $C_F$ be a circuit representation of the Boolean function $A_F$ on $E$,
  whose size depends only on~$k$. Add $C_F$ to $C$, add $F$ to $I$, and
  set $\phi(F)$ to
  be the distinguished node of\/ $C_F$. We have thus built $J'$, which by
  construction is equivalent to $J$.

  We now build $T'$ by making it a copy of $T$. Now, for each fact $F$,
  considering its bag $b_F$, and $b_F'$ the corresponding bag in~$T'$, we add
  all elements of\/ $C_F$ to $b_F'$. This decomposition clearly covers all facts
  of $I_{J'}$, and event occurrences form subtrees because they do in $T$ and the
  elements that we added to $T'$ are always in a single bag only. Last, it is
  clear that the bag size depends only on $k$, as the size of the $C_F$ added to
  the bags depends only on~$k$, and at most $k$ of them are added to each bag (because
  there are at most $k$~elements per bag).

  We have not talked about probabilities, but clearly if $J$ is a pc-instance
  the probabilities of the inputs of the pcc-instance $J'$ should be defined
  analogously.
\end{proof}

We can now combine the above with Theorem~\ref{thm:main}, and deduce the
tractability of query evaluation on bounded-treewidth pc-instances.

\begin{theoremrep}{thm:pc}
  For bounded-treewidth pc-instances, the probability query evaluation problem
  for Boolean MSO queries can be solved in ra-linear time data complexity.
\end{theoremrep}

\begin{proof}
  The result is an immediate consequence of Proposition~\ref{prp:c2cc2}
  and Theorem~\ref{thm:main} as long as we show that,
  for any fixed $k \in \mathbb{N}^*$, and for every (p)c-instance $J$ of width $\leq
  k$, one can compute in linear time a (p)c-instance~$J'$ with the same events
  such that for any valuation $\nu$, we have $\nu(J) = \nu(J')$ (and $\Pr_J(\nu)
  = \Pr_{J'}(\nu)$), and the annotations of~$J'$ have size depending only
  on~$k$.

  Fix $k$ and $J$. We observe that by our assumption that $\width(J) \leq k$, for any
  formula $\Phi$ in an annotation, the number $p_\Phi$ of distinct events
  occurring in $\Phi$
  is at most $k$. Indeed, there is a $\mathrm{Cooc}$ clique between these
  events in $I_J$ and each of them is connected by the $\mathrm{Occ}$
  relation to domain elements of the fact~$F$ annotated by~$\Phi$ (there is at least one), so we have
  in total a $(p_\Phi+1)$-clique. By Lemma~1 of~\cite{gavril1974intersection},
  any tree decomposition must have one node containing all these $p_\Phi+1$
  elements, and therefore $p_\Phi\leq k$.

  Now, we observe that any formula in $J$ can be rewritten, in linear time in
  this formula for fixed $k$, to an equivalent formula whose size depends only
  on $k$. Indeed, for every valuation of the input events, which means at most
  $2^k$ valuations by the above, we can evaluate the formula in linear time;
  then we can rewrite the formula to the disjunction of all valuations that
  satisfy it, each valuation being tested as a conjunction of at most $k$
  literals.
  So this overall process produces in linear time
  an equivalent (p)c-instance where the annotation size depends only on $k$.
\end{proof}

\end{toappendix}

\myparagraph{}{Probabilistic XML}
We start with the problem of probabilistic query evaluation, beginning with the
setting of \emph{trees}. We use the framework of \deft{probabilistic
XML}, denoted $\prxml^\fie$, to represent
probabilistic trees as trees annotated by propositional formulas over
independent probabilistic events 
(see~\cite{kimelfeld2013probabilistic} for the formal
definitions), and consider the \emph{data complexity} of the \deft{query
evaluation} problem for a MSO query~$q$ on such trees (i.e., computing the
probability that $q$ holds).

This problem is intractable in
general, which is not surprising: it is harder than
determining the probability of a single propositional
annotation. However, for the less expressive \emph{local}
$\prxml$ model, $\prxml^{\muxind}$, query evaluation has tractable data
complexity~\cite{cohen2009running}; this model restricts edges to
be annotated by only one event literal that is only used on that edge (plus a
form of mutual exclusivity).

\begin{toappendix}
    We will first prove the result on scopes (Proposition~\ref{prp:scopes}) and
  then prove the result on local models (Theorem~\ref{thm:muxind}).

  \paragraph{XML and instances.}
We first describe XML documents and their connections to relational models.

\begin{definition}
  \label{def:xml}
  An \deft{XML document} with \deft{label set} $\Lambda$ (or
  \deft{$\Lambda$-document}) is an \emph{unranked}
  $\Lambda$-tree.
\end{definition}

  \begin{definition}
    A $\prxml$-tree $T$ is an \emph{unranked} $\Lambda$-tree, for a fixed
    alphabet $\Lambda$ of \deft{labels}, augmented with a set of Boolean events
    $\calE$ where each event $e_x$ has a probability $0 \leq p_x \leq 1$, and
    where each edge of the tree is labeled by a propositional formula over
    $\calE$.
    
    We see $T$ as defining a probability distribution over
    $\Lambda$-trees in the following fashion: for every valuation $\nu$ over
    $\calE$, the possible world $\nu(T)$ is obtained by removing all edges whose
    annotation evaluates to false under $\nu$, and all their descendent nodes
    and edges. The probability of a $\Lambda$-tree $T'$ according to~$T$ is the
    sum of the probability of all valuations $\nu$ such that $\nu(T) = T'$,
    where the probability of a valuation is defined assuming that the events in
    $\calE$ are drawn independently with their indicated probability.
  \end{definition}

  \begin{definition}
To perform \deft{query evaluation} on a $\prxml$ document is to
determine, for a fixed query over $\Lambda$-trees, given an input $\prxml$
document $T$, what is the total probability of its possible worlds that satisfy
$q$; we study its \deft{data complexity}, i.e., its complexity as a function of~$T$.
  \end{definition}

We always assume that the label set $\Lambda$ is fixed (not provided as
input).
As XML documents are unranked, it is often more convenient to
manipulate their binary left-child-right-sibling representation:

\begin{definition}
  The \deft{left-child-right-sibling} (LCRS) representation of an unranked rooted ordered
  $\Lambda$-tree $T$ is the following $\Lambda$-tree $T'$:
  a node $n$ whose children are the ordered sequence of siblings $n_1, \ldots,
  n_k$ is encoded as the node $n$ with $\LC(n) = n_1$, $\RC(n_1) = n_2$, ...,
  $\RC(n_{k-1}) = n_k$; we complete by nodes labeled $\bot \notin \Lambda$ to
  make the tree full.
\end{definition}

We now define how XML documents can be encoded to the relational setting.
\begin{definition}
  \label{def:xml2rel}
  Given a $\Lambda$-document $D$, let $\sigma_{\Lambda}$
  be the relational signature with two binary predicates $\FC$ and $\NS$ (for
  ``first child'' and ``next sibling''), and unary predicates $P_{\lambda}$ for
  every $\lambda \in \Lambda$. The \deft{relational encoding} $I_D$ of $D$ is
  the-$\sigma_{\Lambda}$ instance with $\dom(I_D) = \dom(D)$, such that:
  \begin{compactitem}
    \item for any consecutive siblings $(n, n')$,
      $\NS(n, n')$ holds;
    \item for every pair $(n, n')$ of a node $n \in D$ and its first child $n'
      \in D$ following sibling order, $\FC(n, n')$ holds;
    \item for every node $n \in D$, the fact $P_{\lbl{n}}(n)$ holds.
  \end{compactitem}
\end{definition}

\begin{lemmarep}{lem:relenclin}
  The relational encoding $I_D$ of an XML document $D$ has treewidth~$1$
  and can be computed in linear time.
\end{lemmarep}

\begin{proof}
  Immediate: the relational encoding is clearly computable in linear time
  and there is a
  width-1 tree decomposition of the relational encoding that has same
  skeleton as the LCRS representation of the XML document.
\end{proof}

Importantly, the language of MSO queries
on XML documents~\cite{neven2002query},
which we now define formally, can be easily translated to
queries on the relational encoding:

\begin{definition}
  An \deft{MSO query} on XML documents is a MSO formula where first-order
  variables refer to nodes and where
  atoms are $\lambda(x)$ ($x$ has label $\lambda$), $x
  \rightarrow y$ ($x$ is the parent of $y$), and $x < y$ ($x$ and $y$ are
  siblings and $x$ comes before $y$).
\end{definition}

\begin{lemmarep}{lem:xml2relq}
  For any MSO query $q$ on $\Lambda$-documents, one
  can compute in linear time an MSO query $q'$ on $\sigma_{\Lambda}$ such that
  for any $\Lambda$-document $D$, $D \models q$ iff $I_D
  \models q'$.
\end{lemmarep}

\begin{proof}
  We add a constant overhead to~$q$ by defining the predicates
  $\lambda(x)$ for $\lambda\in\Lambda$ as $P_\lambda(x)$,
  the predicate $x < y$ to be the transitive closure of $\NS$
  ($\neg (x = y) \wedge \forall S (x \in S \wedge (\forall z z' (z \in S \wedge
  \NS(z, z')) \Rightarrow z' \in S) \Rightarrow y \in S)$), and the predicate $x
  \rightarrow y$ to be $\exists z, \FC(x, z) \wedge (z = y \vee z < y)$. It is
  clear that the semantics of those atoms on $I_D$ match that of the
  corresponding atoms on $D$, so that a straightforward structural induction on
  the formula shows that $q'$ satisfies the desired properties.
\end{proof}

\begin{definition}
  Given label set $\Lambda$, we say that an XML document $D$ on $\Lambda \sqcup
  \{\bot,\det\}$ is a \deft{sparse representation} of an XML document $D'$ on
  $\Lambda$ if the root is labeled with an element of $\Lambda$, and the
  XML document obtained from $D$ by removing every $\bot$ node and their
  descendants, and replacing every $\det$ node
  by the collection of its children, in order, is exactly $D'$.

  We say that a $\sigma_{\Lambda\sqcup\{\det\}}$ instance $I$ is a \deft{weak relational
  encoding} of an XML document $D$ with label set~$\Lambda$ if there exists a
  sparse representation $D'$ of $D$ such that $I$ is the relational encoding of
  $D'$ except that $P_{\bot}$ facts are not written.
\end{definition}

\begin{propositionrep}{prp:rwweak}
  For any MSO query $q$ on XML documents with (fixed) label set $\Lambda$, one
  can compute in linear time an MSO query $q'$ on $\sigma_{\Lambda}$
  such that for any XML document $D$ on label set $\Lambda$, if $D \models q$
  then $I \models q'$ for any weak relational encoding $I$ of~$D$; and conversely
  if $D \not\models q$ then $I \not\models q'$ for any weak relational encoding
  $I$ of $D$.
\end{propositionrep}

\begin{proof}
  We show that,
  for any MSO query $q$ on XML documents with (fixed) label set $\Lambda$, one
  can compute in linear time an MSO query $q'$ on documents with label in
  $\Lambda \sqcup \{\bot,\det\}$ such that for any XML document $D$ on label set
  $\Lambda$, if $D \models q$ then $D' \models q'$ for any sparse representation
  $D'$ of $D$; and conversely if $D \not\models q$ then $D' \not\models
  q'$ for any sparse representation $D'$ of $D$. The result then
  follows by Lemma~\ref{lem:xml2relq}.

  We call \emph{regular} the nodes with label in $\Lambda$. Consider a
  document $D$ and sparse representation $D'$ of $D$ with a mapping $f$ from
  $D$ to~$D'$ witnessing that $D'$ is a sparse representation of $D$. Let us
  consider a node $n \in D$ with children $n_1, \ldots, n_k$ in order, and
  determine what is the relationship between $f(n)$ and the $f(n_i)$ in $D'$.

  It is straightforward to observe that $f(n)$ is regular and the $f(n_i)$ are
  topmost regular descendants of $f(n)$ in $D'$; and for $i < j$, there is some
  node $n'$ in $D'$ (intuitively, their lowest common ancestor, which is a
  descendant of $f(n)$, possibly $f(n)$ itself) such that $n'$ is both an
  ancestor of $f(n_i)$ and $f(n_j)$, $n'$ is a descendant of $f(n)$, and $n'$
  has two children $n_1'$ and $n_2'$ such that $f(n_i)$ is a descendant of
  $n_1'$ (maybe they are equal), $f(n_j)$ is a descendant of $n_2'$ (maybe they
  are equal), and $n_1' < n_2'$ in $D'$. Note that $n'$, $n_1'$ and $n_2'$ are
  not necessarily regular nodes of~$D$ but can be $\det$ nodes. In
  addition, no $\bot$ node can be traversed in any of the
  ancestor--descendant chains discussed in this paragraph.

  It is now clear that we can have MSO predicates $\rightarrow'$ and $<'$ in $D'$
  following these informal definitions (and not depending on $D$ or $D'$),
  defined from predicates $\rightarrow$, $<$ and $\lambda(\cdot)$ on $D'$, such
  that for every $D$ and sparse encoding $D'$ of $D$, for every nodes $n, n' \in
  D$, we have $n \rightarrow n'$ in $D$ iff $f(n) \rightarrow f(n')$ in $D'$
  (which should only hold between regular nodes, so nodes in the image of~$f$),
  and likewise for $<$. Last, it is clear that the predicates $\lambda(\cdot)$
  of $D$ can be encoded directly to the same predicates in $D'$.
\end{proof}

\paragraph{Probabilistic XML.}
We formally introduce probabilistic XML. We start by $\prxml^{\fie}$, i.e.,
$\prxml$ with events.

\begin{definition}
  \label{def:prxml}
  A $\prxml^{\fie}$ probabilistic XML document $D = (D', \pfun)$ is a
  $(\Lambda \sqcup \{\fie\})$-document $D'$
  where edges from $\fie$ nodes to
  their children are labeled with a propositional formula over some set of
  Boolean events $X$, and a probabilistic valuation~$\pfun$ mapping each $e \in
  X$ used in $D$ to an independent probability $\pfun(e) \in [0, 1]$ of being true.

  The semantics $\sems{D}$ of
  $D$ is obtained by extending $\pfun$ to a probability distribution on
  valuations $\nu$ of $X$ as usual, and defining $\nu(D)$ for
  $\nu$ to be $D'$ where all $\fie$ nodes are replaced by the
  collection of their children with edge annotation $\Phi$ such that $\nu(\Phi)
  = \true$ (the others, and their descendants, are
  discarded).
  We require the root to have label in~$\Lambda$.
\end{definition}

We will prove Proposition~\ref{prp:scopes} via
an encoding of $\prxml^\fie$ to pc-instances:

\begin{definition}
  \label{def:xmlfie2rel}
  The \deft{pc-encoding} of a $\prxml^{\fie}$ document $D = (D',
  \pfun)$ in $\Lambda\sqcup\{\fie\}$ is the pc-instance $J_D = (J'_D,
  \pfun')$ with same events, $\pfun' = \pfun$,
      and where the c-instance $J'_D$ is the relational encoding of
      $D'$
      with the
      following annotations. $\NS$- and~$\FC$-facts are annotated with
      $\true$.
      $P_\lambda(n)$-facts are annotated with the annotation
      $\Phi$ of the edge from the parent of~$n$ to~$n$, if $\Phi$ exists,
      with $\true$ otherwise.
\end{definition}

\begin{propositionrep}{prp:relencw}
  For any MSO query $q$ on $\Lambda$-documents, one
  can compute in linear time an MSO query $q'$ on $\sigma_{\Lambda}$ such that
  for any $\prxml^{\fie}$ XML document $D$, for any
  valuation $\nu$ of $D$, letting $\nu'$ be the corresponding valuation of
  $J_D$, we have that $\nu(D) \models q$ iff $\nu'(J_D) \models q'$.
\end{propositionrep}

\begin{proof}
  We prove that for any valuation $\nu$ of $D$, letting $\nu'$ be the corresponding valuation
  of $J_D$, we have that $\nu'(J_D)$ is a weak encoding of $\nu(D)$ (we
see $P_\fie$ facts in $\nu'(J_D)$ as if they were $P_{\det}$ facts). The
  result then follows by
  Proposition~\ref{prp:rwweak}.

  We first show that for any valuation $\nu$ of $D$ and corresponding valuation
  $\nu'$ of $J_D$, for every $\lambda \in \Lambda$, $n$ is a node of
  $\nu'(J_D)$ that is retained in the XML document $\nu'(J_D)$ is a
  sparse representation of
  iff $n$ is a node which is
  retained in $\nu(D)$, with same labels.
  Indeed, for the forward implication, observe that any
  fact $P_{\lambda}(n)$ is created for node $n$ with label $\lambda$ in
  $n$, and it is retained if and only if all its regular ancestors are
  retained and the annotation of its parent edge in $\nu(D)$ evaluates to
  $\true$;
  conversely, if $n$ has label $\lambda$ in $D$ then a fact $P_{\lambda}(n)$
  was created in $I$ and if $n$ is retained in $\nu(D)$ then all the
  conditions on edges in the chain from $n$ to the root
  evaluate to~$\true$ so $P_{\lambda}(n)$ does hold and $n$ is retained
  in $\nu'(J_D)$.

  We further know that by construction relations $\FC$ and $\NS$ correspond to the
  first-child and next-sibling relations in $D$ no matter the valuation.

  So we deduce that $J_D$ is the relational encoding of the XML document
  obtained from $D$ by replacing
  all nodes not kept in $\nu(D)$ by $\bot$ nodes, and removing all edge
  annotations.
\end{proof}

Observe that in this definition of pc-encoding, it is \emph{not} the
case that the possible worlds of $J_D$ are the relational encodings of the
possible worlds of $D$. For instance, the $\fie$ nodes are retained as is, and
$\FC$- and $\NS$-facts are always retained even if the corresponding nodes are
dropped. The following example shows that it would not be reasonable to ensure
such a strong property:
\begin{example}
  \label{exa:badclique}
  Consider an $\fie$ node with $k$ children $n_1, \ldots, n_k$, all
  annotated with independent events with probability $1/2$. In a straightforward attempt to encode this  node and its
  descendants to a pc-instance $J$ (or even to a pcc-instance~$J$), we would create one domain element   $e_i$ for
  each of the $n_i$. But then we would need to account for the fact that, as any
  pair $n_i, n_j$ may be retained individually, the fact $\NS(e_i, e_j)$ would
  need to occur in a possible world of $J$, and thus would also occur in $J$. So
  this na\"ive attempt to ensure that the possible worlds of $J$ are exactly the
  relational encodings of the possible worlds of $D$ leads to a pcc-instance of
  quadratic size and linear treewidth.
\end{example}

\paragraph{Tractability for $\prxml^{\fie}$.}
Of course we cannot hope that the pc-encoding of a $\prxml^{\fie}$
document always has constant treewidth for it is known that for
$\prxml^{\fie}$, evaluating MSO queries is almost always
$\#P$-hard~(\cite{kimelfeld2008query}, Theorem~5.2).
A first notion of tractability for a $\prxml^{\fie}$ document~$D$ is the
treewidth (following Definition~\ref{def:tw-pc-instance})
of the pc-encoding of $D$. Indeed, 
Proposition~\ref{prp:relencw} and Theorem~\ref{thm:pc} imply the following:
\begin{corollary}
  \label{cor:corstruct}
  For $\prxml^{\fie}$ documents with bounded-treewidth pc-encoding, the
  MSO probabilistic query evaluation problem
  can be solved in ra-linear time data complexity.
\end{corollary}

The condition on event scopes is a simpler sufficient condition for
tractability. We give its formal definition:

\begin{definition}
  Consider a $\prxml^{\fie}$ document $D$ with event set $X$ and its
  LCRS representation $D'$. We say that an event $e \in X$
  \deft{occurs} in a node $n$ of $D'$ if $e$ occurs in the annotation of the edge
  from the parent of $n$ to~$n$. For every $e \in X$, let $D'_e$ be the smallest
  connected subtree of $D'$ that covers all nodes where $e$ occurs. The
  \deft{event scope} $S(n)$ of a node $n \in D'$ is $\{e \in X \mid n \in
  D'_e\}$.
  The \deft{event scope width} of $D$ is $w_{\s}(D) \defeq \max_{n \in D} \card{S(n)}$.
\end{definition}

We are now ready to prove the result on XML element scopes:

\begin{propositionrep}{prp:scopecorr}
  For any $\prxml^{\fie}$ document $D$, we have $\width(J_D) \leq w_{\s}(D) +
  1$.
\end{propositionrep}

\begin{proof}
  We show how to build a tree decomposition of the relational encoding of
  $J_D$ from the event scopes.
  Consider the tree decomposition $T$ of $I_D$ that is isomorphic to a
  LCRS encoding $D'$ of $D$: the root node of $D'$ is coded to an empty bag, and
  each node $n$ of the LCRS encoding with parent $n'$ is coded to $\{n', n\}$.

  We now add to $T$, for each bag $b$ corresponding to a node $n$, the events of
  $S(n)$. It is clear that $T$ is of the prescribed width and that the
  occurrences of all nodes and events are connected subtrees.

  We now argue that it is a tree decomposition of the relational encoding
  of~$J_D$, but this is easily
  seen: it covers all $\NS$- and $\FC$- facts represented in $J_D$, and covers all
  occurrences and co-occurrences by construction of the scopes.
\end{proof}

This implies Proposition~\ref{prp:scopes} because of
Corollary~\ref{cor:corstruct}.

\paragraph{Tractability of $\prxml^{\muxind}$.}

We now introduce the definitions and proofs for the local model,
$\prxml^{\muxind}$.

\begin{definition}
  A $\prxml^{\muxind}$ probabilistic document is an XML document $D$ over
  $\Lambda \sqcup \{\ind, \mux\}$, where edges from $\ind$ and $\mux$ nodes to
  their children are labeled with a probability in $[0, 1]$, the
  annotations of outgoing edges of every $\mux$ node summing to $\leq 1$.
  
  The
  semantics $\sems{D}$ of $D$ is obtained as follows: for every $\ind$ node,
  decide to keep or discard each child according to the indicated probability,
  and replace the node by the (possibly empty) collection of its kept children;
  for every $\mux$ node, choose one child node to keep according to the
  indicated probabilities (possibly keep no node if they sum to $< 1$), and
  replace the $\mux$ node by the chosen child (or remove it if no child was
  chosen). All probabilistic choices are performed independently.
  When a node is not kept, its descendants are also
  discarded. We require the root to have label in~$\Lambda$.
\end{definition}

Observe that in $\prxml^{\muxind}$ all probabilistic choices are ``local'', in a
similar fashion to the tuple-independent (TID) and BID probabilistic relational formalisms. As we
show later, this helps ensure the tractability of query evaluation.

We use Corollary~\ref{cor:corstruct} to show the tractability
of query evaluation on the $\prxml^{\muxind}$ local model, which was already
proven in~\cite{cohen2009running}.
We first rewrite input documents to a simpler form:

\begin{definition}
  Two $\prxml^{\muxind}$ documents $D_1$ and $D_2$ are equivalent if for
  every XML document $D$, $\Pr_{D_1}(D) = \Pr_{D_2}(D)$.
\end{definition}

\begin{definition}
  We say that a $\prxml^{\muxind}$ is in \deft{binary form} if it is a full
  binary tree, and the sum of the outgoing probabilities of every $\mux$ node is
  equal to $1$.
\end{definition}

The following definition is needed to ensure linear time execution for technical
reasons:

\begin{definition}
  A $\prxml^{\muxind}$ document is \deft{normalized} if for every $\mux$ nodes,
  the rational probabilities that annotate its child nodes all share the same
  denominator.
\end{definition}

\begin{lemmarep}{lem:canbinary}
  From any normalized $\prxml^{\muxind}$ document $D$, we can compute in linear time in
  $D$ an equivalent $\prxml^{\muxind}$ document $D'$ which is in binary form.
\end{lemmarep}

\begin{proof}
  In this proof, for brevity, we use $\det$ nodes to refer to $\ind$ nodes whose
  child edges are all annotated with probability~$1$.

  First, rewrite $\mux$ nodes whose outgoing probabilities sum up to $<1$ by
  adding a $\det$ child for them with the remaining probability. This operation
  is in linear time because the corresponding number has same denominator as
  other children of the $\mux$ node (as the document is normalized), and the
  numerator is smaller than the denominator.

  Next, use $\det$ nodes to rewrite the children of regular and $\ind$ nodes to
  a chain so that all regular and $\ind$ nodes have at most $2$ children. This
  only causes a constant-factor blowup of the document.

  Next, rewrite $\mux$ nodes with more than two children to a hierarchy of
  $\mux$ nodes in the obvious way: considering a $\mux$ node $n$ with $k$ children
  $n_1, \ldots, n_k$ and probabilities $p_1, \ldots, p_k$ summing to $1$, we
  replace $n$ by a hierarchy $n_1', \ldots, n_{k-1}'$ of $\mux$ nodes: the
  children of each $n_i'$ is $n_i$ with probability $\frac{p_i}{\sum_{j < i}
  p_j}$ and $n_{i+1}'$ with probability $1 - \frac{p_i}{\sum_{j < i} p_j}$;
  except for $n_{k-1}'$ whose children are $n_{k-1}$ and $n_k$ (with the same
  probabilities). This operation can be performed in linear time as the
  denominators of the fractions simplify (by the assumption that the document is
  normalized), and the sum operations work on operands and results which are
  smaller than the numerator.

  Now, replace $\mux$ nodes with $<2$ children by $\ind$ nodes (the
  probabilities are unchanged).

  Last, add $\det$ children to nodes so that the degree of every node is either
  $2$ or $0$.

  This process can be performed in linear time and that the
  resulting document is in binary form; equivalence has been maintained through
  all steps.
\end{proof}

Now, we can show:

\begin{propositionrep}{prp:muxindscope}
  For any $\prxml^{\muxind}$ document $D$ in binary form, one can compute in
  linear time an equivalent $\prxml^{\fie}$ document whose scopes have size
  $\leq 1$.
\end{propositionrep}

\begin{proof}
  For every $\ind$ node $n$ with two children $n_1$ and $n_2$ with probabilities
  $p_1$ and $p_2$, introduce two
  fresh events $e_n^{\ind,1}$ and $e_n^{\ind,2}$ with probabilities $p_1$ and $p_2$, and replace
  $n$ by a $\fie$ node so that its first and second outgoing edges are annotated
  with $e_n^{\ind,1}$ and $e_n^{\ind,2}$.

  Likewise, for every $\mux$ node $n$ with two children $n_1$ and $n_2$ with
  probabilities $p$ and $1-p$, introduce a fresh event $e_n^{\mux}$ with probability
  $p$ and replace $n$ by a $\fie$ node so that its first and second outgoing
  edges are annotated with $e_n^{\mux}$ and $\neg e_n^{\mux}$.

  It is immediate that the resulting document $D'$ is equivalent to $D$. Now,
  consider the scope of any node of this document. Only one event occurs in this
  node, and the only events that occur more than one time in the document occur
  exactly twice, on the edges of two direct sibling nodes, so they never
  occur in the scope of any other node. Hence all scopes in $D$ have size $\leq
  1$.
\end{proof}

From this, given that $\prxml^{\muxind}$ document can be normalized in
ra-linear time, we deduce the tractability of MSO query evaluation on
$\prxml^{\muxind}$, as claimed in the main text:
\medskip

\end{toappendix}

We can use the provenance circuits of Section~\ref{sec:encodings} to justify that
query evaluation is tractable for $\prxml^{\muxind}$ and capture the data
complexity tractability result of~\cite{cohen2009running}. We say that an
algorithm runs in \deft{ra-linear time} if it runs in linear time assuming that
arithmetic operations over rational numbers take constant time and rationals are
stored in constant space, and runs in polynomial time without this assumption.
We can show:

\begin{theoremrep}[\cite{cohen2009running}]{thm:muxind}
  MSO query evaluation on $\prxml^{\muxind}$ has ra-linear data complexity.
\end{theoremrep}

We can also show extensions of this result. For instance, on $\prxml^{\fie}$,
defining the \deft{scope} of event $e$ in a document $D$ as the smallest
subtree in the left-child-right-sibling encoding of $D$ covering
nodes whose parent edge mentions $e$, and the \deft{scope size} of a
node $n$ as the number of events with $n$ in their scope, we show:

\begin{proposition}\label{prp:scopes}
  For any fixed $k \in \NN$, MSO query evaluation on $\prxml^{\fie}$ documents with scopes
  assumed to have size $\leq k$ has ra-linear data complexity.
\end{proposition}

\myparagraph{}{BID instances}
\begin{toappendix}
  Following~\cite{barbara1992management,re2007materialized}, we define:
\begin{definition}
  A \deft{BID instance} $I$ is a relational instance with each relation
  partitioned into \deft{key} and \deft{value} positions. For each
  valuation of the key positions, all matching facts (that form a \deft{block})
  are mutually exclusive, each has a probability $>0$ and the
  probabilities of the block sum to $\leq 1$. The semantics is to keep,
  independently between blocks,
  one fact at random in each block, according to the indicated
  probabilities (or possibly no fact if probabilities sum to $<1$).
\end{definition}

To ensure ra-linear time complexity, we assume that
BID instances are given with facts regrouped per blocks; otherwise our bounds
are PTIME as we first need to sort
the facts.

\begin{definition}
  We define the \deft{treewidth} $\width(I)$ of a BID instance $I$ as that of the
  underlying relational instance, forgetting about the probabilities.
\end{definition}

We show the tractability of MSO query evaluation on
BID through Theorem~\ref{thm:main} and Corollary~\ref{cor:prob}, using the
following result:
\begin{lemmarep}{lem:bidwidths}
  For any fixed $k \in \mathbb{N}^*$, given a BID instance $J$ with  $\width(J) \leq
  k$,
  we can compute in
  ra-linear time
  an equivalent pcc-instance $J'$ where $\width(J')$ depends only on~$k$.
\end{lemmarep}

However, the proof of this result is non-trivial. By an encoding to pc-instances, it is
straightforward to show the result if we assume that the size of each
block is bounded by a constant. But otherwise, we need to build
a decision circuit for which value to pick for each key; we do so
in a tree-like fashion following a decomposition of the BID
instance.

\begin{proof}
  Fix $k$ and $J$. First, compute in linear time a tree decomposition $T$ of $J$ of
  width $\width(J) \leq k$.

  Without loss of generality, we can assume that probabilities within
each block of $J$ are rationals with the same denominator (if this is not
the case, we normalize these probabilities in ra-linear time).

  As in the proof of Lemma~\ref{lem:encode}, we can assume that every fact of
  $J$ has been assigned to a bag of $T$ where it is covered (i.e., $F =
  R(\mathbf{a})$ with $\mathbf{a} \subseteq \dom(b)$ for $b$ the covering bag).
  Actually, still in the spirit of the proof of Lemma~\ref{lem:encode}, we can modify the decomposition  $T$
  by copying nodes to create chains, so that we can assume that at most one fact
  is assigned to each bag. This preprocessing can be performed in linear time.
  For every fact $F$ of $J$ we let $\beta(F)$ be the bag of $T$ to which fact $F$
  was assigned.

  We compute the pcc-instance $J' = (J, C, \phi)$ by building $C$ and $\phi$
  and a tree decomposition $T'$ for $J'$ with same skeleton as $T$, which is
  initialized as a copy of $T$. We add the gates of\/ $C$ to $T'$ to turn it
  into a tree decomposition of $J'$.

  Let $\mathcal{B}$ be the set of blocks: a key $\mathbf{a} \in \calB$ is a pair of a relation symbol    and a tuple
  that is a key in $J$ for that relation. We write $J_{\mathbf{a}}$ to refer to
  $J$ restricted to the facts of block $\mathbf{a}$; and
  $\card{I_{\mathbf{a}}}$
  is the size (\emph{not} the number of facts!) of this part of the instance
  (the size of both the facts and the associated probabilities). It is then
  clear that $\sum_{\mathbf{a} \in \calB} \card{I_{\mathbf{a}}} = \card{J}$, the
  size of the original instance.

  Now, for every $\mathbf{a} \in \calB$, consider the subset of bags
  $T_{\mathbf{a}}$ of $T$ that cover $\mathbf{a}$; it is a connected subtree, as
  it is the intersection for every element $a \in \mathbf{a}$ of the occurrence
  subtree $T_k$ of this element, which are connected subtrees, and it is not
  empty because the elements of $\mathbf{a}$ must occur together in some fact of $J$ so they
  also do in some bag of $T$. What is more, we can precompute in linear time the
  roots of all the $T_{\mathbf{a}}$ (by the same precomputation as in the proof
  of Lemma~\ref{lem:encode}). It is also clear that $\sum_{\mathbf{a} \in
  \calB}
  \card{T_{\mathbf{a}}}$ is of size linear in $\card{J}$, as, for fixed $\sigma$
  and $k$, each bag of $T$ can only occur in a constant number of
  $T_{\mathbf{a}}$.

  So we prove the result in the following way: for each $\mathbf{a}
  \in \calB$,
  we compute in time $O(\card{I_{\mathbf{a}}} + \card{T_{\mathbf{a}}})$ a
  circuit $C_{\mathbf{a}}$ to annotate the facts of $I_{\mathbf{a}}$ in $J'$, and
  we add the gates of\/ $C_{\mathbf{a}}$ to $T'$ to obtain a tree
  decomposition of $J'$ so far, making sure that we add only a constant number of
  gates to each bag, and only to bags that are in $T_{\mathbf{a}}$. If we can
  manage this for every $\mathbf{a} \in \calB$, then the result follows, as we can
  process the blocks in $J$ in order (as they are provided); our final
  pcc-instance has width that is still constant (for each bag of $T$ can
  only occur in a constant number of $T_{\mathbf{a}}$); and by the arguments
  about the sizes of the sums, the overall running time of the algorithm is
  linear in $J$.

  So in what follows we fix $\mathbf{a} \in \calB$ and describe the construction of
  $C_{\mathbf{a}}$ and the associated decomposition.

  Using our preprocessed table to find the root of $T_{\mathbf{a}}$, we can
  label its nodes by going over it top-down, in time linear in $T_{\mathbf{a}}$.
  We now notice that for every fact $F = R(\mathbf{a}, \mathbf{v})$ of
  $I_{\mathbf{a}}$, the bag $\beta(F)$ covers $F$ so it must be in
  $T_{\mathbf{a}}$. We write $\beta_{\mathbf{a}}$ for the restriction of the
  function $\beta$ to the facts of $I_{\mathbf{a}}$.

  We now say that a bag $b \in T_{\mathbf{a}}$ is an \deft{interesting bag}
  either if it is in the image of $f_{\mathbf{a}}$ or if it is a lowest common ancestor
  of some subset of bags that are in the image of $f_{\mathbf{a}}$. We now observe that
  the number of interesting bags of $T_{\mathbf{a}}$ is linear in the number of
  facts of $I_{\mathbf{a}}$; indeed, the interesting bags form the internal
  nodes and leaves of a binary tree whose leaves must all be in the image of
  $\beta_{\mathbf{a}}$, so the number of leaves is at most the number of facts of
  $I_{\mathbf{a}}$, so the total number of nodes in the tree is linear in the
  number of leaves.

  We now define a weight function $\weight$ on $T_{\mathbf{a}}$ by $\weight(b) = \pfun(F)$
  (the probability of $F$) for $F\in I_{\mathbf{a}}$ and
  $\beta_{\mathbf{a}}(F)=b$, if any such $F$ exists; $\weight(b) = 0$ otherwise. We define
  bottom-up a cumulative weight function $\weight'$ on $T_{\mathbf{a}}$ as $\weight(b)$,
  plus $\weight'(\LC(b))$ if $\LC(b) \in T_{\mathbf{a}}$, plus $\weight'(\RC(b))$ if $\RC(b)
  \in T_{\mathbf{a}}$. For notational convenience we also extend $\weight'$ to
  anything by saying that $\weight'(b) = 0$ if $b \notin T_{\mathbf{a}}$ or $b$ does
  not exist.

  Observe now that for a non-interesting bag $b$, $\weight(b)$ and $\weight'(b)$ can be
  represented either as $0$ or as a pointer to some $\weight(b')$ or $\weight'(b')$ for an
  interesting bag $b'$. Indeed, if $b$ is non-interesting then we must have
  $\weight(b) =
  0$. Now we show that if $b$ has a topmost interesting descendant $b'$ then it
  is unique: indeed, the lowest common ancestor of two interesting descendants
  of $b$ is a descendant of $b$ and it is also interesting, so there is a unique
  topmost one. Now this means that either $b'$ does not exist and $\weight'(b) = 0$,
  or it does exist and all descendants of $b$ that are in the image of
  $\beta_{\mathbf{a}}$ are descendants of $b'$, so that $\weight'(b) = \weight'(b')$ and we can
  just make $\weight'(b)$ a pointer to $\weight'(b')$.

  Now this justifies that we can compute $\weight$ and $\weight'$ bottom-up in linear time
  in $\card{T_{\mathbf{a}}} + \card{I_{\mathbf{a}}}$: observe that we are
  working on rationals with the same denominator, so the sums that we perform
  are sums of integers, whose size  always remains less than the common
  denominator; as there is a number of interesting bags which is linear in the
  number of facts of $I_{\mathbf{a}}$, and those are the only nodes for which a
  value (whose size is that of the probabilities in $I_{\mathbf{a}}$) actually
  needs to be computed and written, the computation is performed in time
  $O(\card{T_{\mathbf{a}}} + \card{I_{\mathbf{a}}})$ overall.

  We now justify that we can encode $T_{\mathbf{a}}$ to a circuit with the
  correct probabilities. For each bag $b \in T_{\mathbf{a}}$, we create a
  gate~$g_b^{\i}$; for the root bag $b$ it is an input gate with
  probability $\weight'(b)$; for other bags it is a gate whose value is defined
  by the parent bag. Intuitively, $g_b^{\i}$ describes whether to choose a fact
  from $F_{\mathbf{a}}$ within the subtree rooted at~$b$.

  For every interesting bag $b$, writing $\weight'(b) = k'/d$ and $\weight(b) = k/d$ with
  $d$ the common denominator, create one input gate $g_b^{\h}$ with probability
  $\frac1{\weight'(b)} \weight(b) = k/k'$, and one gate $g_b^{\h\wedge}$ which is the AND
  of $g_b^{\h}$ and $g_b^{\i}$. Intuitively, this gate describes whether to
  generate the fact assigned at this node, if any. If there is such a fact,
  set its image by $\phi$ to be $g_b^{\h\wedge}$. Now if $\weight'(b) > \weight(b)$
  (intuitively: there is still the possibility to generate fact at child nodes),
  we create one input gate $g_b^{\leftrightarrow}$ which has probability
  $\frac{1}{\weight'(b)-\weight(b)} \weight'(\LC(b))$. Once again, this probability simplifies to a
rational whose numerator and denominator are $< d$. We create a gate
$g_b^{\LCf}$ to be $g_b^{\i} \wedge \neg g_b^{\h} \wedge g_b^{\leftrightarrow}$
(creating a constant number of intermediate gates as necessary), and
$g_b^{\RCf}$ to be $g_b^{\i} \wedge \neg g_b^{\h} \wedge \neg
g_b^{\leftrightarrow}$, setting them to be $g_{\LC(b)}^{\i}$ and
$g_{\RC(b)}^{\i}$ where applicable (i.e., if $\LC(b)$ and $\RC(b)$ exist and are
in~$T_{\mathbf{a}}$).

By contrast, non-interesting bags $b$ just set $g_{\LC(b)}^{\i}$ and
$g_{\RC(b)}^{\i}$ (where applicable) to be $g_b^{\i}$, with no input gates.

  We now observe that by construction the resulting circuit has a tree
  decomposition that is compatible with $T$, so that we can add its events to
  $T'$ and only add constant width to the nodes of $T_{\mathbf{a}}$ as required.
  It is also easy to see that the circuit gives the correct distribution on the
  facts of $F_{\mathbf{a}}$, with the following invariant: for any bag $b \in
  T_{\mathbf{a}}$, the probability that $g^{\i}_b$ is $\true$ is $\weight'(b)$, and
  $g^{\h\wedge}_b$, $g^{\LCf}_b$ and $g^{\RCf}_b$ are either all $\false$ if
  $g^{\i}_b$ is $\false$ or, if $g^{\i}_b$ is $\true$, exactly one is true and
  they respectively have marginal probabilities $\weight(b)$,
  $\weight'(\LC(b))$, and $\weight'(\RC(b))$.
  Now the circuit construction is once
  again in time $O(\card{I_{\mathbf{a}}} + \card{T_{\mathbf{a}}})$, noting that
  interesting nodes are the only nodes where numbers need to be computed and
  written; and we have performed the entire computation in time
  $O(\card{I_{\mathbf{a}}} + \card{T_{\mathbf{a}}})$, so the overall result is
  proven.
\end{proof}

Combining Lemma~\ref{lem:bidwidths} and
Theorem~\ref{thm:main}, we can conclude:

\end{toappendix}

We move from trees to relational instances, and show another bounded-width
tractability result for \deft{block-independent disjoint} (BID)
instances (see~\cite{suciu2011probabilistic}, or~\cite{barbara1992management,re2007materialized}
for formal
definitions).
We define the \deft{treewidth} of a BID instance as
that of its underlying relational instance, and claim the following (remember
that query
evaluation on a probabilistic instance means determining the probability that
the query holds):

\begin{theoremrep}{thm:bid}
  For any fixed $k \in \NN$, MSO query evaluation on an input BID instance of treewidth $\leq
  k$ has ra-linear data complexity.
\end{theoremrep}

This implies the same claim for tuple-independent
databases~\cite{dalvi2007efficient,lakshmanan1997probview}.

All 
probabilistic results are proven by rewriting to a formalism of relational
instances with a circuit annotation, such that instance and circuit have a
bounded-width joint decomposition. We compute a treelike provenance circuit for the
instance using Theorem~\ref{thm:provenance-encodings}, combine it with the
annotation circuit, and
apply existing message passing techniques~\cite{lauritzen1988local} to compute the probability of the circuit.

\myparagraph{}{Counting}
We turn to the problem of counting query results, and reduce it in
ra-linear time to query evaluation on treelike instances, capturing a
result of~\cite{arnborg1991easy}:

\begin{theoremrep}[\cite{arnborg1991easy}]{thm:counting}
  For any fixed MSO query $q(\mathbf{x})$ with free first-order
  variables and $k\in\NN$,
  the number of matching assignments to $\mathbf{x}$ on an
  input instance $I$ of width $\leq k$ can be computed in
  ra-linear data complexity.
\end{theoremrep}

\begin{proof}
  Let $k \in \NN$.
  Let $q(\mathbf{x})$ be the MSO query. We rewrite it to the following query: $q' : \exists
  \mathbf{x} \bigwedge_{x \in \mathbf{x}}P_x(x) \wedge q(\mathbf{x})$, where the
  $P_x$ are fresh unary predicates.
  Consider an input instance $I$ of width $\leq k$, and expand it to a BID
  instance~$I'$ by setting existing relations to be trivial BID tables (i.e., all
  attributes of the relation are keys, and all facts have probability~$1$) and
  adding tables $P_x$ for every $x \in \mathbf{x}$ with one attribute, with the
  empty set as key, and with facts $P_x(a)$ for all $a \in \dom(I)$, with
  probability $1/\card{\dom(I)}$. This rewriting can clearly be performed in
  ra-linear time, and if $I$ has treewidth $\leq k$ then so does $I'$.
  Intuitively, the possible worlds of $I'$ are all the possible ways of
  extending $I$ with facts $P_x(a)$ for $x \in \mathbf{x}$ and $a \in \dom(I)$,
  with only one fact $P_x(a)$ for every $x \in \mathbf{x}$, and each possible
  world has probability $1/\card{\dom(I)}^{\card{\mathbf{x}}}$.

  We now make the immediate observation that for every such possible world
  $I'_{\mathbf{a}}$ of~$I'$
  indexed by the $a_x \in \dom(I)$ for $x \in \mathbf{x}$, where we add the
  facts $P_x(a_x)$ for $x \in \mathbf{x}$, we have $I'_{\mathbf{a}} \models q'$
  iff $I \models q(\mathbf{a})$. Hence, the number of matches of $q$ in $I$ is
  the number of possible worlds of $I'$ where $q'$ holds, that is, the
  probability of $q'$ on $I'$ multiplied by $M \defeq \card{\dom(I)}^{\card{\mathbf{x}}}$.

  We conclude by Theorem~\ref{thm:bid} that we can compute this probability in
  ra-linear time in $I'$, that is, in~$I$, and we compute the count from the
  probability by multiplying by $M$ in ra-linear time, proving the result.
\end{proof}

\section{Related Work}\label{sec:related}
\paragraph*{Bounded treewidth.}
From the original results~\cite{Courcelle90,flum2002query} on
the linear-time data complexity of MSO evaluation on treelike structures,
works such as~\cite{arnborg1991easy} have investigated counting problems,
including applications to probability computation (on graphs).
A recent paper~\cite{bodlaender2012probabilistic} also shows
the linear-time data complexity of evaluating an MSO query on a
treelike probabilistic network (analogous to a circuit). Such
works, however, do not decouple the computation of a treelike
\emph{provenance} of the query and the \emph{application} of
probabilistic inference on this provenance, as we do.
We also note results from another
approach~\cite{pichler2010counting} on treelike structures,
based on monadic Datalog (and not on MSO as the other works), that are
limited to counting.

\paragraph*{Probabilistic databases.}
The \emph{intensional} approach~\cite{suciu2011probabilistic} to query
evaluation on probabilistic databases is to compute a lineage of the query and evaluate its
probability via general purpose methods;
tree-like lineages allow for tractable
probabilistic query evaluation
\cite{jha2012tractability}. Many works in this field provide sufficient
conditions for lineage tractability, only a few based on the
data~\cite{sen2010read,roy2011faster} but most based on the
query~\cite{dalvi2007efficient,jha2012tractability}.
For treelike
instances, as we show, we can \emph{always} compute treelike lineages, and we can do so for expressive queries (beyond UCQs considered in these
works), or alternatively generalize Boolean lineages to connect them to more expressive semirings.

\paragraph*{Provenance.}
Our provenance study is inspired by the usual 
definitions of semiring provenance for the relational algebra and Datalog~\cite{green2007provenance}. 
Another notion of provenance, for XQuery queries on
trees, has been introduced in~\cite{foster2008annotated}.
Both~\cite{green2007provenance} and~\cite{foster2008annotated} provide
\emph{operational} definitions of provenance, which cannot be directly
connected to tree automata.
A different relevant work on
provenance is~\cite{deutch2014circuits}, which introduces
provenance circuits, but uses them for
Datalog and only on \emph{absorptive} semirings. 
Last, other works study 
provenance for \emph{transducers}~\cite{bojanczyk2014transducers}, but with no
clear connections to semiring provenance or provenance for Boolean queries.

\section{Conclusion}\label{sec:conclusion}
We have shown that two provenance constructions can be computed in linear time
on trees and treelike instances: one for UCQs on arbitrary semirings, the other
for arbitrary GSO queries as non-monotone Boolean expressions.
A drawback of our results is their high combined complexity, as they rely on 
non-elementary encoding of the query to an automaton. One approach to fix this
is monadic Datalog~\cite{gottlob2010monadic,pichler2010counting}; this
requires defining and computing provenance in this setting.

\paragraph{Acknowledgements.}
This work was partly supported by a financial
contribution from the Fondation Campus Paris-Saclay and the French ANR Aggreg project.

\putbib
\end{bibunit}

\end{document}